\documentclass[a4paper]{amsart}

\usepackage{amsmath,amsfonts,amssymb,amsthm,graphicx,cite,mathtools}
\usepackage{latexsym,graphicx}
\usepackage{tikz}
\usepackage{dsfont}
\usepackage{amscd}
\usepackage{mathrsfs}
\usepackage[pdfstartview=FitV]{hyperref}

\newtheorem{theorem}{Theorem}[section]

\newtheorem{corollary}{Corollary}[section]
\newtheorem{lemma}{Lemma}[section]
\newtheorem{proposition}{Proposition}[section]
\theoremstyle{remark}
\newtheorem{remark}{Remark}[section]

\numberwithin{equation}{section}

\newcommand{\e}{\mathrm{e}}
\newcommand{\imag}{\mathrm{i}}
\newcommand{\gtr}{>}
\newcommand{\less}{<}
\newcommand{\abs}[1]{\lvert #1 \rvert}

\newcommand{\Z}{{\mathbb Z}}
\newcommand{\C}{{\mathbb C}}
\newcommand{\R}{{\mathbb R}}
\newcommand{\cA}{{\mathcal A}}
\newcommand{\cE}{{\mathcal E}}
\newcommand{\cD}{{\mathcal{D}}}

\newcommand{\cU}{{\mathcal{U}}}
\newcommand{\cV}{{\mathcal{V}}}

\newcommand{\half}{\frac{1}{2}}

\newcommand{\ve}{\varepsilon}
\renewcommand{\r}{\mathfrak{r}}

 \newcommand{\tPhi}{\Phi^{\vee}}

\newcommand\varpm{\mathbin{\vcenter{\hbox{%
  \oalign{\hfil$\scriptstyle+$\hfil\cr
          \noalign{\kern-.3ex}
          $\scriptscriptstyle({-})$\cr}%
}}}}

\begin{document}

\title[Eigenfunctions of van Diejen's model]{Eigenfunctions of the van Diejen model generated by gauge and integral transformations}

\author{Farrokh Atai}
\address{School of Mathematics, University of Leeds, Leeds, LS2 9JT, UK}
\email{F.Atai@leeds.ac.uk} 

\author{Masatoshi Noumi}
\address{Department of Mathematics, Rikkyo University, Toshima-Ku, Tokyo 171-8501, Japan}
\email{mnoumi@rikkyo.ac.jp}

\date{\today}

\begin{abstract}
We present how explicit eigenfunctions of the principal Hamiltonian for the $BC_{m}$ relativistic Calogero-Moser-Sutherland model, due to van Diejen, can be constructed using gauge and integral transformations. In particular, we find that certain $BC$-type elliptic hypergeometric integrals, including elliptic Selberg integrals, of both Type I and Type II arise as eigenfunctions of the van Diejen model under some parameter restrictions. Among these are also joint eigenfunctions of so-called modular pairs of van Diejen operators. Furthermore, these transformations are related to reflections of the $E_{8}$ Weyl group acting on the space of model parameters. 
\end{abstract}

\maketitle
\setcounter{tocdepth}{1}
\tableofcontents

\section{Introduction and preliminaries}

\subsection{Introduction} 
\label{sec_intro}
In this paper we present several gauge and integral transformations for the eigenfunctions of the van Diejen model and construct various classes of explicit solutions to the eigenvalue equation for the principal van Diejen Hamiltonian. The van Diejen model \cite{vDi94,KH97} is defined by a family of mutually commuting analytic difference operators, depending on 8 parameters $a = ( a_{0} , \ldots, a_{7} )\in(\C^{\ast})^{8}$ ($\C^{\ast}=\C\setminus\{0\}$) determining the `external' interaction and the coupling parameter $t\in\C^{\ast}$ determining the `pairwise' interaction, as well as the shift parameter $q \in \C^{\ast} $ and the elliptic nom\'e $p \in \C$ satisfying $\abs{p} < 1$. The principal Hamiltonian is given by the analytic difference operator
\begin{equation}
\label{eq_van_diejen_operator}
\cD_{x}{(a \lvert p , q,t)} =  A^{0}(x ; a \lvert p, q , t) +  \sum_{ 1 \leq i \leq m} A_{i}^{+}(x ; a \lvert p, q,t ) T_{q , x_{i}} + A_{i}^{-}(x ; a \lvert p, q , t ) T_{q, x_{i}}^{-1},
\end{equation}
where $x = (x_1 , x_2, \ldots , x_m)$ ($m\in\Z_{\gtr 0}$) are the canonical coordinates on the algebraic torus $(\C^{\ast})^{m}$ and $T_{q , x_i }$ ($i\in\{1,2,\ldots,m\}$) denotes the multiplicative $q$-shift operator\footnote{That is, $T_{q , x_i}$ acts on suitable functions $f(x)$ by shifting $x_i \to q x_i$ while leaving remaining variables unaffected.}, with coefficients 

\begin{equation}
\label{eq_coefficients_positive_multiplicative}
A_{i}^{\ve}(x ; a \lvert p,q,t ) = \frac{\prod_{0 \leq s \leq 7} \theta(a_s x_{i}^{\ve};p) }{ \theta(x_{i}^{\ve 2} , q x_{i}^{\ve 2} ; p )} \prod_{j \neq i } \frac{ \theta( t x_{i}^{\ve} x_{j}^{\pm};p) }{\theta(x_{i}^{\ve} x_{j}^{\pm};p)} \quad (i\in\{1,2,\ldots,m\}; \ \ve \in\{\pm\}),
\end{equation}
where $\theta(x;p)$ denotes the multiplicative theta function (see Section~\ref{sec_prel_and_notation}); see Section~\ref{sec_vD_symmetries_and_properties}, Eq.~\eqref{eq_vD_external_potential_1}-\eqref{eq_vD_external_potential_2}, for the definition of the zeroth order coefficient $A^{0}(x ; a \lvert p,q,t)$. (We also refer the reader to Section~\ref{sec_prel_and_notation} for details on our notation.) Throughout the paper, the principal Hamiltonian is simply referred to as \emph{the van Diejen operator} as a thorough consideration of the family of commuting operators is outside the scope of this paper.

Constructing the exact eigenfunctions of the van Diejen operator has been an ongoing open problem and little is still known for non-trivial parameters (non-trivial in these sense that the coefficients do not reduce to constants) except in the univariate case \cite{Spi03,Spi04,Kom04,Cha07,Spi07,Spi08} which we collect in Section~\ref{sec_previous_results}. In particular, it is known that certain eigenfunctions in the univariate case are given by the \emph{elliptic hypergeometric integrals} $I(b_{0},\ldots,b_{5}, c x , c x^{-1};p,q)$ ($c \in \C^{\ast}$) where $b_{s} = b_{s}(a \lvert p,q)$ ($s\in\{0,1,\ldots,5\}$) and 
\begin{equation}
I(b_{0},\ldots,b_{7}; p , q) = \int_{\mathbb{T}^{1}} \frac{dy}{2\pi\imag y} \ \frac{\prod_{0 \leq s \leq 7} \Gamma( b_s y^{\pm} ; p , q)}{\Gamma( y^{\pm 2};p,q) },
\end{equation}
with $\Gamma(y ;p ,q)$ denoting the elliptic Gamma function \cite{Rui97} and $\mathbb{T}^{n} = \mathbb{T}^{n}_{1}$ ($n\in\Z_{>0}$) the $n$-dimensional torus, \emph{i.e.} 
\begin{equation}
\mathbb{T}^{n}_{\rho} = \bigl\{ y = (y_{1}, \ldots , y_{n}) \in (\C^{\ast})^{n} \lvert \ \abs{y_k} = \rho \quad (k \in \{1,2,\ldots, n \} )\bigr\} \quad ( \rho \in \mathbb{R}_{> 0}  ),
\label{eq_torus}
\end{equation}
under two restrictions on the model parameters.
(In this Section, the cycles for the integrations are set to the torus for simplicity and their analytic continuation are discussed in Section~\ref{sec_domain_of_holomorphy}.) The van Diejen model is known to be closely related to the Type II $BC_{m}$ elliptic hypergeometric integrals, of which the integral above corresponds to the $m=1$ case: The van Diejen operator is known to be formally self-adjoint (or symmetric) with respect to the symmetric $\C$-bilinear form \cite{vDi94,Rui04}
\begin{equation}
\langle f , g \rangle = \int_{\mathbb{T}^{m}}d\omega_{m}(x) \  w_{m}(x ; a \lvert p,q,  t) \ f(x) g(x) , \quad d\omega_{m}(x) = \frac{1}{(2 \pi \imag)^{m}} \frac{dx_1 \cdots dx_m }{x_1 \cdots x_m},
\label{eq_inner_product}
\end{equation}
with \emph{weight function} $w_{m}(x ; a \lvert p,q , t )$ given by 
\begin{equation}
w_{m}(x ; a \lvert p , q , t ) = \prod_{1 \leq i \leq m} \frac{ \prod_{0 \leq s \leq 7} \Gamma(a_{s} x_{i}^{\pm} ; p , q)}{\Gamma( x_{i}^{\pm 2}; p , q)} \prod_{1 \leq i \less j \leq m} \frac{\Gamma( t x_{i}^{\pm} x_{j}^{\pm} ; p , q)}{\Gamma( x_{i}^{\pm} x_{j}^{\pm};p , q)}.
\label{eq_weight_function}
\end{equation}
Setting $f = g = 1$ in the $\C$-bilinear form yields the Type II elliptic $BC_{m}$ hypergeometric integrals of Selberg type and it is clear that the hypergeometric integral $I(a_{0},\ldots,a_{7};p, q)$ corresponds to the $m=1$ case.{\medskip}

To our knowledge, the only known eigenfunction in the literature for $m\geq 2$ and non-trivial parameters is the constant eigenfunction when the parameters satisfy the ellipticity condition \eqref{eq_ellipticity_condition_1}; see also Section~\ref{sec_special_cases_eigenfunction}. (The ellipticity condition constrains the space of model parameters to a level set of an $E_8$ root, which we elaborate further upon below.) The purpose of this paper is to extend the list of known eigenfunctions by applying gauge and integral transformations (\emph{cf.} Theorems~\ref{thm_Cauchy_eigenfunction_transform_Type_II}-\ref{thm_dual-Cauchy_eigenfunction_transform_Type_I}). These transformations follow from the symmetries of the operator's coefficients which provide two types of gauge transformations for the van Diejen model, and the kernel functions \cite{Rui09a,KNS09} which play crucial roles in constructing integral transformations \cite{KNS09,Rui15,AL20}. It is not by chance that the exact eigenfunctions are given by elliptic hypergeometric integrals: The combination of the $\C$-bilinear form, gauge transformations, and the kernel functions yield that eigenfunctions of the van Diejen operator are given by multivariate elliptic hypergeometric integrals of type $BC$. 
Although our main results are summarized in Section~\ref{sec_main_results}, here we present two illuminating examples of typical eigenfunctions that are constructed by gauge and integral transformations:
\begin{theorem}
\label{thm_Selberg_type_II}
Let $K \subseteq \{0,1,\ldots,7\}$ with even cardinality, i.e. $\abs{K}\in2\Z_{\geq 0}$, and the parameters satisfy $\abs{p}< \min(\abs{q t}, \abs{q^{-1} t^{n-m-1}})$, $\abs{q},\abs{t} < 1$, and $\abs{a_{s}} > \abs{(p q t)^{\half}}$ for all $s\notin K$. The Type II $BC_{n}$ elliptic hypergeometric integrals
\begin{multline}
\label{eq_Selberg_eigenfunction_type_II}
\psi_{K}(x ; a \lvert p, q , t) = \int_{\mathbb{T}^{n}} d\omega_{n}(y) \prod_{ 1 \leq k \leq n} \frac{\bigl(\prod_{s\notin K } \Gamma( b_{s} y_{k}^{\pm} ; p ,q )\bigr) \bigl(\prod_{1 \leq i \leq m} \Gamma( c x_{i}^{\pm} y_{k}^{\pm} ; p ,q )\bigr)}{\Gamma( y_{k}^{\pm2 } ; p , q) } \\ 
\cdot \prod_{1 \leq k < l \leq n} \frac{\Gamma( t y_{k}^{\pm} y_{l}^{\pm} ; p ,q )}{\Gamma( y_{k}^{\pm} y_{l}^{\pm} ; p , q)},
\end{multline}
where $c = p^{\half} q^{\half} t^{-\half}$ and $b_{s}=p^{\half} q^{\half} t^{\half}a_{s}^{-1}$ for all $s\notin K$, are eigenfunctions of the van Diejen operator $\cD_{x}(a \lvert p,  q, t)$ \emph{(}with eigenvalues $\Lambda_{0}( b_{K(pq)} \lvert q, t)$ in \eqref{eq_vD_constant_eigenvalues}\emph{)} for $x\in\mathbb{T}^{m}$ if the parameters satisfy 
\begin{equation}
\prod_{s\in K} a_{s} = \ve t^{-m + \half \abs{K}} , \quad \prod_{s\notin K} a_{s} = \ve p^{2} q^{2} t^{2n-m-\half\abs{K} + 2} \quad (\ve \in \{\pm\}).
\end{equation}
\end{theorem}
\noindent{}(Proof of Theorem~\ref{thm_Selberg_type_II} is given in Section~\ref{sec_proofs}.)
\begin{theorem}
\label{thm_Selberg_type_I}
Let $K \subseteq \{0,1,\ldots,7\}$ with even cardinality, i.e. $\abs{K}\in2\Z_{\geq 0}$, and the parameters satisfy $\abs{p} < \min(\abs{q t} , \abs{q^{-1} t^{n-m-1}})$, $\abs{q} < 1$ and $\abs{a_{s}} > \abs{(p q t)^{\half}}$ for all $s\notin K$. The Type I $BC_{n}$ elliptic hypergeometric integrals
\begin{multline}
\label{eq_Selberg_eigenfunction_type_I}
\tilde{\psi}_{K}(x ; a \lvert p, q , t) = \int_{\mathbb{T}^{n}} d\omega_{n}(y) \prod_{1 \leq k \leq n} \frac{\bigl(\prod_{s\notin K } \Gamma( b_{s} y_{k}^{\pm} ; p ,q )\bigr) \bigl(\prod_{1 \leq i \leq m} \Gamma( c x_{i}^{\pm} y_{k}^{\pm} ; p ,q )\bigr)}{\Gamma( y_{k}^{\pm2 } ; p , q) } \\ 
\cdot \prod_{1 \leq k < l \leq n} \frac{ 1 }{\Gamma( y_{k}^{\pm} y_{l}^{\pm} ; p , q)},
\end{multline}
where $c = p^{\half} q^{\half} t^{-\half}$ and $b_{s} = p^{\half} q^{\half} t^{\half} a_{s}^{-1}$ for all $s \notin K$, are eigenfunctions of the van Diejen operator $\cD_{x}(a \lvert p,  q, t)$ \emph{(}with eigenvalues $\Lambda_{0}( b_{K(pq)} \lvert q, p q t^{-1})$ in \eqref{eq_vD_constant_eigenvalues}\emph{)} for $x\in\mathbb{T}^{m}$ if the parameters satisfy 
\begin{equation}
\prod_{s\in K} a_{s} = \ve (pq)^{-n+1} t^{2 n - m + \half \abs{K}-2} , \quad \prod_{s\notin K} a_{s} = \ve (pq)^{n+1} t^{-m-\half\abs{K} + 4} \quad (\ve \in \{\pm\}).
\end{equation}
\end{theorem}
\noindent{}(Proof of Theorem~\ref{thm_Selberg_type_I} is given in Section~\ref{sec_proofs}.) \newline
We note that the restriction $\abs{p} < \abs{q t} $ in the statements above ensures that the eigenvalue equation makes sense in the domain of holomorphy for the functions \eqref{eq_Selberg_eigenfunction_type_II} and \eqref{eq_Selberg_eigenfunction_type_I} . The integrand itself is holomorphic in the domain $\abs{c} < \abs{ x_{i} } < \abs{c}^{-1}$ ($i\in\{1,\ldots,m\}$) as far as $\abs{c} = \abs{ p^{\half} q^{\half} t^{-\half} }< 1 $.

Setting $m=n=1$ in the Theorems above yields the elliptic beta/hypergeometric integrals (for different cardinality of $K$) and allows us to provide an alternative proof that the elliptic hypergeometric integrals $I(b_{0},\ldots,b_{5},c x , c x^{-1};p,q)$ are exact eigenfunctions of the van Diejen operator. Furthermore, the integrals \eqref{eq_Selberg_eigenfunction_type_II} and \eqref{eq_Selberg_eigenfunction_type_I} have known evaluation formulas \cite{vDS01,Rai10,NI19} for particular values of $m$, $n$, and $\abs{K}$. These elliptic Selberg integrals then yield (globally meromorphic) exact eigenfunctions given in terms of particular products of the elliptic Gamma function. We collect some of these known evaluations in Appendix~\ref{app_elliptic_SF} for the convenience of the reader.
{\medskip}

Finally, we wish to recall that the van Diejen operator $\cD_{x}(a\lvert p, q, t)$ has a fascinating Weyl group symmetry in the model parameters $a\in(\C^{\ast})^{8}$, in the sense that its eigenvalues are invariant under the reflections of the Weyl group associated to the Lie algebra $D_8$ \cite{Rui09a}; see also Section~\ref{sec_E8_symmetry}. In the one-variable case, it is also proven that the eigenvalues are invariant under the actions of the Weyl group associated with the exceptional Lie algebra $E_8$ \cite{Rui15}. This symmetry can be extended to the multivariate case: The gauge transformations correspond to the $D_8$ Weyl group $W(D_{8})$ acting on the parameter space. Combining the gauge transformations with the Cauchy-type integral transformation (for $n=m$ and certain constraints on $t$) allows us to obtain an additional reflection on the parameter space related to the $E_8$ Weyl group $W(E_{8})$; see Section~\ref{sec_E8_symmetry} for details. Although we suppose that the $E_8$-symmetry have been known to experts in this field, this has (to our knowledge) not been shown in the literature previously.
{\smallskip}

Before proceeding to give our main results in Section~\ref{sec_main_results}, let us take a moment to introduce our notation and outline the plan of the paper.


\subsection{Notation and preliminaries}
\label{sec_prel_and_notation}
We use the standard notation $\Z$, resp. $\C$, for the set of integers, resp. complex numbers and use $\imag = + \sqrt{-1}$. The hyperoctahedral group (Weyl group of type $BC_{m}$) is denoted by $W_{m} = \{\pm 1\}^{m} \rtimes \mathfrak{S}_{m}$ and acts naturally on $x=(x_1 , x_2, \ldots , x_m)$ through permutations $x_i \leftrightarrow x_j$ and inversions $x_i \to x_{i}^{-1}$ ($i,j\in\{1,2,\ldots,m\}$, $i \neq j$).

Throughout the paper we have that $n,m\in\Z_{\geq1}$ are positive integers. For any $a=(a_1,\ldots,a_m) \in (\C^{\ast})^{m}$ we use the notation $c a =(c a _1, \ldots, c a_m)$ for any $c\in\C$ and $a^{-1} = (1/a_1,\ldots,1/a_m)$. Given a constant $c\in\C^\ast$ and two vector $a , b\in\C^{m}$, we write $c^{a} b$ for the vector $( c^{a_{1}} b_1 , \ldots , c^{a_{m}} b_{m}) \in \C^{m}$. For any constant $c$ and set $\mathbb{S}$, we define
$
c^{\mathbb{S}} = \bigl\{ c^{k} \  \lvert \  k \in \mathbb{S} \bigr\}
$
as a shorthand notation.\newline
We always assume that the parameters $q$ and $p$ satisfy $\abs{p}<1$ and $\abs{q}<1$, and use the canonical notation for the shifted factorials, \emph{i.e.} $(a ; p)_{\infty} = \prod_{\ell = 0}^{\infty} ( 1 - a p^{\ell})$ and $(a; p , q )_\infty = \prod_{\ell=0}^{\infty} (p^{\ell} a ; q)_{\infty}$.
The usual multiplicative theta function and elliptic Gamma function are denoted by $\theta(x ; p)$ and $\Gamma(x ; p , q)$ respectively, and are given by $\theta(x;p) = ( x ;p)_{\infty}(p x^{-1};p)_{\infty}$, resp. $\Gamma(x ; p ,q ) = ( p q x^{-1} ; p , q)_{\infty} (x ; p , q )^{-1}_{\infty}$. The properties of the multiplicative theta and elliptic Gamma functions, that are needed for this paper, are summarized in Appendix~\ref{sec_theta_and_gamma_functions}.
It will also be convenient to introduce the shorthand notation 
\begin{equation}
\Gamma( x_1, x_2, \ldots, x_r ;p,q) = \Gamma(x_1 ; p , q) \Gamma( x_2 ; p , q) \cdots \Gamma(x_r ; p , q) \quad (r\in\Z_{> 0})
\end{equation}
as well as
\begin{equation}
\Gamma(c x^{\pm};p,q) = \Gamma(c x;p,q) \Gamma(c x^{-1} ; p , q ) \quad (c\in \C^{\ast}),
\end{equation}
\begin{equation}
\Gamma(c x^{\pm} y^{\pm} ; p , q ) = \Gamma( c x y^{\pm} ; p , q ) \Gamma (c x^{-1} y^{\pm} ; p , q ) \quad ( c \in \C^{\ast}),
\end{equation}
and so forth, and similar for the multiplicative theta function. 
{\medskip}

From here and onwards, we omit the dependence on the elliptic nom\'e $p$ in the operators, eigenfunctions, and gauge functions, \emph{e.g.} write $\cD_{x}{(a \lvert q, t )} = \cD_{x}(a\lvert p , q , t)$, when this does not lead to ambiguity.

\subsection{Plan of the paper}
 In the next Section, Section~\ref{sec_main_results}, we present the gauge and integral transformations for the eigenfunctions of the van Diejen operator. The previously known eigenfunctions, that these transformations can be applied to, are presented in Section~\ref{sec_previous_results} and we explore further properties of the van Diejen operator in the parameter space in Section~\ref{sec_shift_properties}. In  Section~\ref{sec_vD_symmetries_and_properties}, the symmetries of the van Diejen model are given which allow us to provide straightforward proofs of Lemmas~\ref{lemma_gauge_transform_V} and \ref{lemma_gauge_transform_U}, and Propositions~\ref{prop_gauge_transf_V} and \ref{prop_gauge_transf_U}. In order for the paper to be self-contained, we also present a short proof that the van Diejen operator is (formally) self-adjoint, or symmetric, with respect to the symmetric $\C$-bilinear form \eqref{eq_inner_product}. The proofs for Theorems~\ref{thm_Cauchy_eigenfunction_transform_Type_II}-\ref{thm_dual-Cauchy_eigenfunction_transform_Type_I} and Theorems~\ref{thm_Selberg_type_II} and \ref{thm_Selberg_type_I} are given in Section~\ref{sec_proofs}. In Section~\ref{sec_domain_of_holomorphy}, we focus on the analytic continuation of the integral transformations presented in Section~\ref{sec_main_results} as well as the domain and range of the integral transformations. A brief introduction to the $E_{8}$ root lattice and the relation between our transformations and the $E_{8}$ Weyl group $W(E_{8})$ are given in Section~\ref{sec_E8_symmetry}. Section~\ref{sec_conclusions_outlook} contains some final remarks and an outlook. Appendix~\ref{sec_theta_and_gamma_functions} presents the properties of the multiplicative theta functions and elliptic Gamma function needed for our paper. Appendix~\ref{app_elliptic_SF} presents the elliptic hypergeometric series and its relation to the elliptic hypergeometric integrals and concludes with known evaluations of (multivariate) elliptic hypergeometric integrals, \emph{i.e.} elliptic Selberg integrals of type $BC$. The various possible integral transformations that we can obtain are presented in Appendix~\ref{app_transforms} for the convenience of the reader. The paper concludes by presenting the conventions and relations to \cite{KNS09} and \cite{Rui09a} in Appendix~\ref{app_KFI}.

\section{Main results}
\label{sec_main_results}

As mentioned in Section~\ref{sec_intro}, we construct several different transformations for the eigenfunctions of the van Diejen operator. These gauge and integral transformations, as well as how they can be used to construct exact eigenfunctions, are presented in this Section:
\begin{lemma}\label{lemma_gauge_transform_V} 
Define the gauge function $V(x \lvert q , t)$ as 
\begin{equation}
V( x \lvert q ,t ) = \prod_{1 \leq i \less j \leq m} \frac{1}{ \Gamma( t x_{i}^{\pm} x_{j}^{\pm} ; p , q)}.
\label{eq_gauge_function_V}
\end{equation}
The van Diejen operator $\cD_{x}{(a \lvert q, t )}$ satisfies the relation
\begin{equation}
V(x \lvert q , t )^{-1} \circ \cD_{x}{(a \lvert q , t)} \circ V(x \lvert q , t ) = \cD_{x}{(a \lvert q , p q  t^{-1})}.
\label{eq_gauge_transform_V}
\end{equation}
\end{lemma}
\noindent(Proof of Lemma~\ref{lemma_gauge_transform_V} is given in Section~\ref{eq_symmetries_proof}.)

\begin{lemma}\label{lemma_gauge_transform_U}
For any $K\subseteq\{0,1,\ldots, 7\}$ with even cardinality, \emph{i.e.} $\abs{K} \in 2 \Z_{\geq 0}$, define the gauge function $U_{K}(x ; a \lvert q )$ as
\begin{equation}
U_{K}(x ; a \lvert q ) =  \prod_{1 \leq i \leq m} \prod_{s\in K} \frac{1}{\Gamma( a_{s} x_{i}^{\pm};p , q)}
\label{eq_gauge_function_U},
\end{equation}
and the notation $a_{K(pq)} = ( (a_{K(pq)})_{0} , \ldots , (a_{K(pq)})_{7} ) \in (\C^{\ast})^{8}$ as
\begin{equation}
(a_{K(pq)})_{s}  = \begin{cases}
p q a_{s}^{-1} \quad &\text{if } s\in K \\
a_{s} \quad &\text{if } s\notin K
\end{cases} \quad (s\in\{0,1,\ldots,7\}).
\end{equation}
The van Diejen operator $\cD_{x}{(a \lvert q, t )}$ satisfies the relation
\begin{equation}
\label{eq_gauge_transform_U}
U_{K}(x ; a \lvert q )^{-1} \circ \cD_{x}{(a \lvert q , t )} \circ U_{K}(x ; a \lvert q ) = \cD_{x}{(a_{K(pq)} \lvert q , t)}.
\end{equation}
\end{lemma}
\noindent{}(Proof of Lemma~\ref{lemma_gauge_transform_U} is given in Section~\ref{eq_symmetries_proof}.)

The following results follow naturally from the relations \eqref{eq_gauge_transform_V} and \eqref{eq_gauge_transform_U} (and using the constant eigenfunction under certain parameter restriction):
\begin{proposition}
\label{prop_gauge_transf_V}
Let $\varphi(x)$ be an eigenfunction of the van Diejen operator \newline $\cD_{x}{( a \lvert q , p q t^{-1})}$ with eigenvalue $\Lambda \in \C$, then the function
\begin{equation}
\label{eq_prop_eigenfunction_V}
\psi(x) = V(x \lvert q , t ) \varphi(x) ,\quad V(x \lvert q, t ) = \prod_{1 \leq i \less j \leq m} \frac{1}{\Gamma( t x_{i}^{\pm} x_{j}^{\pm}; p , q )},
\end{equation}
is an eigenfunction of the operator $\cD_{x}{(a \lvert q,t)}$ with the same eigenvalue.

Furthermore, suppose that the parameters satisfy
\begin{equation}
\label{eq_ellipticity_condition_V}
a_{0} \cdots a_{7}  p^{2m} q^{2m} t^{-2m}  = p^{4} q^{4} t^{-2}
\end{equation}
then the function $V(x \lvert q, t )$ is an eigenfunction of the van Diejen operator $\cD_{x}{(a \lvert q , t )}$ with eigenvalue $\Lambda_{0}(a \lvert q  , p q t^{-1} )$ in \eqref{eq_vD_constant_eigenvalues}.
\end{proposition}
\noindent(Proof of Proposition~\ref{prop_gauge_transf_V} is given in Section~\ref{eq_symmetries_proof}.)
\begin{proposition}
\label{prop_gauge_transf_U}
Let $K\subseteq\{0,1,\ldots,7\}$ with even cardinality, i.e. $\abs{K} \in 2 \Z_{\geq 0}$, and $\varphi(x)$ be an eigenfunction of the van Diejen operator $\cD_{x}{(a_{K(pq)}\lvert q , t )}$ with eigenvalue $\Lambda$. The function
\begin{equation}
\label{eq_prop_eigenfunction_U}
\psi(x) = U_{K}(x ; a \lvert q ) \varphi(x) , \quad U_{K}(x ; a\lvert q ) = \prod_{1 \leq i \leq m} \prod_{s\in K} \frac{1}{\Gamma( a_{s} x_{i}^{\pm};p , q)},
\end{equation}
is then an eigenfunction of the van Diejen operator $\cD_{x}{(a \lvert q, t )}$ with the same eigenvalue $\Lambda$. 

Furthermore, suppose the parameters satisfy
\begin{equation}
\label{eq_balancing_condition_U}
 (\prod_{s\notin K} a_{s} )p^{\abs{K}} q^{\abs{K}} t^{2 m}  = (\prod_{s\in K} a_{s})p^{2} q^{2}  t^{2} 
\end{equation}
then the function $U_{K}(x ; a \lvert q )$ is an exact eigenfunction of the van Diejen operator $\cD_{x}{(a \lvert q , t )}$ with eigenvalue $\Lambda_{0}(a_{K(pq)} \lvert q, t )$ in \eqref{eq_vD_constant_eigenvalues}.
\end{proposition}
\noindent(Proof is Proposition~\ref{prop_gauge_transf_U} is given in Section~\ref{eq_symmetries_proof}.)
{\medskip}

A key ingredient for the integral transformations are the kernel functions, and the corresponding kernel function identities, found by Ruijsenaars \cite{Rui09a} and by Komori, Shiraishi, and one of the authors in \cite{KNS09}. Let us recall these identities and express them in a form more suitable for our purposes here:
\begin{lemma}[Theorem~2.3 (1) of \cite{KNS09}]
\label{lemma_Cauchy_KFI}
Under the balancing condition
\begin{equation}
a_{0} \cdots a_{7} t^{2(m-n)} = p^{2} q^{2} t^{2},
\label{eq_balancing_condition_KFI}
\end{equation}
the function
\begin{equation}
\label{eq_Cauchy_KF}
\Phi(x, y \lvert q,  t) = \prod_{1 \leq i \leq m} \prod_{1 \leq k \leq n} \Gamma(p^{\half} q^{\half} t^{-\half} x_{i}^{\pm} y_{k}^{\pm} ; p ,q)
\end{equation}
satisfies the functional identity
\begin{equation}
\label{eq_Cauchy_KFI}
\cD_{x}(a \lvert q, t ) \Phi(x , y \lvert q, t) = \cD_{y}( b \lvert q, t ) \Phi(x,y\lvert q, t)
\end{equation}
where $b=(b_0 , \ldots , b_7)\in(\C^{\ast})^{8}$ is given by $b_{s} = p^{\half} q^{\half} t^{\half} a_{s}^{-1}$ for all $s \in \{0,1,\ldots,7\}$.
\end{lemma}
\noindent(The identity was proven in the additive notation in \cite{Rui09a} when $m=n$ and for general $m$ and $n$ in \cite{KNS09}. For the convenience of the reader, we present the relation between our conventions and those of \cite{Rui09a} and \cite{KNS09} in Appendix~\ref{app_KFI}.)

The function $\Phi(x,y\lvert q ,t )$ in \eqref{eq_Cauchy_KF} is referred to as a kernel function of \emph{Cauchy type} \cite{KNS09}.
{\smallskip}

The Cauchy type kernel function can be used to construct an integral transform that maps given solutions $\varphi(y;b\lvert q , t)$ of the van Diejen operator $\cD_{y}(b\lvert q ,t )$, that are holomorphic in some $n$-dimension annulus $\mathbb{A}_{\rho}^{n}$, where
\begin{equation}
\mathbb{A}^{n}_{\rho} = \{ y \in (\C^{\ast})^{n} \lvert \ \rho \leq \abs{y_{k}} \leq \rho^{-1} \ ( k\in\{1,2,\ldots,n\})\} \quad (\rho \in (0,1]),
\label{eq_annulus}
\end{equation}
for some $\rho\in(0,\abs{q}]$, to solutions of the operator $\cD_{x}(a\lvert q, t)$, if the parameters satisfy the balancing condition \eqref{eq_balancing_condition_KFI}, which follows as a direct consequence of the kernel function identity \eqref{eq_Cauchy_KFI} and that the van Diejen operator is formally self-adjoint with respect to the symmetric $\C$-bilinear form \eqref{eq_inner_product}. We refer to these transformations as the \emph{eigenfunction transforms of Cauchy-type}. Combining this with the results of Lemmas~\ref{lemma_gauge_transform_V} and \ref{lemma_gauge_transform_U} yields two of our main results:
\begin{theorem}
\label{thm_Cauchy_eigenfunction_transform_Type_II}
Let $K\subseteq \{0,1,\ldots,7\}$ and have even cardinality, the parameters $p,q,t\in\C^\ast$ and $b = (b_{0},\ldots, b_{7}) \in (\C^{\ast})^{8}$ satisfy
\begin{equation}
\begin{cases}\abs{p} < \abs{q t }, \ &\text{ if } K \neq \emptyset \\
\abs{p} < \min( \abs{ q  t } , \abs{ q^{-1} t^{n-m-1}})\ &\text{ if } K = \emptyset
\end{cases}, \quad \abs{q}< 1, \quad \abs{t} < 1, \quad \abs{b_{s}}< 1 \ (s\notin K).
\end{equation}
Let $\varphi(y;b_{K(pq)} \lvert q , t)$ $(y=(y_1,\ldots,y_n))$ be holomorphic in a domain that contains the $n$-dimensional annulus $\mathbb{A}_{\abs{q}}^{n}$ and suppose that $\varphi(y;b_{K(pq)} \lvert q , t)$ is an eigenfunction of the van Diejen operator $\cD_{y}(b_{K(pq)} \lvert q, t)$ with eigenvalue $\Lambda\in\C$, then the function
\begin{equation}
\psi_{K}(x ; a \lvert q , t) = \int_{\mathbb{T}^{n}} d\omega_{n}(y) w_{n}(y; b \lvert q , t) \Phi(x , y \lvert q , t) U_{K}(y ; b \lvert q ) \varphi(y ; b_{K(pq)} \lvert q , t)
\end{equation}
is an eigenfunction of the van Diejen operator $\cD_{x}(a \lvert q, t)$ $($where $a=p^{\half} q^{\half} t^{\half} b^{-1}$ $)$, with the same eigenvalue $\Lambda$ and $x$ in the domain 
\begin{multline}
\label{eq_theorem_domain}
\{ x \in (\C^{\ast})^{m} \lvert \ \abs{p^{\half} q^{-\half} t^{-\half}} < \abs{ x_{i} } < \abs{ p^{-\half} q^{\half} t^{\half}} \ (i \in\{1,2,\ldots,m\}), \\
x_{i}^{\ve} \notin a_{s} p^{\Z_{<0}} q^{{\Z_{\leq 0}}} \ ( s\notin K ; \ i\in\{1,2,\ldots,m\};\ \ve\in\{\pm\}), \\
x_{i}^{\ve} x_{j}^{\ve^\prime} \notin p^{\Z_{<0}} q^{\Z_{\leq 0}} t  \ (1 \leq i < j \leq m ; \ \ve,\ve^\prime \in\{\pm\}) \},
\end{multline} 
if the parameters satisfy the balancing condition \eqref{eq_balancing_condition_KFI}.
\end{theorem}
\noindent{}(Proof of Theorem~\ref{thm_Cauchy_eigenfunction_transform_Type_II} is given in Section~\ref{sec_proofs}.)
\begin{theorem}
\label{thm_Cauchy_eigenfunction_transform_Type_I}
Let $K\subseteq \{0,1,\ldots,7\}$ and have even cardinality, the parameters \newline $p,q \in\C^\ast$ and $b = (b_{0},\ldots, b_{7}) \in (\C^{\ast})^{8}$ satisfy
\begin{equation}
\begin{cases} \abs{p} < \abs{q t } \ &\text{ if } K \neq \emptyset \\
\abs{p} < \min(\abs{q  t } , \abs{ q^{-1} t^{n-m-1}}) \ &\text{ if } K = \emptyset
\end{cases} , \quad \abs{q}< 1, \quad \abs{b_{s}}< 1 \ (s\notin K).
\end{equation}
Let $\varphi(y;b_{K(pq)} \lvert q , p q t^{-1} )$ $(y=(y_1,\ldots,y_n))$ be holomorphic in a domain that contains $\mathbb{A}^{n}_{\abs{q}}$ and suppose that $\varphi(y;b_{K(pq)} \lvert q , p q t^{-1})$ is an eigenfunction of the van Diejen operator $\cD_{y}(b_{K(pq)} \lvert q, p q t^{-1} )$ with eigenvalue $\Lambda\in\C$, then the function
\begin{multline}
\tilde{\psi}_{K}(x ; a \lvert q , t) = \int_{\mathbb{T}^{n}} d\omega_{n}(y) w_{n}(y; b \lvert q , t) \Phi(x , y \lvert q , t) \\
\cdot U_{K}(y ; b \lvert q ) V(y \lvert q , t) \varphi(y ; b_{K(pq)} \lvert q , p q t^{-1})
\end{multline}
is an eigenfunction of the van Diejen operator $\cD_{x}(a \lvert q, t)$ $($where $a=p^{\half} q^{\half} t^{\half} b^{-1}$ $)$, with the same eigenvalue $\Lambda$ and $x$ in the domain \eqref{eq_theorem_domain}, if the parameters satisfy the balancing condition \eqref{eq_balancing_condition_KFI}.
\end{theorem}
\noindent{}(Proof of Theorem~\ref{thm_Cauchy_eigenfunction_transform_Type_I} is given in Section~\ref{sec_proofs}.)
{\medskip}

Based on these results, it is possible to obtain several different types of transformations for the eigenfunctions of the van Diejen operator. We have decided to collect these transformations in Appendix~\ref{app_transforms} for the convenience of the reader, but it is worth pointing out that some of the eigenfunction transformations could possibly be used to construct simultaneous shift operators for the parameter space: The integral transformation in \eqref{eq_Cauchy_transf_eigenfunctions_4} for $I=\emptyset$ and $J=\{0,1,\ldots,7\}$ yields a transformation of the eigenfunctions with parameters $a=(a_0 , \ldots , a_7)$ to eigenfunctions with parameters $t^{-\half} a=(t^{-\half}a_{0},\ldots,t^{-\half}a_{7})$. This can also be verified by direct calculations as the function
\begin{equation}
\mathcal{K}(x,y;a\lvert q, t) = V(x \lvert q , t) \Phi(x , y \lvert q, p q t^{-1}  ) U_{\{0,1,\ldots,7\}}(y ; t^{\half} a \lvert q ) V(y \lvert q , t )
\end{equation}
satisfies the kernel function identity
\begin{equation}
\cD_{x}{(a \lvert q, t )} \mathcal{K}(x , y ;a \lvert q , t ) = \cD_{y}{(t^{\half} a \lvert q , t )} \mathcal{K}(x , y;a \lvert q , t )
\end{equation} 
if the parameters satisfy the condition 
\begin{equation}
a_{0} \cdots a_{7}  p^{2(m-n)} q^{2(m-n)} t^{{2(n-m)}}  = p^{4} q^{4} t^{-2}.
\end{equation}
(Note that there is a difference in the parameters for the kernel function $\mathcal{K}(x,y ;a \lvert q, t)$ and the integrand in \eqref{eq_Cauchy_transf_eigenfunctions_4} for $I=\emptyset$ and $J=\{0,1,\ldots,7\}$.) It should be possible to use the kernel function above, or equivalently Eq.~\eqref{eq_Cauchy_transf_eigenfunctions_4}, to construct integral operators that simultaneously shifts the parameters $a \in(\C^\ast)^{8}$ by $t^{\half k}$ with $k\in\Z$, by iterating the integral transform. However, each transformation yields an additional constraint on the parameters due to the shifts and these simultaneous shift operators cannot be obtained by simply applying the transformation several times. 

Finally, we wish to stress that the inclusion of $x$-dependent gauge functions $U_{I}(x;a\lvert q)$ for $I\neq \emptyset$ and $V(x\lvert q ,t)$ in \eqref{eq_Cauchy_transf_eigenfunctions_1}-\eqref{eq_Cauchy_transf_eigenfunctions_4} also changes the balancing condition, as can be seen above, and that Eqs. \eqref{eq_Cauchy_transf_eigenfunctions_1}-\eqref{eq_Cauchy_transf_eigenfunctions_4} require different balancing conditions than \eqref{eq_balancing_condition_KFI}; see Appendix~\ref{app_transforms} for details. 
{\medskip}

The results of \cite{KNS09} allow us to introduce a second type of integral transformation using the kernel function of \emph{dual-Cauchy type}. (Again, we are only presenting the results in a form suitable for our purposes.) 
\begin{lemma}[Theorem~2.3 (2) of \cite{KNS09}]
\label{lemma_dual-Cauchy_KFI}
Under the balancing condition
\begin{equation}
\label{eq_dual_balancing_condition}
a_{0} \cdots a_{7} q^{2n} t^{2m} = p^{2} q^{2} t^{2},
\end{equation}
the function 
\begin{equation}
{\tPhi}(x,y) = \prod_{1 \leq i \leq m} \prod_{1 \leq k \leq n} x_{i}^{-1} \theta( x_{i} y_{k}^{\pm} ; p)
\label{eq_dual_Cauchy_KF}
\end{equation}
satisfies the functional identity
\begin{equation}
\label{eq_dual_Cauchy_KFI}
t^{-m} \theta( t ; p ) \cD_{x}( a \lvert q , t ) {\tPhi}(x,y) = - q^{-n} \theta( q ;p) \cD_{y}(a \lvert t , q ) {\tPhi}(x,y).
\end{equation}
\end{lemma}
By the same arguments that was presented above, the kernel function of dual-Cauchy type can also be used to construct integral transformations for the eigenfunctions of the van Diejen operator and we obtain the following results:

\begin{theorem}
\label{thm_dual-Cauchy_eigenfunction_transform_Type_II}
Let $K \subseteq \{0,1,\ldots,7\}$ have even cardinality and the model parameters $p,q,t \in \C^{\ast}$, and $a = (a_0 ,\ldots, a_7)\in(\C^{\ast})^{8}$ satisfy
\begin{equation}
\begin{cases}\abs{p} < 1 &\text{if } K\neq \emptyset\\
\abs{p} < \abs{q^{n-1} t^{m-1}} &\text{if } K = \emptyset
\end{cases} , \quad  \abs{q} <1 , \quad \abs{t} < 1, \quad \abs{a_{s}} < 1 \ (s\notin K).
\end{equation}
Let $\varphi( y ; a_{K(p t)} \lvert t , q)$ $(y=(y_1,\ldots,y_n))$ be holomorphic in a domain that contains the $n$-dimensional annulus $\mathbb{A}^{n}_{\abs{t}}$ and suppose that $\varphi( y ;  a_{K(p t)} \lvert t , q)$ is an eigenfunction of the van Diejen operator $\cD_{y}( a_{K(p t)} \lvert t , q)$ with eigenvalue $\Lambda \in \C$, then the function
\begin{equation}
\psi_{K}^{\vee}(x ; a \lvert q , t) = \int_{\mathbb{T}^{n}} d\omega_{n}(y) w_{n}(y ; a \lvert t , q ) \Phi^{\vee}(x ,y ) U_{K}(y ; a \lvert t ) \varphi( y ;  a_{K(p t)} \lvert t , q)
\end{equation}
is an eigenfunction of the van Diejen operator $\cD_{x}(a \lvert q , t)$, with eigenvalue 
\begin{equation}
\label{eq_dual-Cauchy_eigenvalue}
-( q^{-n}t^{m} \theta( q ;p) / \theta(t;p))\Lambda
\end{equation}
and $x\in(\C^{\ast})^{m}$, if the parameters satisfy the balancing condition \eqref{eq_dual_balancing_condition}.
\end{theorem}
\noindent{}(Proof of Theorem~\ref{thm_dual-Cauchy_eigenfunction_transform_Type_II} is given in Section~\ref{sec_proofs}.)
\begin{theorem}
\label{thm_dual-Cauchy_eigenfunction_transform_Type_I}
Let $K \subseteq \{0,1,\ldots,7\}$ have even cardinality and the model parameters $p,q,t \in \C^{\ast}$, and $a = (a_0 ,\ldots, a_7)\in(\C^{\ast})^{8}$ satisfy
\begin{equation}
\begin{cases}\abs{p} < 1 &\text{if } K\neq \emptyset\\
\abs{p} < \abs{q^{n-1} t^{m-1}} &\text{if } K = \emptyset
\end{cases} , 
 \quad \abs{t} < 1, \quad \abs{a_{s}} < 1 \ (s\notin K).
\end{equation}
Let $\varphi( y ; a_{K(p t)} \lvert t , p q^{-1} t )$ $(y=(y_1,\ldots,y_n))$ be holomorphic in a domain that contains the $n$-dimensional annulus $\mathbb{A}^{n}_{\abs{t}}$ and suppose that $\varphi( y ;  a_{K(p t)} \lvert t ,  p q^{-1} t )$ is an eigenfunction of the van Diejen operator $\cD_{y}( a_{K(p t)} \lvert t ,  p  q^{-1} t )$ with eigenvalue $\Lambda \in \C$, then the function
\begin{multline}
\tilde{\psi}_{K}^{\vee}(x ; a \lvert q , t) = \int_{\mathbb{T}^{n}} d\omega_{n}(y) w_{n}(y ; a \lvert t , q ) \Phi^{\vee}(x ,y ) \\
\cdot U_{K}(y ; a \lvert t ) V(y \lvert t , q ) \varphi( y ;  a_{K(p t)} \lvert t ,  p q^{-1} t )
\end{multline}
is an eigenfunction of the van Diejen operator $\cD_{x}(a \lvert q , t)$, with eigenvalue in \eqref{eq_dual-Cauchy_eigenvalue}
and $x\in(\C^{\ast})^{m}$, if the parameters satisfy the balancing condition \eqref{eq_dual_balancing_condition}.
\end{theorem}
\noindent{}(Proof of Theorem~\ref{thm_dual-Cauchy_eigenfunction_transform_Type_I} is given in Section~\ref{sec_proofs}.)

Using the eigenfunction transformations with the dual-Cauchy kernel allows us to construct additional eigenfunction transformations for the eigenfunctions of the van Diejen operator. We also collect these transformations in Appendix~\ref{app_transforms}.

The kernel function $\tPhi(x,y)$ is clearly holomorphic on $(\C^\ast)^{m}\times(\C^{\ast})^{n}$ so the integral transforms yields holomorphic functions of $x\in(\C^{\ast})^{m}$. Furthermore, this kernel function has a known expansion in terms of the elliptic interpolation functions of type $BC_{m}$ \cite{NI16,NI19} which we believe can be used to evaluate the integrals \cite{NI19}.

\begin{remark} In the $n=1$ case, the van Diejen operators $\cD_{y}(a \lvert t , q )$ and $\cD_{y}(a \lvert t , q^\prime )$, for two bases $q$ and $q^\prime$, differ by an additive constant only (see also Remark \ref{remark_zeroth_order_coefficient}). Since the balancing condition \eqref{eq_dual_balancing_condition} becomes $a_{0}\cdots a_{7} t^{2m} = p^{2} t^{2}$ which does not depend on $q$, we obtain the following equation for the difference of two analytic difference operators from Lemma~\ref{lemma_dual-Cauchy_KFI}:
$$
\bigl(\frac{1}{\theta( q; p )} \cD_{x}(a \lvert q  , t) - \frac{1}{\theta(q^{\prime}  ; p )}\cD_{x}(a \lvert q^{\prime}  , t) \bigr) \prod_{i=1}^{m} x_{i}^{-1} \theta( x_{i} y^{\pm};p) = E_{0}  \prod_{i=1}^{m} x_{i}^{-1} \theta(x_{i} y^{\pm};p),
$$
where $E_0 = E_{0}(a \lvert q , q^\prime, t)$ is a (known) constant.
\end{remark}

\subsection{Previously known solutions}
\label{sec_previous_results}
As mentioned in the introduction, the van Diejen operator has known eigenfunctions in certain special cases  \cite{Spi03,Spi04,Kom04,Cha07,Spi07,Spi08,KNS09}. Here, we recall some classes of known eigenfunctions to which one can apply our transformations to obtain new eigenfunctions.

\subsubsection{Known eigenfunctions in the univariate case}
Exact eigenfunctions of Bethe Ansatz type (or Floquet-Block type) for the van Diejen operator has been constructed by Chalykh \cite{Cha07} in the univariate case. These eigenfunctions are holomorphic in $\C^{\ast}$, require that the parameters $a$ are restricted to a certain lattice (or rather, integer values of the couplings in the additive notation), and also involve the solutions to a transcendental system of equations for the Weierstrass elliptic functions.

Exact solutions to the eigenvalue equation for the van Diejen operator has been constructed in terms of the \emph{elliptic hypergeometric series} $_{12}V_{11}$ \cite{DJKMO88,FT97} for 5 generic parameters and one quantized parameter \cite{Spi03,Spi04,Kom04}. Exploring that the elliptic hypergeometric series can be expressed in terms of explicit integrals of the elliptic Gamma functions, exact solutions were also found in terms of the elliptic hypergeometric integral $I(b_{0},\ldots,b_{5},c x , c x^{-1};p,q)$ for $6$ generic parameters.

We also wish to mention results obtained by Ruijsenaars~\cite{Rui04,Rui09a,Rui15} that showed the symmetries in the parameter space and explored the Hilbert space aspects of the eigenfunctions of the van Diejen operator. In particular, Ruijsenaars proved the existence of exact eigenfunctions of the univariate van Diejen model and that these eigenfunctions form a complete orthogonal basis. Of particular interest is that Ruijsenaars found joint eigenfunctions of modular pairs of van Diejen operators. (Here, the modular pairs refer to two van Diejen operators with $p$ and $q$ interchanged, \emph{i.e.} $\cD_{x}(a\lvert p , q , t)$ and $\cD_{x}(a\lvert q , p , t)$.) 
\begin{remark}
\label{remark_p_q_symmetry}
It should be noted that the functions \eqref{eq_Selberg_eigenfunction_type_II} and \eqref{eq_Selberg_eigenfunction_type_I} are also joint eigenfunctions of modular pairs as they are invariant under $p \leftrightarrow q$. Furthermore, it is possible to generate more joint eigenfunctions by the eigenfunction transformations in Theorems~\ref{thm_Cauchy_eigenfunction_transform_Type_II} and \ref{thm_Cauchy_eigenfunction_transform_Type_I}.
\end{remark}

\subsubsection{Known eigenfunctions in the multivariate case}
\label{sec_special_cases_eigenfunction}
Let us proceed to the special cases where eigenfunctions of the many-variable van Diejen operator are known.

Under the \emph{ellipticity condition} 
\begin{equation}
\label{eq_ellipticity_condition_1}
a_0 \cdots a_7 t^{2 m} = p^{2} q^{2} t^{2},
\end{equation} 
it is known that the van Diejen operator can be expressed as \cite{Rui04,KNS09}
\begin{equation}
\cD_{x}{(a \lvert q, t )} =  \Lambda_{0}(a \lvert q , t ) +  \sum_{1 \leq i \leq m} \sum_{\ve = \pm}  A_{i}^{\ve}(x ; a \lvert q,t ) \Bigl( T_{q, x_{i}}^{\ve} - 1 \Bigr)
\label{eq_vD_operator_balanced}
\end{equation}
where $\Lambda_{0}(a \lvert q , t ) = A^{0}( x ; a \lvert q , t) +  \sum_{1 \leq i \leq m ; \ve = \pm} A_{i}^{\ve}(x; a \lvert q , t ) $, under \eqref{eq_ellipticity_condition_1}, is an elliptic function without poles with respect to the $x$-variables and thus a constant. Setting $n=0$ in the kernel function identities in Lemmas~\ref{lemma_Cauchy_KFI} and \ref{lemma_dual-Cauchy_KFI} yields that
\begin{equation}
\label{eq_vD_constant_eigenvalues}
\Lambda_{0}(a \lvert q , t )  = A^{0}(- ; p^{\half} q^{\half} t^{\half} a^{-1} \lvert q , t)
=  - \frac{t^{m} \theta( q ; p)}{\theta( t;p)}A^{0}(- ; a \lvert t , q).
\end{equation}
From the expression in \eqref{eq_vD_operator_balanced}, it is clear that any constant is an elementary eigenfunction of the van Diejen operator for parameters satisfying \eqref{eq_ellipticity_condition_1}, \emph{i.e.}
\begin{equation}
\cD_{x}{(a \lvert q , t )}\Bigl\lvert_{a_0 \cdots a_7  t^{2 m } = p^{2} q^{2}t^{2}}. 1 = 1. \Lambda_{0}(a \lvert q , t).
\label{eq_vD_constant_eigenvalue_eq}
\end{equation}

It is also known that the $W_{m}$-invariant monomials are exact eigenfunctions in the so-called \emph{free case} \cite{Rui09b,Rui15}: When the parameters $a$ are given by any permutation of
\begin{equation}
\label{eq_free_parameters}
 a_{\text{free}} = (1 , -1 ,p^{\half},-p^{\half}, q^{\half}, -q^{\half} , p^{\half} q^{\half}, -p^{\half} q^{\half}),
\end{equation}
and we set $t=1$ afterwards, the shift coefficients in \eqref{eq_coefficients_positive_multiplicative} are reduce to $1$ and the zeroth order coefficient $A^{0}(x ; a_{\text{free}}\lvert q ,t)$ is reduced to a constants, which is $0$ in our convention; see Section~\ref{sec_vD_symmetries_and_properties}. The van Diejen operator \eqref{eq_van_diejen_operator} reduces to (essentially) the sum of shift operator, \emph{i.e.} $\cD_{x}{(a_{\text{free}}\lvert q, t)}\lvert_{t=1} = \sum_{1 \leq i \leq m} ( \  T_{q,x_{i}} + T_{q,x_{i}}^{-1} \ )$, and have exact eigenfunctions given by products of (free) one-variable eigenfunctions. If we impose that the eigenfunctions are $W_{m}$-invariant, we obtain that
\begin{equation}
m_{\lambda}(x) = \sum_{\sigma \in W_{m}} \sigma \bigl(\prod_{1 \leq i \leq m} x_{i}^{\lambda_i}\bigr) \quad ( \lambda \in \C^{m}),
\end{equation}
where the sum is over distinct permutations, satisfies the eigenvalue equation
\begin{equation}
\cD_{x}{(a_{\text{free}}\lvert q, t)}\lvert_{t=1} m_{\lambda}(x) = m_{\lambda}(x) \  d^{\text{free}}_{\lambda}(q), \quad d^{\text{free}}_{\lambda}(q) = \sum_{ 1 \leq i \leq m} q^{\lambda_i} + q^{- \lambda_i}
\end{equation}
for any $\lambda = (\lambda_1,\ldots,\lambda_m) \in \C^{m}$. Restricting $\lambda$ to partitions would yield $W_{m}$-invariant Laurent polynomials in the $x$-variables.

In a recent paper, van Diejen and G{\"o}rbe \cite{vDG21} considered the eigenvalue equation for a discrete variant of the van Diejen operator restricted to a finite-dimensional space of lattice functions and explored the Hilbert space aspects of this model as well as exact eigenfunctions in special cases that have an interpretation of Pieri-type formulas for the eigenfunctions of the van Diejen model. 

\subsubsection{On $p$-shifts in the model parameters.}
\label{sec_shift_properties}

Another interesting property of the van Diejen operator is its transformation under $p$-shifts of the model parameters. It is straightforward to check that the van Diejen operator, in our conventions, only changes by a multiplicative factor when simultaneously shifting two parameters (say) $a_{0}$ and $a_{1}$ to $p a_{0}$ and $p^{-1}a_{1}$, more specifically
\begin{equation}
\cD_{x}{((p a_{0} , p^{-1} a_{1} , a_{2},\ldots,a_{7})\lvert q , t)} = (a_{1}/(p a_{0})) \cD_{x}{(a \lvert q , t )},
\label{eq_parameter_shifts_opposite}
\end{equation}
or for any other pairs $a_{r}$ and $a_{s}$ ($r,s\in\{0,1,\ldots,7\}$, $r \neq s$) by permutation symmetry. Intrigued by this transformation property, we investigate other possible transformations of the van Diejen operator under $p$-shifts of the model parameters. In particular, we consider the van Diejen operators $\cD_{x}{( p^{\half} a \lvert q , t )}$ and $\cD_{x}{( p^{\epsilon_{K}} a \lvert q, t )},$ where $\epsilon_{K} = \sum_{s \in K} \epsilon_{s}$ for any $K\subseteq \{0,1,\ldots,7\}$ and $\epsilon_{s}$ ($s \in \{0,1,\ldots,7\}$) are the canonical basis in $\C^{8}$, and show that they can be expressed as a transformation of the van Diejen operator: \newline
Suppose $f(x) = f(x ; a \lvert q , t)$ satisfies the $q$-difference equations 
\begin{equation}\vspace{-5pt}
T_{q,x_{i}} f(x) = p q^{3} t^{-(m-1)} (a_{0} \cdots a_{7})^{-\half} x_{i}^{4} \ f(x) \quad (i \in\{1,2,\ldots,m\}),
\label{eq_f_q-difference_equation}
\end{equation}
then the van Diejen operator satisfies
\begin{equation}
\label{eq_parameter_shift_half-period_shift_x}
\cD_{p^{\half} x}{(a \lvert q , t )} =p^{-1} q^{-2} (a_{0} \cdots a_{7})^{\half} f(x)^{-1} \circ \cD_{x}{( p^{\half} a \lvert q , t )} \circ f(x).
\end{equation}
\newline
Let $K\subseteq \{0,1,\ldots,7\}$ with $\abs{K} = 2$ and suppose $g(x) = g(x \lvert q)$ satisfies
\begin{equation}
T_{q,x_{i}} g(x) = q^{-1} x_{i}^{-2} g(x) \quad (i\in\{1,2,\ldots,m\}),
\label{eq_g_q-difference_equation}
\end{equation}
then the van Diejen operator satisfies
\begin{equation}
D_{x}{ ( p^{\epsilon_{K}} a \lvert q , t)} = \frac{q}{\prod_{s \in K} a_{s}} g(x)^{-1} \circ D_{x}{(a \lvert q , t)} \circ g(x) .
\label{eq_parameter_shift_epsilon}
\end{equation}

There are two interesting observations here: \newline
(1) From Eqs. \eqref{eq_parameter_shifts_opposite}, \eqref{eq_parameter_shift_half-period_shift_x}, and \eqref{eq_parameter_shift_epsilon}, it is clear that the parameters of the van Diejen model (defined by the principal operator) have symmetries under shifts by $p^{P}$ where $P = Q(E_{8})$ is the $E_{8}$ root lattice; see Section~\ref{sec_E8_symmetry}. It follows from straightforward checks that these $p$-shifts are generated by the simple roots $\alpha_{s}$ ($s\in\{0,1,\ldots,7\}$) \eqref{eq_simple_roots} related to the root system $\Delta(E_{8})$. We need therefore only consider $a \in ( \C^{\ast} )^{8} / p^{P}$ and use the symmetries in \eqref{eq_parameter_shifts_opposite}, \eqref{eq_parameter_shift_half-period_shift_x}, and \eqref{eq_parameter_shift_epsilon} for other values. \newline
(2) Combining \eqref{eq_vD_constant_eigenvalue_eq} and \eqref{eq_parameter_shift_half-period_shift_x} yields that the function $f(x)$, satisfying \eqref{eq_f_q-difference_equation}, is an eigenfunction of the van Diejen operator $\cD_{x}{( p^{\half} a \lvert q , t )}$ if the model parameters satisfy the ellipticity condition \eqref{eq_ellipticity_condition_1}. Similarly, it also follows that the function $g(x)$, satisfying \eqref{eq_g_q-difference_equation}, is an eigenfunction of the van Diejen operator if the model parameters satisfy $
a_{0}\cdots a_{7} t^{2m} = q^{2} t^{2}$. These functions could also be used as solutions $\varphi(y)$ for our transformations. 

\begin{remark}
\label{remark_quasi-constants}
It is straightforward to construct (meromorphic) solutions to the $q$-difference equations above by using the multiplicative theta function with base $q$; see also Lemmas~\ref{lemma_parameter_p_shift} and \ref{lemma_half-period_shift}. However, it is clear that the defining $q$-difference equations, \emph{i.e.}  \eqref{eq_f_q-difference_equation} and \eqref{eq_g_q-difference_equation}, do not have unique solutions. As such, the functions $f(x)$ and $g(x)$ are only defined up to multiplication by any $q$-periodic (meromorphic) function. Throughout the paper, we refer to these $q$-periodic functions as \emph{quasi-constants}. Note that the eigenfunctions of the van Diejen operator are also only defined up to multiplication by quasi-constant. A thorough investigation of the quasi-constants is beyond the scope of this paper and we will always choose the quasi-constant that is most convenient for our purposes as we proceed.
(Note that one way of determining the quasi-constants for the eigenfunctions of the van Diejen operator is to require that they are simultaneous eigenfunctions of a modular pair.)
\end{remark}

\section{Properties and symmetries of the van Diejen operator}
\label{sec_vD_symmetries_and_properties}

In the previous Section we introduced the eigenfunction transformations that can be used to generate new eigenfunctions of the van Diejen operator starting from previously known eigenfunctions or elementary eigenfunctions presented in Section~\ref{sec_special_cases_eigenfunction}. The gauge symmetries in Lemmas~\ref{lemma_gauge_transform_V} and \ref{lemma_gauge_transform_U} follow from the symmetries of the coefficients, as do the properties under $p$-shifts, while the integral transformations can be seen as a consequence of the kernel function identities and the (formal) self-adjointness of the van Diejen operator with respect to the symmetric $\C$-bilinear form. These properties are explained in this section.

Firstly, let us start by giving the definition of the van Diejen operator used in this paper. We recall that the van Diejen operator is given by an analytic difference operator of the form
\begin{equation}
\cD_{x}{(a \lvert q,t)} =  A^{0}(x ; a \lvert q , t) +  \sum_{1 \leq i \leq m} \sum_{\ve=\pm} A_{i}^{\ve}(x ; a \lvert q,t ) T_{q , x_{i}}^{\ve},
\end{equation}
for shift coefficients $A_{i}^{\pm}(x ; a \lvert q , t)$ ($i\in\{1,2,\ldots,m\}$) given in \eqref{eq_coefficients_positive_multiplicative}
, \emph{i.e.}
\begin{equation}
A_{i}^{\ve}(x ; a \lvert q,t ) = \frac{\prod_{0 \leq s \leq 7} \theta(a_s x_{i}^{\ve} ; p ) }{ \theta(x_{i}^{\ve 2} , q x_{i}^{\ve 2} ; p )} \prod_{j \neq i } \frac{ \theta( t x_{i}^{\ve} x_{j}^{\pm};p) }{\theta(x_{i}^{\ve} x_{j}^{\pm} ; p)} \quad (i\in\{1,2,\ldots,m\}, \ \ve \in \{\pm\}),
\end{equation}
and zeroth order coefficient given as
\begin{equation}
\label{eq_vD_external_potential_1}
A^{0}(x;a\lvert q ,t) = \frac{1}{2} \sum_{0 \leq r \leq 3} A^{0}_{r}(x;a\lvert q, t)
\end{equation} 
where
\begin{equation}
\label{eq_vD_external_potential_2}
A^{0}_{r}(x ; a \lvert q,t) =L_r^{(m)}(a \lvert q, t) \frac{\prod_{0 \leq s \leq 7} \theta(c_r q^{-\half} a_s ; p)}{\theta(t, q^{-1} t ; p)} \prod_{1 \leq j \leq m} \frac{\theta( c_r q^{-\half} t x_{j}^{\pm};p)}{\theta(c_r q^{-\half} x_{j}^{\pm};p)}
\end{equation}
with $L_0^{(m)} = L_1^{(m)} = 1$, $L_2^{(m)}(a \lvert q , t) = p^2 L_{3}^{(m)}(a \lvert q, t)^{-1} =  p^{-1} q^{-2} t^{m} ( a_0 \cdots a_7)^{\half}$, $c_0 = - c_1 = 1,$ and $c_2 = -c_3^{-1} = p^{\half}.$
\begin{remark}
\label{remark_zeroth_order_coefficient}
We remark that the zeroth order coefficient $A^{0}(x ; a \lvert q , t )$ is determined, up to an additive constant, by specifying the residues at 
\begin{equation}
 x_{i} \in \ve p^{\half \Z} q^{-\half} \text {\ \ \ and\ \ \   } x_{i} \in \ve p^{\half \Z} q^{\frac{1}{2}} \quad (i\in\{1,2,\ldots,m\}, \ \ve\in\{\pm\}).
\end{equation}
More specifically, the zeroth order coefficient is characterized by the following conditions: 
\newline(1) $A^{0}(x ; a \lvert q , t)$ is an elliptic function of each $x_{i} \in \C^{\ast}$ ($i\in\{1,2,\ldots,m\}$) with simple poles at $x_{i}\in \ve p^{\half \Z} q^{\half \ve'}$ for all $i\in\{1,2,\ldots,m\}$ and $\ve,\ve'\in\{\pm\}$,
\newline(2) $A^{0}(x ; a \lvert q , t )$ is $W_{m}$-invariant in the $x$-variables, and
\newline(3) $A^{0}(x ; a \lvert q , t )$ satisfies the residue conditions \begin{multline}
	\text{Res}\bigl( A^{0}(x ; a \lvert q , t ) \frac{d x_{i}}{x_{i}} \lvert x_i = \varpm p^{\half \ell} q^{\half \ve }\bigr) \\ 
	= -  \Bigl( \frac{p^{2}q^{2} t^{2} }{a_{0} \cdots a_{7} t^{2m} }\Bigr)^{\half \ve \ell } \text{Res}\bigl( A_{i}^{-\ve}(x ; a \lvert q , t ) \frac{d x_{i}}{x_{i}} \lvert x_{i} = \varpm p^{\half\ell} q^{\half\ve}\bigr)
	\label{eq_zeroth_order_residue_condition}
	\end{multline}
	for all $ i \in \{ 1 ,2,\ldots, m\}$, $\ve \in \{\pm\}$, and $\ell \in \Z$. (Here and below, we write `$\varpm$' to indicate that this relation holds for both signs.)

In the $m=1$ case, the coefficients $A_{1}^{\pm}(x ;a \lvert q , t)$ do not actually depend on the parameter $t$. The residue conditions then show that there is no essential dependence on $t$ in the zeroth order coefficient. Namely, the difference $A^{0}(x;a\lvert q , t) - A^{0}(x;a \lvert q, t^\prime)$ is a constant with respect to $x$ for different $t, t^\prime$.
This constant can be computed by straightforward calculations using the three-term relation (c.f. Eq.~\eqref{eq_Weierstrass_relation}).
\end{remark}
\subsection{On the symmetries of the zeroth order coefficient}
Using that the zeroth order coefficient is characterized by the conditions in Remark~\ref{remark_zeroth_order_coefficient}, it follows that if any other function ${B}^{0}(x ; a \lvert q , t)$ satisfies the same conditions, and does not have any other poles, then there exists a constant $C\in\C$ such that $B^{0}(x ; a \lvert q , t ) = A^{0}(x ; a \lvert q , t ) + C.$ Using this as our (informal) definition makes it straightforward to find the symmetries of the zeroth order coefficient:
\begin{lemma}
	\label{lemma_zeroth_order_symmetries}
	The zeroth order coefficient $A^{0}(x ; a \lvert q , t)$ satisfies
	\begin{subequations}
	\begin{eqnarray}
	A^{0}( x ; p^{\epsilon_{K}} a \lvert q , t) &=& q^{\half \abs{K} } (\prod_{s\in K} a_{s}^{-1}) A^{0}( x ; a \lvert q , t), \label{eq_zeroth_order_symmetries_1}\\
	A^{0}(p^{\half} x ; a \lvert q , t) &=& p^{-1} q^{-2} (a_{0} \cdots a_{7})^{\half} A^{0}(x ; p^{\half} a \lvert q , t),\label{eq_zeroth_order_symmetries_2} \\
		A^{0}( x ; a_{K(pq)} \lvert q , t) &=& A^{0}( x ; a \lvert q , t),\label{eq_zeroth_order_symmetries_3} \\
			A^{0}(x ; a \lvert q , p q t^{-1} ) &=& A^{0}( x ; a \lvert q , t), \label{eq_zeroth_order_symmetries_4}
	\end{eqnarray}
	\end{subequations}
	for any $K\subseteq \{0,1,\ldots, 7\}$ with even cardinality.
	\end{lemma}
	\begin{proof}
	It is clear that all of these are elliptic functions and it is straightforward to check that the l.h.s. of \eqref{eq_zeroth_order_symmetries_1}-\eqref{eq_zeroth_order_symmetries_4} are invariant under permutations and that only \eqref{eq_zeroth_order_symmetries_2} is not obviously invariant under inversions. Using that the zeroth order coefficients are elliptic functions in the $x$-variables yields that
	\begin{multline}
	A^{0}((p^{\half} x_1, p^{\half} x_2, \ldots ,p^{\half} x_{m}) ; a \lvert q , t ) = A^{0}( (p^{-\half} x_1^{-1},\ p^{ \half} x_2 , \ldots ,p^{\half} x_{m}) ; a \lvert q , t ) \\
	= A^{0}(  ( p p^{-\half} x_1^{-1}, p^{\half} x_2 , \ldots ,p^{\half} x_{m}) ; a \lvert q , t  ) = A^{0}( (p^{\half} x_1^{-1}, p^{\half} x_{2} , \ldots ,p^{\half} x_{m}) ; a \lvert q , t  ),
	\end{multline}
showing that all of these are $W_{m}$-invariant in the $x$-variables by permutation invariance. It is straightforward to check that they only have poles at $x_{i} = \varpm p^{ \half \ell } q^{ - \half \ve }$ ($i\in \{1,2,\ldots,m\}$, $\ve \in \{\pm\}$, and $\ell \in \Z$), since the half-period shift in \eqref{eq_zeroth_order_symmetries_2} only changes $\ell \to \ell+1$. It only remains to check the residue conditions: It follows from straightforward calculations that
\begin{subequations}
\begin{eqnarray}
A^{\ve}_{i}(x ; p^{\epsilon_K} a \lvert q , t ) &=& (\prod_{s\in K} -a_{s}^{-1} x_{i}^{-\ve}) A^{\ve}_{i}(x ; a \lvert q , t), \label{eq_coefficient_symmetries_1}\\
A^{\ve}_{i}( p^{\half} x  ; a \lvert q , t)  &=& \bigl(p^{-2} a_{0} \cdots a_{7} \bigr)^{\half (1-\ve)} ( q t^{1-m} x_{i}^{4})^{\ve} A^{\ve}_{i}(x ; p^{\half} a \lvert q , t ), \label{eq_coefficient_symmetries_2}\\
A^{\ve}_{i}( x ; a_{K(pq)} \lvert q , t) &=& (\prod_{s\in K} \theta( q^{-1} a_{s} x_{i}^{-\ve};p) / \theta( a_{s} x_{i}^{\ve} ; p ) ) A^{\ve}_{i}(x; a \lvert q , t), \label{eq_coefficient_symmetries_3}\\
A^{\ve}_{i}( x ; a \lvert q, pqt^{-1}) &=& ( \prod_{j \neq i}^{m} \theta( q^{-1} t x_{i}^{-\ve} x_{j}^{\pm};p) / \theta( t x_{i}^{\ve} x_{j}^{\pm};p) ) A^{\ve}_{i}(x; a \lvert q, t) \label{eq_coefficient_symmetries_4}
\end{eqnarray}
\end{subequations}
for all $i \in\{1,2,\ldots,m\}$, $K\subseteq\{0,1,\ldots,7\}$, and $\ve \in\{\pm\}$. Using \eqref{eq_coefficient_symmetries_1}, \eqref{eq_coefficient_symmetries_3}, and \eqref{eq_coefficient_symmetries_4}, it is straightforward to check that
\begin{multline}
\text{Res}( A^{0}(x ; p^{\epsilon_{K}} a \lvert q , t) \frac{d x_{i} }{ x_{i} }  \lvert x_{i} = \varpm p^{\half \ell } q^{- \half\ve }) \\
= (\prod_{s\in K} - a_{s} q^{\half} ) \text{Res}(A^{(0)}(x ; a \lvert q , t) \frac{d x_{i} }{x_{i} } \lvert x_{i} = \varpm p^{\half\ell } q^{- \half\ve } ),
\end{multline}
\begin{multline}
\text{Res}( A^{0}(x ; a_{K(p q)} \lvert q , t) \frac{d x_{i} }{ x_{i} }  \lvert x_{i} = \varpm p^{\half \ell } q^{-\half\ve }) \\
= \text{Res}( A^{0}(x ; a \lvert q , t) \frac{d x_{i} }{ x_{i} }  \lvert x_{i} = \varpm p^{\half \ell } q^{-\half \ve }),
\end{multline}
and
\begin{multline}
\text{Res}( A^{0}(x ; a \lvert q , p q t^{-1}) \frac{d x_{i}}{x_{i}} \lvert x_{i} = \varpm p^{\half\ell } q^{- \half \ve })  \\
= \text{Res}( A^{0}(x ; a \lvert q , t ) \frac{d x_{i}}{x_{i}} \lvert x_{i} = \varpm p^{\half\ell } q^{- \half\ve }),
\end{multline}
for all $i\in\{1,2,\ldots,m\}$, $\ve \in \{\pm\}$, and $\ell \in \Z$, by using 
\begin{equation}
\frac{\theta( \varpm p^{- \half \ve \ell } q^{-\half} a_{s} ; p)}{\theta( \varpm p^{+\half \ve \ell } q^{-\half} a_{s};p)} = (-1)^{\ell} ( p q a_{s}^{-2})^{- \half \ve \ell } \quad (s\in\{0,1,\ldots,7\}, \ \ve \in \{\pm\}, \ \ell\in\Z)
\end{equation}
and
\begin{equation}
\frac{\theta( \varpm p^{-\half \ve \ell } q^{-\half} t x_{j}^{\pm} ; p)}{\theta( \varpm p^{+\half \ve \ell } q^{-\half} t x_{j}^{\pm} ; p)}  = (p q t^{-2})^{\ve 2 \ell}  \quad (j\in\{1,2,\ldots,m\}, \ \ve \in \{\pm\}, \ \ell\in\Z),
\end{equation}
which reduces the prefactors to $1$. 
For \eqref{eq_zeroth_order_symmetries_2}, we need to take into account the half-period shift when considering the residue conditions which yields that
\begin{multline}
\text{Res}(A^{0}(p^{\half} x ; a \lvert q , t ) \frac{dx_{i}}{x_{i} }\lvert x_{i} = \varpm p^{\half \ell} q^{-\half \ve}) \\
= - \Bigl( \frac{p^{2} q^{2} t^{2}}{a_{0} \cdots a_{7} t^{2m}} \Bigr)^{-\half \ve (\ell+1)} \text{Res}( A^{\ve}_{i}( p^{\half} x; a \lvert q , t) \frac{d x_{i}}{x_{i}} \lvert x_{i} = \varpm p^{\half \ell} q^{-\half \ve }).
\end{multline}
Using \eqref{eq_zeroth_order_residue_condition} and \eqref{eq_coefficient_symmetries_2} yields that 
\begin{multline}
\text{Res}(A^{0}(p^{\half} x ; a \lvert q , t ) \frac{dx_{i}}{x_{i} }\lvert x_{i} = \varpm p^{\half \ell} q^{-\half \ve}) \\
= C^{(1)} \text{Res}(A^{0}(x ; p^{\half}  a \lvert q , t ) \frac{dx_{i}}{x_{i} }\lvert x_{i} = \varpm p^{\half \ell} q^{-\half \ve}) 
\end{multline}
where
\begin{eqnarray}
C^{(1)} &=& \Bigl( \frac{p^{2} q^{2} t^{2}}{a_{0} \cdots a_{7} t^{2m}}\Bigr)^{-\half \ve (\ell + 1 )} (p^{-2} a_{0}\cdots a_{7})^{\half (1-\ve)} \\
&&\cdot (q t^{1-m} (p^{\half \ell} q^{-\half\ve})^{4})^{\ve} \Bigl( \frac{p^{2} q^{2} t^{2}}{p^{4} a_{0}\cdots a_{7} t^{2m}} \Bigr)^{\half\ve \ell } \nonumber \\
&=& p^{-1} q^{-2} (a_{0}\cdots a_{7})^{\half}. \nonumber
\end{eqnarray}

Next, we determine the additive constants. For Eqs. \eqref{eq_zeroth_order_symmetries_1}, \eqref{eq_zeroth_order_symmetries_3}, and \eqref{eq_zeroth_order_symmetries_4} it is straightforward to check that 
\begin{align}
A^{0}_{r}(x ; p^{\epsilon_{K}} a \lvert q ,t ) &= q^{\half \abs{K}} (\prod_{s\in K} a_{s}^{-1} ) A^0_{r}(x ; a \lvert q , t),\\
A^{0}_{r}(x ; a_{K(pq)} \lvert q, t ) &= A^{0}_{r}(x ; a \lvert q , t), \\
A^{0}_{r}(x ; a \lvert q , p q t^{-1}) &= A_{r}^{0}(x ; a \lvert q, t)
\end{align}
for all $r\in\{0,1, 2, 3\}$ which gives that the additive constant $C$ is 0 in all these cases. For \eqref{eq_zeroth_order_symmetries_2}, we set $x_{j} = \alpha t^{m-j}$ for all $j\in \{1,2,\ldots,m\}$ which yields that
\begin{equation}
\prod_{1 \leq j \leq m} \frac{\theta( c_{r} q^{-\half} t x_{j}^{\pm} ; p) }{\theta( c_{r} q^{-\half} x_{j}^{\pm} ; p )} = \frac{\theta( c_{r} q^{-\half} t^{m} \alpha , c_{r} q^{-\half} t \alpha^{-1} ;p)}{ \theta( c_{r} q^{-\half} \alpha , c_{r} q^{-\half} t^{-m+1} \alpha^{-1} ; p )}
\end{equation}
for any $\alpha \in \C^{\ast}$. By specializing $\alpha \to p^{-\half}$, we obtain that
\begin{equation}
A^{0}_{0}(p^{\half} x ; a\lvert q ,t) \to \frac{\prod_{0 \leq s \leq 7} \theta( q^{-\half} a_{s};p) }{\theta( t ,  q^{-1} t ; p)}\frac{\theta( q^{-\half} t , q^{-\half} t^{m} ;p )}{\theta( q^{-\half} , q^{-\half} t^{-m+1} ;p )}
\end{equation}
and 
\begin{multline}
p^{-1} q^{-2} (a_{0} \cdots a_{7})^{\half} A_{2}^{0}(x ; p^{\half}a \lvert q, t) 
 \to \frac{\prod_{ 0 \leq s \leq 7} \theta( q^{-\half} a_{s};p) }{\theta( t ,  q^{-1} t ; p)}\frac{\theta( q^{-\half} t , q^{-\half} t^{m} ;p )}{\theta( q^{-\half} , q^{-\half} t^{-m+1} ;p )}
\end{multline}
by using the $p$-difference equation satisfied by the multiplicative theta functions. We have thus obtained that $A^{0}_{0}(p^{\half} x ; a\lvert q ,t) - p^{-1} q^{-2} (a_{0} \cdots a_{7})^{\half} A_{2}^{0}(x ; p^{\half}a \lvert q, t) \to 0
$ as $x_{j}\to p^{-\half} t^{m-j}$ for all $j\in\{1,2,\ldots,m\}$ (and congruent points by $p$-shifts as $A_{r}^{0}(x;a\lvert q ,t)$ ($r\in\{0,1,2,3\}$) are elliptic functions). The same specialization also yields that
\begin{align}
A_{1}^{0}(p^{\half} x ; a\lvert q, t) - p^{-1} q^{-2} (a_{0} \cdots a_{7})^{\half} A_{3}^{0}(x ; p^{\half}a \lvert q, t) \to 0, \\
A_{2}^{0}(p^{\half} x ; a\lvert q, t) - p^{-1} q^{-2} (a_{0} \cdots a_{7})^{\half} A_{0}^{0}(x ; p^{\half}a \lvert q, t) \to 0, \\
A_{3}^{0}(p^{\half} x ; a\lvert q, t) - p^{-1} q^{-2} (a_{0} \cdots a_{7})^{\half} A_{1}^{0}(x ; p^{\half}a \lvert q, t) \to 0,
\end{align}
by straightforward computations, and we obtain the desired result in \eqref{eq_zeroth_order_symmetries_2}.
\end{proof}

\subsection{Symmetries of the van Diejen operator.}
\label{eq_symmetries_proof}

In the introduction, we presented how the van Diejen operator changes under gauge transformations as well as under shifts of the model parameters by factors of $p$ and $p^{\half}$. We are now in a position to prove these statements.

\subsubsection{Proofs of gauge symmetries.}
Let us start by giving the proofs of Lemmas~\ref{lemma_gauge_transform_V} and \ref{lemma_gauge_transform_U}, as well as Propositions~\ref{prop_gauge_transf_V} and \ref{prop_gauge_transf_U}.
\begin{proof}[Proof of Lemma~\ref{lemma_gauge_transform_V}]
The gauge function $V(x\lvert q , t)$ satisfies the $q$-difference equations
\begin{multline}
V(x\lvert q, t)^{-1} T_{q,x_{i}}^{\ve} V(x\lvert q, t) = \prod_{ j > i} \frac{\theta(q^{-1} t x_{i}^{-\ve} x_{j}^{\pm};p) }{\theta( t x_{i}^{\ve} x_{j}^{\pm};p)} \prod_{ j < i} \frac{\theta( q^{-1} t x_{j}^{\pm} x_{i}^{-\ve} ; p) }{\theta( t x_{j}^{\pm} x_{i}^{\ve} ; p )} \\
= \prod_{j \neq i } \frac{\theta(q^{-1} t x_{i}^{-\ve} x_{j}^{\pm};p)}{\theta( t x_{i}^{\ve} x_{j}^{\pm} ; p )} = \prod_{j \neq i } \frac{\theta(p q t^{-1} x_{i}^{\ve} x_{j}^{\pm};p)}{\theta( t x_{i}^{\ve} x_{j}^{\pm} ; p )}
\end{multline}
 for all $i\in\{1 ,2,\ldots,  m \}$ and $\ve\in\{\pm\}$, where the last equality is obtained by using the reflection property of the multiplicative theta function.
Recalling the coefficients $A^{\pm}_{i}( x ; a \lvert q , t)$ in \eqref{eq_coefficients_positive_multiplicative}, we find that $
A^{\ve}_{i}( x ; a \lvert q , t) V(x\lvert q, t)^{-1} T_{q,x_{i}}^{\ve} V(x\lvert q, t) 
= A_{i}^{\ve}( x ; a \lvert q , p q t^{-1}),
$ 
which yields that 
\begin{equation}
A^{\ve}_{i}( x ; a \lvert q , t) V(x\lvert q, t)^{-1} \circ T_{q,x_{i}}^{\ve} \circ V(x\lvert q, t) = A^{\ve}_{i}( x ; a \lvert q , p q t^{-1} ) T_{q,x_{i}}^{\ve} 
\end{equation}
for all $i\in\{1,2,\ldots,m\}$ and $ \ve\in\{\pm\}$. Using the symmetries of the zeroth order coefficient in \eqref{eq_zeroth_order_symmetries_4} yields that
\begin{multline}
V(x\lvert q, t)^{-1} \circ \cD_{x}(a \lvert q , t) \circ V(x\lvert q, t) = A^{0}(x ; a \lvert  q , t) \\
+ \sum_{ 1 \leq i \leq m} \sum_{\ve=\pm} A_{i}^{\ve}(x; a \lvert q, t ) V(x\lvert q, t)^{-1} \circ T_{q,x_{i}}^{\ve} \circ V(x\lvert q, t) \\
= A^{0}(x ; a \lvert q , p q t^{-1} ) + \sum_{ 1 \leq i \leq m} \sum_{\ve=\pm} A_{i}^{\ve}(x ; a \lvert q , p q t^{-1}) T_{q,x_{i}}^{\ve} = \cD_{x}(a \lvert q , p q t^{-1} )
\end{multline}
which is the desired result.
\end{proof}

\begin{proof}[Proof of Lemma~\ref{lemma_gauge_transform_U}]
The gauge function $U_{K}(x; a \lvert q )$ satisfies the $q$-difference equation
\begin{equation}
U_{K}(x;a\lvert q )^{-1} T_{q,x_{i}}^{\ve} U_{K}(x; a \lvert q ) = \prod_{s\in K} \frac{\theta( q^{-1} a_{s} x_{i}^{-\ve} ; p )}{\theta( a_{s} x_{i}^{\ve} ; p) } = \prod_{s\in K} \frac{\theta( p q a_{s}^{-1} x_{i}^{\ve} ; p )}{\theta( a_{s} x_{i}^{\ve} ; p) } 
\end{equation}
for all $i \in \{1,2,\ldots, m\}$ and $\ve\in\{\pm\}$. Recalling the coefficients $A^{\pm}_{i}( x ; a \lvert q , t)$ in \eqref{eq_coefficients_positive_multiplicative}, it becomes clear that 
\begin{multline}
U_{K}(x ; a \lvert q )^{-1} \circ A^{\ve}_{i}(x ; a \lvert q , t) T_{q, x_{i}}^{\ve} \circ U_{K}(x; a \lvert q ) \\ 
= A^{\ve}_{i}(x ; a \lvert q , t)  \bigl( U_{K}(x ; a \lvert q )^{-1} T_{q, x_{i}}^{\ve} U_{K}(x; a \lvert q )  \bigr) T_{q,x_{i}}^{\ve}\\
= A^{\ve}_{i}( x ; a_{K(p q )} \lvert q, t ) T_{q,x_{i}}^{\ve} ,
\end{multline}
with $a_{K(pq)}$ as in Lemma~\ref{lemma_gauge_transform_U}. Combining this relation with the symmetries of the zeroth order coefficient in Eq. \eqref{eq_zeroth_order_symmetries_3} yields the desired result.
\end{proof}
\begin{proof}[Proof of Proposition~\ref{prop_gauge_transf_V}]
We consider the action of the van Diejen operator \newline $\cD_{x}(a \lvert q , t)$ on the function $\psi(x)$ \eqref{eq_prop_eigenfunction_V} and obtain that
\begin{multline*}
\cD_{x}(a \lvert q , t ) \psi(x) = \cD_{x}(a \lvert q , t ) V(x \lvert q,t ) \varphi(x) \\
= V(x \lvert q , t ) \bigl( V(x\lvert q, t )^{-1} \circ \cD_{x}(a \lvert q, t ) \circ V(x\lvert q, t) \bigr) \varphi(x)
\end{multline*}
by straightforward calculations. Using the result of Lemma~\ref{lemma_gauge_transform_V} simplifies the equation into the form 
\begin{equation}
\cD_{x}(a \lvert q , t ) \psi(x) = V(x \lvert q , t ) \cD_{x}(a \lvert q, p q t^{-1} )  \varphi(x) 
\end{equation}
which shows that $\psi(x)$ is an eigenfunction of the van Diejen operator $\cD_{x}(a \lvert q , t)$ if $\varphi(x)$ is an eigenfunction of the van Diejen operator $\cD_{x}(a \lvert q , p q t^{-1})$ and concludes the first part of the proof. Setting $\varphi(x)=1$ in the equation above yields that $\psi(x) = V(x \lvert q , t)$ and we have the $\cD_{x}(a \lvert  q , t) V(x \lvert q , t) = V(x \lvert q , t) \cD_{x}(a \lvert q, p q t^{-1}).1.$ Replacing $t \to p q t^{-1}$ in the ellipticity condition \eqref{eq_ellipticity_condition_1} yields that $1$ is an eigenfunction of $\cD_{x}(a \lvert q , p q t^{-1})$ if the parameters satisfy the condition in \eqref{eq_ellipticity_condition_V} and the eigenvalue is obtained from Eq. \eqref{eq_vD_constant_eigenvalues} by replacing $t \to p q t^{-1}$, \emph{i.e.} $\Lambda_{0}(a \lvert q , p q t^{-1})$. 
\end{proof}
\begin{proof}[Proof of Proposition~\ref{prop_gauge_transf_U}]
We consider the action of the operator $\cD_{x}(a \lvert q, t)$ on $\psi(x)$ \eqref{eq_prop_eigenfunction_U} and obtain that
\begin{multline*}
\cD_{x}(a \lvert q , t) \psi(x) =  \cD_{x}( a \lvert q , t) U_{K}(x ; a\lvert q)\varphi(x) \\
 = U_{K}(x ;a \lvert q ) \bigl( U_{K}(x;a\lvert q )^{-1} \circ \cD_{x}(a\lvert q, t) \circ U_{K}(x ; a \lvert q ) \bigr) \varphi(x) 
\end{multline*}
by straightforward calculations. Using Lemma~\ref{lemma_gauge_transform_U} yields that
\begin{equation}
\cD_{x}(a \lvert q, t) \psi(x)  = U_{K}(x ;a \lvert q ) \cD_{x}(a_{K(pq)} \lvert q , t) \varphi(x) 
\end{equation}
which shows that $\psi(x)$ is an eigenfunction of the van Diejen operator $\cD_{x}(a\lvert q , t)$ if $\varphi(x)$ is an eigenfunction of the van Diejen operator $\cD_{x}(a_{K(pq)} \lvert q , t)$ and concludes the first part of the proof. Setting $\varphi(x) = 1$ in the equation yields that $\psi(x) = U_{K}(x ; a \lvert q )$ and we have that 
\begin{equation}
\cD_{x}(a \lvert q ,t ) U_{K}(x;a \lvert q)  = U_{K}(x ;a \lvert q ) \cD_{x}(a_{K(pq)} \lvert q , t) .1 .
\end{equation}
The ellipticity condition for the operator $\cD_{x}(a_{K(pq)} \lvert q , t)$ is given by
\begin{equation}
\label{eq_ellipticity_condition_gauge}
(\prod_{s\in K} p q a_{s}^{-1} )( \prod_{s\notin K} a_{s} ) t^{2m} = p^{2} q^{2} t^{2}
\end{equation}
and we obtain the condition \eqref{eq_balancing_condition_U} by multiplying both sides of \eqref{eq_ellipticity_condition_gauge} by $\prod_{s\in K} a_{s}$. Thus, we have that $\cD_{x}(a_{K(pq)} \lvert q , t).1 = 1. \Lambda_{0}(a_{K(pq)} \lvert q , t) $ if the model parameters satisfy \eqref{eq_balancing_condition_U} and we obtain the desired result.
\end{proof}

\subsubsection{$p$-shifts in the model parameters.}
\begin{lemma}
\label{lemma_parameter_p_shift}
Let $K\subseteq \{0 ,1, \ldots , 7\}$ with even cardinality and $g(x\lvert q)$ satisfy the $q$-difference equations
\begin{equation}
T_{q,x_{i}} g(x\lvert q ) = (q x_{i}^{2})^{\half \abs{K}} g(x\lvert q) \quad (i\in\{1,2,\ldots,m\}).  
\end{equation}
Then
\begin{equation}
g(x \lvert q ) = [\mbox{\emph{quasi-const.}}] \times \prod_{1 \leq i \leq m} \theta( c x_{i}^{\pm};q )^{\half \abs{K}},
\end{equation}
for any $c\in\C^{\ast}$, and the van Diejen operator satisfies
\begin{equation}
q^{-\half \abs{K}} (\prod_{s\in K} a_{s} ) \cD_{x}(p^{\epsilon_{K}} a \lvert q , t) = g(x\lvert q)^{-1} \circ \cD_{x}(a \lvert q , t) \circ g(x \lvert q ).
\end{equation}
Furthermore, if the parameters satisfy
\begin{equation}
\label{eq_ellipticity_condition_2}
a_{0} \cdots a_{7} p^{\abs{K}} t^{2m} = p^{2} q^{2} t^{2}
\end{equation}
then the function $g(x \lvert q )$ is an eigenfunction of the van Diejen operator $\cD_{x}(a \lvert q , t)$ with eigenvalue
\begin{equation}
(q p t^{-1})^{- \half \abs{K}} \Lambda_{0}(a \lvert q , t)
\end{equation}
\end{lemma}
\begin{proof}
Let $g(x) = g(x\lvert q )$ to simplify notation. It follows from straightforward calculations that $\prod_{1 \leq i \leq m} \theta( c x_{i}^{\pm};q )^{\half \abs{K}}$ satisfies the $q$-difference equation. Comparing the $q$-difference equation satisfied by $g(x)$ with Eq.~\eqref{eq_coefficient_symmetries_1} yields that
\begin{multline}
q^{- \half \abs{K} } (\prod_{s\in K}a_{s})  A_{i}^{\ve}(x ; p^{\epsilon_{K}} a \lvert q , t) T_{q,x_{i}}^{\ve} 
= A_{i}^{\ve}(x;a \lvert q, t) g(x)^{-1}   ({T_{q,x_{i}}^{\ve} g(x)}) \ \! T_{q,x_{i}}^{\ve}   \\
= A_{i}^{\ve}(x;a \lvert q, t) g(x)^{-1} \circ  {T_{q,x_{i}}^{\ve} \circ g(x)} 
\end{multline}
for all $i\in\{1,2,\ldots,m\}$ and $\ve\in\{\pm\}$. Using this and \eqref{eq_zeroth_order_symmetries_2} yields
\begin{equation}
q^{- \half \abs{K}} (\prod_{s\in K} a_{s})\cD_{x}(p^{\epsilon_{K}} a \lvert q , t) = g(x)^{-1} \circ \cD_{x}(a \lvert q , t) \circ g(x).
\end{equation}
which concludes the first part of the proof.

If the parameters satisfy the condition in \eqref{eq_ellipticity_condition_2}, then any constant is an eigenfunction of the operator $\cD_{x}(p^{\epsilon_{K}} a \lvert q , t)$, \emph{i.e.}
\begin{equation}
\cD_{x}( p^{\epsilon_K} a \lvert q , t)\Bigl.\Bigr|_{a_{0}\cdots a_{7} p^{\abs{K}} t^{2m} = p^{2} q^{2} t^{2}}  . 1 = 1 . \Lambda_{0}(p^{\epsilon_{K}} a \lvert q , t )
\end{equation}
and we have that 
\begin{equation}
g(x)^{-1} \circ \cD_{x}(a \lvert q , t) \circ g(x) . 1 = 1. q^{- \half \abs{K}} (\prod_{s\in K} a_{s}) \Lambda_{0}(p^{\epsilon_{K}} a \lvert q , t)
\label{eq_g_eigenvalue_eq}
\end{equation}
if the parameters satisfy \eqref{eq_ellipticity_condition_2}. Furthermore, we have that 
\begin{equation}
\cD_{x}(a \lvert q , t) \circ g(x) . 1 = A^{0}(x ; a \lvert q , t) g(x) +  \sum_{ 1 \leq i \leq m}\sum_{\ve=\pm} A^{\ve}_{i}(x ; a \lvert q , t) T_{q,x_{i}}^{\ve} g(x)
\end{equation}
which yields that $\cD_{x}(a \lvert q , t) \circ g(x) . 1$ is the same as the operator $\cD_{x}(a\lvert q , t )$ acting on $g(x)$. Multiplying both sides of \eqref{eq_g_eigenvalue_eq} with $g(x)$ yields the eigenvalue equation 
\begin{equation}
\cD_{x}(a \lvert q , t)  g(x)  = q^{- \half\abs{K}} (\prod_{s\in K} a_{s}) \Lambda_{0}(p^{\epsilon_{K}} a \lvert q , t) \ g(x),
\end{equation}
if the parameters satisfy \eqref{eq_ellipticity_condition_2}. Using \eqref{eq_vD_constant_eigenvalues} and \eqref{eq_vD_external_potential_1}, we find that
\begin{equation}
\Lambda_{0}(p^{\epsilon_{K}} a \lvert q , t) = p^{-\half \abs{K}} t^{\half \abs{K}} (\prod_{s\in K} a_{s}^{-1}  ) \Lambda_{0}(a \lvert q , t)
\end{equation}
which simplifies the eigenvalues to
\begin{equation}
q^{-\half \abs{K}} (\prod_{s\in K} a_{s} ) \Lambda_{0}(p^{\epsilon_{K}} a \lvert q , t) = (q p t^{-1})^{-\half \abs{K}} \Lambda_{0}(a \lvert q , t)
\end{equation}
and concludes the proof.
\end{proof}
\begin{lemma}
\label{lemma_half-period_shift}

Let $f(x; a \lvert q , t)$ satisfy the $q$-difference equations
\begin{equation}
T_{q,x_{i}} f(x;a \lvert q, t) = p q^{3} t^{-(m-1)}( a_{0} \cdots a_{7})^{-\half} x_{i}^{4} f(x ; a \lvert q , t) \quad (i\in\{1,2,\ldots,m\}).
\end{equation}
Then 
\begin{equation}
f(x;a\lvert q , t) = [\mbox{\emph{quasi-const.}}] \times \prod_{1 \leq i \leq m} \theta( d^{\half} x_{i}^{2} ; q)^{-1},
\end{equation}
where $d= p q^2 t^{1-m} (a_{0} \cdots a_{7})^{-\half}$ and the principal van Diejen operator satisfies 
\begin{equation}
p q^{2} (a_{0} \cdots a_{7})^{-\half} \cD_{p^{\half} x}(a \lvert q , t ) = f(x;a\lvert q, t )^{-1} \circ \cD_{x}(p^{\half} a \lvert q , t) \circ f(x;a\lvert q , t).
\end{equation}
Furthermore, if the parameters satisfy the ellipticity condition in Eq.~\eqref{eq_ellipticity_condition_1}, then $f(x;a \lvert q , t) $ is an eigenfunction of the van Diejen operator $\cD_{x}( p^{\half} a \lvert q , t )$ with eigenvalue 
\begin{equation}
p q^{2} (a_{0} \cdots a_{7})^{-\half} \Lambda_{0}(a \lvert q , t).
\end{equation}
\end{lemma}
\begin{proof}
Let $f(x) = f(x;a\lvert q ,t )$ to simplify notation. We note that $f(x)$ also satisfies the $q$-difference equations
\begin{equation}
T_{q,x_{i}}^{-1} f(x) = \frac{ \theta( d^{\half} x_{i}^{2} ; q)}{ \theta(  q^{-2} d^{\half} x_{i}^{2} ; q)} f(x) = p^{-1} q t^{m-1} (a_{0}\cdots a_{7})^{\half} x_{i}^{-4} f(x) 
\end{equation}
for all $i \in \{1,2,\ldots,m\}$. From the $q$-difference equation for $f(x)$ and Eq. \eqref{eq_coefficient_symmetries_2} it follows that 
\begin{equation}
\label{eq_f_transform}
A^{\ve}_{i}(p^{\half} x ; a \lvert q , t) T_{q,x_{i}}^{\ve} = p^{-1}q^{-2} (a_{0} \cdots a_{7})^{\half} f(x)^{-1} \circ A^{\ve}_{i}(x ; p^{\half} a \lvert q , t) T_{q,x_{i}}^{\ve} \circ f(x).
\end{equation}
Using \eqref{eq_f_transform} and \eqref{eq_zeroth_order_symmetries_2} yields that
\begin{multline}
f(x)^{-1} \circ \cD_{x}(p^{\half} a \lvert q , t) \circ f(x) = 
\sum_{ 1 \leq i \leq m} \sum_{\ve = \pm} f(x)^{-1} \circ A^{\ve}_{i}(x ; p^{\half} a \lvert q , t) T_{q,x_{i}}^{\ve} \circ f(x)\\
+ A^{0}(x ; p^{\half} a \lvert q , t )  = p q^{2} (a_{0} \cdots a_{7} )^{-\half} \bigl( A^{0}(p^{\half} x ; a \lvert q , t ) \bigr. \\
+ \bigl. \sum_{ 1 \leq i \leq m} \sum_{\ve = \pm} A^{\ve}_{i}(p^{\half} x ; a \lvert q , t ) T_{q,x_{i}}^{\ve}\bigr)
= p q^{2} (a_{0} \cdots a_{7} )^{-\half} \cD_{p^{\half} x }(a \lvert q , t)
\end{multline}
and concludes the first part of Lemma~\ref{lemma_half-period_shift}.

If the parameters satisfy the ellipticity condition \eqref{eq_ellipticity_condition_1}, then
\begin{equation}
\cD_{p^{\half} x}(a \lvert q , t)\Bigl.\Bigr|_{a_{0}\cdots a_{7} t^{2m} = p^{2} q^{2} t^{2}} . 1 = 1. \Lambda_{0}(a \lvert q, t),
\end{equation}
since the half-period shift does not affect the constant eigenfunction, and we obtain that 
\begin{multline}
f(x)^{-1} \circ \cD_{x}(p^{\half} a \lvert q ,t )\circ f(x)\Bigl.\Bigr|_{a_{0}\cdots a_{7} t^{2m} = p^{2} q^{2} t^{2}}.1\\
 = p q^{2} (a_{0} \cdots a_{7})^{-\half} \cD_{p^{\half} x }(a \lvert q , t)\Bigl.\Bigr|_{a_{0}\cdots a_{7} t^{2m} = p^{2} q^{2} t^{2}} . 1 \\
= 1. p q^{2} (a_{0} \cdots a_{7})^{-\half} \Lambda_{0}(a \lvert q , t).
\end{multline}
It then follows that $\cD_{x}(p^{\half} a \lvert q , t) f(x) = pq^{2}(a_{0}\cdots a_{7})^{-\half} \Lambda_{0}(a \lvert q, t ) \ \! f(x),$
if the parameters satisfy \eqref{eq_ellipticity_condition_1}, since the conjugation by $f(x)$ reduces to an eigenvalue equation when acting on the constant eigenfunction. This concludes the second part of Lemma~\ref{lemma_half-period_shift}.
\end{proof}

\begin{remark}
Note that $f(x; a \lvert q ,t)$ is not a $W_{m}$-invariant function since shifts in the variables break $W_{m}$-invariance in general, but if the parameters satisfy the ellipticity condition \eqref{eq_ellipticity_condition_1}  then
\begin{equation}
f(x;a \lvert q , t)\Bigl.\Bigr|_{a_{0}\cdots a_{7} t^{2m} = p^{2} q^{2} t^{2}} = \prod_{ 1 \leq i \leq m} \theta( q^{\half} x_{i}^{2} ; q)^{-1} = \prod_{ 1 \leq i \leq m} \theta( q^{\half} x_{i}^{-2} ; q)^{-1},
\end{equation}
and we have that the van Diejen operator $\cD_{p^{\half} x}(a \lvert q , t)$ is $W_{m}$-invariant as well.
\end{remark}

\subsection{Symmetric $\C$-bilinear form.}
\label{sec_scalar_product}

As we discussed, the van Diejen operator is formally self-adjoint with regard to the symmetric $\C$-bilinear form in \eqref{eq_inner_product} with weight function is given in \eqref{eq_weight_function}. In Sections~\ref{sec_intro} and \ref{sec_main_results}, we set the cycle of the integrals to the $n$-dimensional torus for simplicity by restricting the values of the model parameters. In this Section, we present the suitable integration cycles for the multivariate elliptic hypergeometric integrals, by considering the poles of the weight function, and the possible analytic continuation of the model parameters. (At this point, we find it convenient to change variables and parameters from $(x,a)$ to $(y,b)\in(\C^{\ast})^{n}\times(\C^{\ast})^{8}$ as it will make the results of the upcoming Sections more clear.) From the definition of the elliptic Gamma function (\emph{cf.} \eqref{eq_elliptic_Gamma_function}), we see that the weight function $w_{n}(y;b\lvert q , t)$ has poles at
\begin{align}
y_{k}^{\ve} \in b_{s} p^{\Z_ {\geq 0}} q^{\Z_{\geq 0}} & \quad (k \in \{1,2,\ldots,n \}; \ s\in\{0,1,\ldots,7\}; \ \ve\in\{\pm\}),\\
y_{k}^{\ve} y_{l}^{\ve^\prime} \in p^{\Z_{\geq 0}} q^{\Z_{\geq 0}} t & \quad ( 1 \leq k < l \leq n ; \ \ve,\ve^\prime \in \{\pm\}) \label{eq_pair_poles}
\end{align}
 from the elliptic Gamma functions in the numerator. The denominators in the weight function does not give rise to other poles as they are holomorphic in $(\C^{\ast})^{n}$: Using the reflection property in \eqref{eq_reflection_property} and the difference equations \eqref{eq_Gamma_difference_eq}, we find that 
\begin{equation}
\frac{ 1}{ \Gamma( y_{k}^{\pm 2} ; p , q )} = \frac{\Gamma( p q y_{k}^{2} ; p , q )}{\Gamma( y_{k}^{2} ; p , q)} = \theta( p y_{k}^{2} ; p) \theta( y_{k}^{2} ; q) \quad (k\in\{1,2,\ldots,n\}) ,
\end{equation}
for all $k\in\{1,2,\ldots,n \}$, and 
\begin{equation}
\frac{1}{\Gamma(y_{k}^{\pm} y_{l}^{\pm} ; p , q)} = \frac{\Gamma( p q y_{k} y_{l}^{\pm} ; p , q)}{\Gamma( y_{k} y_{l}^{\pm} ; p , q)} = \theta( p y_{k} y_{l}^{\pm};p) \theta(y_{k} y_{l}^{\pm};q) \quad (1 \leq k < l \leq n) ,
\end{equation}
for all $1 \leq k < l \leq n$.
Let us start by considering the $n=1$ case where the poles that accumulate towards zero, and those toward infinity, can be separated into
\begin{align}
S_{0} = S_{0}(b \lvert q) =& \bigl\{ b_{s} p^{\mu} q^{\nu} \lvert s\in\{0,1,\ldots , 7 \}; \ \mu , \nu \in\Z_{\geq0}\bigr\}, \\
S_{\infty} = S_{\infty}(b \lvert q ) =& \bigl\{ b_{s}^{-1} p^{-\mu} q^{-\mu} \lvert  s\in\{0,1,\ldots,7\}; \ \mu, \nu \in\Z_{\geq 0}\bigr\}.
\end{align}
(Recall that the van Diejen model does not depend on the parameter $t$ in the one-variable case.)
It is then possible to find a positively oriented, closed curve that separates the sets $S_{0}$ and $S_{\infty}$ if $S_{0} \cap S_{\infty} = \emptyset$. This is equivalent to the condition
\begin{equation}
b_{r} b_{s} \notin p^{\Z_{\leq 0}} q^{\Z_{\leq 0}}  \quad ( r , s \in \{0,1,\ldots,7\})
\label{eq_seperation_condition}
\end{equation}
for the model parameters. Define the circle \cite{NI17}
\begin{equation}
\mathcal{C}_{\rho}(0) = \{ u \in \C^{\ast} \lvert \ \abs{u} = \rho\}, \quad \rho\in [R , R^{-1}]
\end{equation}
for some $R\in (0 , 1]$ such that $\mathcal{C}_{\rho}(0)$ does not intersect $S_{0} \cup S_{\infty}$. The cycle is then defined by
\begin{equation}
\mathcal{C}_{1} = \mathcal{C}_{\rho}(0)  + \sum_{c \in S_{0} : \ \abs{c} > \rho} \mathcal{C}_{\tilde\epsilon}(c) - \sum_{c\in S_{\infty} : \ \abs{c} < \rho} \mathcal{C}_{\tilde\epsilon}(c),
\end{equation}
where $\mathcal{C}_{\tilde\epsilon}(c)$ denotes a sufficiently small circle around the point $c$. Since the weight function $w_{1}(y;b \lvert q ,t )$ is holomorphic around the cycle $\mathcal{C}_{1}$, we have that the integral
\begin{equation}
\int_{\mathcal{C}_{1}} \frac{d y_{1}}{2\pi \imag y_{1}}  \ w_{1}(y; b \lvert q , t)
\end{equation}
is well-defined and does not depend on the choice of $\rho\in[ R , R^{-1}]$. When considering the general $n \in \Z_{\geq 1}$ case, we can take $\abs{t} < R^{2}$ to obtain that $w_{n}(y ; b \lvert q ,t)$ is holomorphic around a neighbourhood of the $n$-cycle \cite{NI17,NI19}
\begin{equation}
\mathcal{C}_{n} =  \underbrace{\mathcal{C}_{1} \times \cdots \times \mathcal{C}_{1}}_{n \text{ times}}.
\label{eq_m_cycle}
\end{equation}
We will often write $\mathcal{C}_{n}(b \lvert q , t)$ for the $n$-cycles to indicate the positions of the poles.
If the parameters satisfy $\abs{b_{s}}< 1$ for all $s\in \{0,1,\ldots,7\}$ and $\abs{t}<1$, then $w_{n}(y ; b \lvert q , t)$ is holomorphic in a neighbourhood of the $n$-dimensional torus $\mathbb{T}^{n} = \mathbb{T}^{n}_{1}$. The integration cycles $\mathcal{C}_{n}$ and $\mathbb{T}^{n}$ are then homologous in the domain of holomorphy of the integrand by Cauchy's theorem. In the following, it is assumed that the model parameters satisfy these restrictions. The analytic continuation with respect to the parameters are given in Section~\ref{sec_domain_of_holomorphy}. 
\begin{remark}
\label{remark_EHI_more_parameters}
We note that the statements above also hold if we increase the number of parameters in the elliptic hypergeometric integral from 8 to any non-negative integer; see also Section~\ref{sec_domain_of_holomorphy}. Thus, we can use the results above when the gauge function $U_{K}$ and the Cauchy-type kernel function are included in the integrals by replacing the 8 parameters above with $2m- \abs{K}+8$ parameters.
\end{remark}

{\medskip}

We are now in a position to show that the van Diejen operator is formally self-adjoint (or symmetric) with respect to the symmetric $\C$-bilinear form in \eqref{eq_inner_product}. In the following, we revert back to using $(x,a)$ and we can assume that the model parameters satisfy $\abs{a_{s}}<1$ and $\abs{t}<1$ in order to simplify the integration cycle to the $m$-torus $\mathbb{T}^{m}$ without loss of generality. Before stating the results on the self-adjointness of the van Diejen operator, we need some preliminaries. Let us introduce the multi-index notation $T_{q,x}^{\mu} = T_{q,x_1}^{\mu_{1}} \cdots T_{q,x_{m}}^{\mu_{m}}$, for any $x \in(\C^{\ast})^{m}$ and $\mu \in\C^{m}$, which allows us to express the van Diejen operator as
\begin{equation}
\cD_{x}(a\lvert q , t) = \sum_{\mu\in\Z^{m}} F_{\mu}(x ; a \lvert q, t ) T_{q,x}^{\mu}
\end{equation}
where the sum is finite. Then the formal adjoint with respect to $d \omega_{m}(x)$, denoted by $\cD_{x}(a \lvert q , t )^{\ast}$, is given by
\begin{equation}
\cD_{x}( a \lvert q , t)^{\ast} = \sum_{\mu\in\Z^{m}} T_{q,x}^{-\mu} F_{\mu}(x ; a \lvert q , t).
\end{equation}
The van Diejen operator is then formally self-adjoint (or symmetric) with respect to the weight function $w_{m}(x; a \lvert q, t)$ in \eqref{eq_weight_function} in the following sense:
\begin{lemma}
\begin{equation}
\label{eq_vD_formal_self-adjoint}
\cD_{x}(a \lvert q , t)^{\ast} = w_{m}(x ; a \lvert q , t) \circ \cD_{x}(a \lvert q , t) \circ w_{m}(x ; a \lvert q , t)^{-1}.
\end{equation}
\label{lemma_vD_formal_self-adjoint}
\end{lemma}
\begin{proof}
The weight function $w_{m}(x;a\lvert q , t)$ can be decomposed as 
\begin{equation}
w_{m}(x ; a \lvert q , t) = w_{m}^{+}( x ; a \lvert q ,t ) \ w_{m}^{+}(x^{-1} ; a \lvert q , t), 
\end{equation}
with 
\begin{equation}
w_{m}^{+}(x ; a \lvert q , t) = \prod_{1 \leq i \leq m} \frac{\prod_{ 0 \leq s\leq 7} \Gamma( a_{s} x_{i} ; p , q)}{\Gamma( x_{i}^{2} ; p , q)} \prod_{1 \leq i < j \leq m} \frac{\Gamma( t x_{i}x_{j}^{\pm} ; p , q)}{\Gamma( x_{i} x_{j}^{\pm} ; p , q)}
\end{equation}
satisfying the $q$-difference equation
\begin{equation}
\frac{T_{q,x_{1}} w_{m}^{+}(x ; a \lvert q , t )}{w_{m}^{+}(x;a \lvert q ,t )} = \frac{\prod_{0 \leq s \leq 7} \theta( a_{s} x_{1} ; p)}{\theta(x_{1}^{2} , q x_{1}^{2} ; p)} \prod_{2 \leq j \leq m} \frac{\theta( t x_{1} x_{j}^{\pm};p) }{\theta( x_{1} x_{j}^{\pm} ; p )}.
\end{equation}
Recalling the coefficients $A_{i}^{\pm}(x; a \lvert q , t)$ in \eqref{eq_coefficients_positive_multiplicative}, it is clear that there exists a $\sigma \in W_{m}$ such that
\begin{equation}
A_{i}^{\pm}(x; a \lvert q , t) = \frac{T_{q,x_{i}}^{\pm1} \sigma w_{m}^{+}(x ; a \lvert q, t) }{ \sigma w_{m}^{+}(x ; a \lvert q, t)} \quad ( i\in\{1,2,\ldots,m\})
\end{equation}
which allows us to express the van Diejen operator as
\begin{equation}
\cD_{x}(a \lvert q, t) = A_{0}(x; a \lvert q, t) + \frac{1}{2^{n-1} (n-1)! }\sum_{\sigma \in W_{m} } \sigma\Bigl(  \frac{(T_{q,x}^{\epsilon_{1}} w_{m}^{+}(x ; a \lvert q , t) )}{w_{m}^{+}(x ; a \lvert q , t)} T_{q,x}^{\epsilon_{1}} \Bigr) ,
\end{equation} 
where $\epsilon_{i}$ with $i\in\{1,2,\ldots,m\}$ are the canonical basis in $\C^{m}$ and $2^{n-1} (n-1)!$ is the order of the stabilizer sub-group. We then obtain the desired relation by
\begin{multline}
w_{m}(x)^{-1} \circ \cD_{x}(a \lvert q, t )^{\ast} \circ w_{m}(x) = A_{0}(x; a \lvert q ,t ) \\
+ \frac{1}{2^{n-1} (n-1)! }\sum_{ \sigma \in W_{m} } \sigma \Bigl( w_{m}^{+}(x)^{-1} w_{m}^{+}(x^{-1})^{-1} \circ T_{q,x}^{-\epsilon_{1}} \circ ( T_{q,x}^{\epsilon_{1} } w_{m}^{+}(x)) w_{m}^{+}(x^{-1})\Bigr) \\
= \cD_{x}(a \lvert q , t), \quad \quad \quad\quad\quad \quad \quad \quad\quad\quad \quad \quad \quad\quad\quad
\end{multline}
where we used that the weight function $w_{m}(x;a \lvert q ,t )$ is $W_{m}$-invariant and simplify the notation by writing $w_{m}^{+}(x) = w_{m}^{+}(x ;a \lvert q , t)$. 
\end{proof}
From the result of Lemma~\ref{lemma_vD_formal_self-adjoint}, we find that 
\begin{equation}
w_{m}(x;a \lvert q, t) \cD_{x}(a\lvert q ,t ) f(x) = \cD_{x}(a\lvert q, t )^{\ast}( w_{m}(x; a \lvert q, t) f(x))
\end{equation}
which makes it straightforward to check that the van Diejen operator is formally self-adjoint.
\begin{lemma}
The van Diejen operator $\cD_{x}(a\lvert q , t)$ is formally self-adjoint with respect to the symmetric $\C$-bilinear form in \eqref{eq_inner_product}, {i.e.}
\begin{equation}
\langle \cD_{x}(a\lvert q , t) f , g \rangle = \langle  f , \cD_{x}(a\lvert q , t) g \rangle .
\end{equation}
\end{lemma}
\begin{proof}
Let $f(x)$ and $g(x)$ be $W_{m}$-invariant functions that are holomorphic in a domain that contains the $m$-dimensional annulus $\mathbb{A}_{\abs{q}}^{n}$. It is clear that zeroth order coefficient $A^{0}(x;a\lvert q ,t)$ is symmetric with respect to the symmetric $\C$-bilinear form, and that we need only consider the shift terms in the van Diejen operator.
Recall that we can set the model parameters to have absolute value less than $1$ and note that the scalar product depends holomorphically on $p$ ($\abs{p}<1$) so that we can assume $\abs{p} < \abs{q}$ without loss of generality.
Let us start by considering the term involving  $T_{q, x_{m}}$, \emph{i.e.}
\begin{equation}
\oint_{\abs{x_{m}} = 1} \frac{d x_{m}}{2\pi \imag x_{m}} w_{m}(x; a \lvert q , t) g(x) A^{+}_{m}(x; a \lvert q , t) T_{q,x_{m}} f(x).
\end{equation}
The product of the coefficient $A^{+}_{m}(x;a\lvert q , t)$ and the weight function $w_{m}(x;a \lvert q, t )$ equals
\begin{equation}
w_{m-1}(x ; a \lvert q , t) \frac{\prod_{0 \leq s \leq 7} \Gamma( q a_{s} x_{m} , a_{s} x_{m}^{-1} ; p , q)}{\Gamma( q^{2} x^{2}_{m} , x_{m}^{-2} ; p , q )} \prod_{ 1 \leq j \leq m-1} \frac{\Gamma( q t x_{m} x_{j}^{\pm}, t x_{m}^{-1} x_{j}^{\pm} ; p , q )}{\Gamma( q x_{m} x_{j}^{\pm}, x_{m}^{-1} x_{j}^{\pm}  ; p , q) },
\end{equation}
where $w_{m-1}(x ; a \lvert q , t)$ is the weight function in the $(m-1)$-variables $(x_{1},\ldots, x_{m-1})$. 
(Here, all the other $x_{i}$ variables ($i \in \{1,2,\ldots,m-1\}$) are fixed to generic values on the $(m-1)$-torus.) 
The integrand has poles at 
\begin{align}
& x_{m} \in a_{s}^{-1} p^{\Z_{\leq 0}} q^{\Z_{\leq 1}}, \quad x_{m} \in a_{s} p^{\Z_{\geq 0}} q^{\Z_{\geq 0}} \ (s\in\{0,1,\ldots,7\}), \quad 
x_{m}^{2} \in p^{\Z_{\neq 0}} q^{-1},\\
& t x_{m} x_{j}^{\pm} \in p^{\Z_{\leq 0}} q^{\Z_{\leq 1}} , \quad  t x_{m}^{-1} x_{j}^{\pm} \in p^{\Z_{\leq 0} } q^{\Z_{\leq 0}} \ (j\in\{1,2,\ldots,m-1\}).
\end{align} 
All these poles lie outside the annulus $1 \leq \abs{x_{m}} \leq \abs{q^{-1}}$ under the assumptions $\abs{p} < \abs{q}$, $\abs{a_{s}} <1$, and $\abs{t} < 1$. We can therefore continuously deform the contour to $\abs{x_{m}} = \abs{q^{-1}}$ by Cauchy's theorem. Secondly, we see that the weight function $w_{m}(x;a\lvert q, t)$ satisfies the $q$-difference equation
\begin{equation}
\frac{T_{q,x_{m}}w_{m}(x;a\lvert q, t )}{w_{m}(x;a \lvert q ,t )} = \frac{A^{+}_{m}(x;a \lvert q, t)}{T_{q,x_{m}} A_{m}^{-}(x;a \lvert q , t) }
\end{equation}
by straightforward calculations, which yields that 
\begin{equation}
w_{m}(x; a \lvert q , t) A^{+}_{m}(x; a \lvert q , t) T_{q,x_{m}} f(x) = T_{q,x_{m}}w_{m}(x; a \lvert q , t) A^{-}_{m}(x; a \lvert q , t) f(x) .
\end{equation}
Combining these results allows us to express the integral as
\begin{multline}
\oint_{\abs{x_{m}} = 1} \frac{d x_{m}}{2\pi \imag x_{m}} \  w_{m}(x; a \lvert q , t) \  g(x) \ A^{+}_{m}(x; a \lvert q , t) T_{q,x_{m}} f(x) \\ 
= \oint_{\abs{x_{m}} = \abs{q^{-1}} } \frac{d x_{m}}{2\pi \imag x_{m}} \  g(x) \ T_{q,x_{m}} w_{m}(x; a \lvert q , t)A^{-}_{m}(x; a \lvert q , t)  f(x) \\
= \oint_{\abs{x_{m}} = 1} \frac{d x_{m}}{2\pi \imag x_{m}} \ w_{m}(x; a \lvert q , t) \ f(x)  \ A^{-}_{m}(x; a \lvert q , t) T_{q,x_{m}}^{-1} g(x),
\end{multline}
where we rescaled $x_{m} \to q^{-1} x_{m}$ in the last step. Thus, we have shown that 
\begin{equation}
\langle A_{m}^{+}(x; a \lvert q , t) T_{q,x_{m}} f , g \rangle = \langle f ,A_{m}^{-}(x; a \lvert q , t) T_{q,x_{m}}^{-1}  g \rangle.
\end{equation}
The identity also holds if we exchange $A^{+}_{m}(x;a\lvert q,t )T_{q,x_{m}}$ and $A^{-}_{m}(x;a\lvert q, t)T_{q,x_{m}}^{-1}$ by inversion symmetry and we obtain that the operator
\begin{equation}
A^{+}_{m}(x;a \lvert q, t ) T_{q,x_{m}} + A^{-}_{m}(x;a\lvert q, t) T_{q,x_{m}}^{-1}
\end{equation}
is symmetric with regard to this $\C$-bilinear form. This also holds for any permutation of the $x$-variables and we obtain the desired results.
\end{proof}

\section{Proofs of main results}
\label{sec_proofs}
In this Section we give the proofs for Theorems~\ref{thm_Selberg_type_II}, \ref{thm_Selberg_type_I}, \ref{thm_Cauchy_eigenfunction_transform_Type_II}, \ref{thm_Cauchy_eigenfunction_transform_Type_I}, \ref{thm_dual-Cauchy_eigenfunction_transform_Type_II}, and \ref{thm_dual-Cauchy_eigenfunction_transform_Type_I}. In the following, the parameters are restricted such that the integration cycles are given the $n$-dimensional torus $\mathbb{T}^{n}$ in our main theorems and the analytic continuation of these integral transformations are given in Section~\ref{sec_domain_of_holomorphy}.
\begin{proof}[Proof of Theorem~\ref{thm_Cauchy_eigenfunction_transform_Type_II}]
It is convenient to introduce some further notation. 
Define the expectation value with respect to the weight function $w_{n} = w_{n}(y ; b \lvert q , t)$ by
\begin{equation}
\label{eq_expectation_value}
\langle f \rangle_{w_{n}} = \int_{\mathbb{T}^{n}} d\omega_{n}(y) w_{n}(y;b \lvert q, t ) f(y),
\end{equation}
for any function $f(y)$ such that $w_{n}(y;b\lvert q, t) f(y)$ is holomorphic for $y$ in $\mathbb{T}^n$. The symmetric $\C$-bilinear form for this weight function $w_{n}(y ;b \lvert q,  t)$ can then be expressed as
\begin{equation}
\langle f , g \rangle = \langle f g \rangle_{w_{n}},
\end{equation}
and the integral transform can be expressed as $\langle \Phi \ U_{K} \ \varphi \rangle_{w_{n}} $. The poles of the integrand are separable, in the sense of Section~\ref{sec_scalar_product}, as the $x$-variables satisfy $p q a_{s}^{-1} x_{i}^{\ve}\notin p^{\Z_{\leq 0}}q^{\Z_{\leq 0}}$ and $p q t^{-1} x_{i}^{\ve} x_{j}^{\ve^\prime} \notin p^{\Z_{\leq 0}}q^{\Z_{\leq 0}}$ for all $s\notin K$, $\ve,\ve^{\prime}\in\{\pm\}$, and distinct $i,j\in\{1,2,\ldots,m\}$. (Recall that $a = p^{\half} q^{\half} t^{\half} b^{-1}$.)  Thus, the existence of a cycle $\mathcal{C}_{n}$ where the integrand is holomorphic is ensured; see also Remark~\ref{remark_EHI_more_parameters}. For $x$-variables in the domain \eqref{eq_theorem_domain} it is clear that $\abs{c x^{\ve}}<1$, where $c=p^{\half} q^{\half}t^{-\half}$, for both $\ve\in\{\pm\}$ and the cycle $\mathcal{C}_{n}$ is homologous to $\mathbb{T}^{n}$. The condition $\abs{p} < \abs{q t }$ ensures that this $m$-dimensional annuls \eqref{eq_theorem_domain} is non-empty.

Let us proceed by showing that $T_{q,x_i}^{\ve}  \langle \Phi \  U_{K}  \varphi \rangle_{w_{n}} =  \langle T_{q,x_{i}}^{\ve} \Phi \ U_{K} \varphi \rangle_{w_{n}}$ for all $i \in \{1,2,\ldots,m\}$ and $\ve\in\{\pm\}$: It is straightforward to check that $T_{q,x_{i}}^{\ve}\Phi(x,y\lvert q , t)$, with $x$ in \eqref{eq_theorem_domain}, is holomorphic for $y$ in $\mathbb{T}^{n}$ since $\abs{ p^{\half} q^{\half} t^{-\half} q^{\ve} x_{i}^{\ve^\prime} y_{k}^{\ve^{\prime \prime}} } < 1$ for all $i\in\{1,2,\ldots,m\}$, $k\in\{1,2,\ldots,n\}$, and $\ve,\ve^\prime, \ve^{\prime \prime} \in\{\pm\}$.  The stricter conditions that $x_{i}^{\ve}\notin a_{s} p^{\Z_{<0}} q^{\Z_{\leq 0}}$ ($s\notin K$, $i\in\{1,2,\ldots,m\}$, and $\ve\in \{\pm\}$) and $x_{i}^{\ve} x_{j}^{\ve^\prime} \notin p^{\Z_{<0}} q^{\Z_{\leq 0}} t$ ( $1 \leq i < j \leq m$ and $\ve , \ve^\prime \in\{\pm\})$ ensures that the poles are separable, in the sense of Section~\ref{sec_scalar_product}, even after $q$-shifts. Therefore, we can interchange the shift operator and integrations without encountering any poles and obtain that $\cD_{x}(a\lvert q , t) \langle \Phi \  U_{K} \varphi \rangle_{w_{n}} = \langle ( \cD_{x}(a\lvert q , t)\Phi) \  U_{K} \varphi \rangle_{w_{n}}.$ The kernel function identity in Lemma~\ref{lemma_Cauchy_KFI} then yields that
\begin{equation}
\cD_{x}(a \lvert q , t)  \langle \Phi \  U_{K} \varphi \rangle_{w_{n}} = \langle (\cD_{y}(b\lvert q, t)\Phi) \  U_{K} \varphi \rangle_{w_{n}} 
\end{equation}
if the parameters satisfy 
\begin{equation}
a_{0}\cdots a_{7} t^{2(m-n)} = p^{2} q^{2} t^{2} \quad \Leftrightarrow \quad  b_{0} \cdots b_{7} t^{2(n-m)} = p^{2} q^{2} t^{2}.
\end{equation}
Using that the van Diejen operator is (formally) self-adjoint, \emph{i.e.}  the operator satisfies $
\langle ( \cD_{y}(b \lvert q , t) \Phi) \  U_{K}  \varphi \rangle_{w_{n}} = \langle \Phi \ \cD_{y}(b \lvert q , t)   U_{K}  \varphi \rangle_{w_{n}}, $ the result in Lemma~\ref{lemma_gauge_transform_U}, and that $\varphi(y)$ is an eigenfunction of the van Diejen operator $\cD_{y}(b_{K(pq)}\lvert q, t)$ yields
\begin{equation}
\langle (\cD_{y}(b\lvert q, t)\Phi) \  U_{K} \varphi \rangle_{w_{n}} = \langle \Phi \  U_{K} \cD_{y}(b_{K(pq)}\lvert q, t)\varphi \rangle_{w_{n}} = \Lambda \langle \Phi \  U_{K} \varphi \rangle_{w_{n}}.
\end{equation}
Combining these results and setting $\langle \Phi \  U_{K} \varphi \rangle_{w_{n}} = \psi_{K}(x ; a\lvert q ,t )$ yields the desired results, \emph{i.e.}
\begin{equation}
\cD_{x}(a \lvert q , t) \psi_{K}(x ; a \lvert q, t ) = \Lambda \psi_{K}(x ; a \lvert q, t)
\end{equation}
if the parameters satisfy \eqref{eq_balancing_condition_KFI}.

For $\abs{K} > 0$, we do not get any contradiction between the balancing condition $b_{0} \cdots b_{7} t^{2(n-m)} = p^{2} q^{2} t^{2}$ and that $\abs{b_s} < 1$ for all $s\notin K$, since they can be satisfied by parameters whose index is in $K$. In the case where $K=\emptyset$, we have that $\abs{ b_{0} \cdots b_{7}} < 1$ and obtain the requirement that 
\begin{equation}
\abs{p^{2} q^{2} t^{2}} = \abs{ b_0 \cdots b_7 t^{2(n-m)}} < \abs{ t^{2(n-m)}}
\end{equation}
in order for these conditions to be satisfied at the same time. Therefore, we need that $\abs{p} < \abs{q^{-1} t^{n-m-1}}$, which is automatically satisfied if $\abs{p} < \abs{q t } < \abs{ q^{-1} t^{n-m-1}}$.
\end{proof}

\begin{proof}[Proof of Theorem~\ref{thm_Cauchy_eigenfunction_transform_Type_I}]
The proof of Theorem~\ref{thm_Cauchy_eigenfunction_transform_Type_I} is essentially the same as the proof of Theorem~\ref{thm_Cauchy_eigenfunction_transform_Type_II}. As such, we will not repeat all the arguments and focus on the differences.

Using the same notation $\langle f \rangle_{w_n}$ for the expectation value in \eqref{eq_expectation_value}, we have that the integral transform can be expressed as $\langle \Phi U_{K} V \varphi \rangle$. The integral is well-defined by the same arguments as before, and we have that  
\begin{multline}
\cD_{x}(a \lvert q ,t )\langle \Phi \ U_{K} V \varphi \rangle_{w_n} = \langle (\cD_{x}(a\lvert q, t) \Phi) \ U_{K} V \varphi \rangle_{w_n} \\
= \langle ( \cD_{y}(b \lvert q , t)\Phi) \ U_{K} V \varphi \rangle_{w_n} = \langle \Phi \ \cD_{y}( b \lvert q , t) U_{K} V \varphi \rangle_{w_n}
\end{multline}
if the parameters satisfy the balancing condition \eqref{eq_balancing_condition_KFI}. Here, the interchange of integration and the shift operators are justified by the same argument as in the proof of Theorem~\ref{thm_Cauchy_eigenfunction_transform_Type_II}. Using the results of Lemmas~\ref{lemma_gauge_transform_V} and \ref{lemma_gauge_transform_U} yields that
\begin{multline}
\langle \Phi \ \cD_{y}(b \lvert q , t) U_{K} V \varphi \rangle_{w_n} = \langle \Phi \ U_{K} \cD_{y}(b_{K(pq)} \lvert q , t) V \varphi \rangle_{w_n} \\
= \langle \Phi \ U_{K} V \cD_{y}( b_{K(pq)} \lvert q , p q t^{-1}) \varphi \rangle_{w_n} = \Lambda \langle \Phi \ U_{K} V \varphi \rangle_{w_n},
\end{multline}
where we used that $\varphi$ is an eigenfunction of $\cD_{y}( b_{K(pq)} \lvert q , p q t^{-1})$ in the last line.
Setting $\langle \Phi \ U_{K} V \varphi \rangle_{w_n} = \tilde{\psi}_{K}(x ; a \lvert q , t)$ yields the desired result, \emph{i.e.}
\begin{equation}
\cD_{x}(a \lvert q , t) \tilde{\psi}_{K}(x ; a \lvert q, t) = \Lambda \tilde{\psi}_{K}(x ; a \lvert q , t),
\end{equation}
under the balancing condition \eqref{eq_balancing_condition_KFI}.
\end{proof}

Using these results, it is straightforward to prove Theorem~\ref{thm_Selberg_type_II} and \ref{thm_Selberg_type_I}:
\begin{proof}[Proof of Theorems~\ref{thm_Selberg_type_II} and \ref{thm_Selberg_type_I}]
Setting $\varphi = 1$ in Theorem~\ref{thm_Cauchy_eigenfunction_transform_Type_II}, resp. Theorem~\ref{thm_Cauchy_eigenfunction_transform_Type_I}, yields that \eqref{eq_Selberg_eigenfunction_type_II}, resp. \eqref{eq_Selberg_eigenfunction_type_I}, is an eigenfunction of the van Diejen operator if the parameters satisfy the balancing condition \eqref{eq_balancing_condition_KFI} and that $1$ is an eigenfunction of the operator $\cD_x(b_{K(pq)}\lvert q , t)$, resp. $\cD_x(b_{K(pq)}\lvert q , p q t^{-1})$. From \eqref{eq_vD_constant_eigenvalue_eq}, it follows that $1$ is an exact eigenfunction if the parameters satisfy
\begin{equation}
(\prod_{0 \leq s \leq 7} (b_{K(pq)})_{s}) t^{2n} = p^{2} q^{2} t^{2}, \quad \text{resp.} \quad (\prod_{0\leq s \leq 7} (b_{K(pq)})_{s}) (p q t^{-1})^{2n} = p^{2} q^{2} (p q t^{-1})^{2}.
\end{equation}

Combining this condition and the balancing condition \eqref{eq_balancing_condition_KFI} yields the two parameter restrictions.
\end{proof}
{\smallskip}

\begin{proof}[Proof of Theorem~\ref{thm_dual-Cauchy_eigenfunction_transform_Type_II}]
We introduce the notation
\begin{equation}
\langle f \rangle^{\vee}_{w_{n}} = \int_{\mathbb{T}^{n}} d\omega_{n}(y) w_{n}(y ; a \lvert t , q) f(y),
\end{equation}
for suitable functions $f$, to simplify notation. The eigenfunction transform with the dual-Cauchy kernel function can then be written as $\langle \Phi^{\vee} U_{K} \varphi \rangle^{\vee}_{w_{n}}.$ The poles of the integrand arises only from the weight function $w_{n}(y;a\lvert t , q)$ and it follows that the integrand is holomorphic and that the cycle is homologous to $\mathbb{T}^{n}$ from the discussions in Section~\ref{sec_scalar_product}, albeit with $q$ and $t$ interchanged. 

Using the dual kernel function identity in Lemma~\ref{lemma_dual-Cauchy_KFI} yields that
\begin{equation}
\cD_{x}(a \lvert q , t) \langle \Phi^{\vee} U_{K} \varphi \rangle^{\vee}_{w_{n}} = - \frac{t^{m}\theta(q ;p)}{q^{n} \theta(t ; p)} \langle ( \cD_{y}(a \lvert t , q) \Phi^{\vee}) U_{K} \varphi \rangle^{\vee}_{w_{n}}
\end{equation}
if the parameters satisfy $a_{0} \cdots a_{7} q^{2n} t^{2m} = p^{2} q^{2} t^{2}$. 
(Here, the interchange of integration and the $q$-shift operator is straightforward since the integrand does not have any poles with respect to the $x$-variables.)
From Lemmas~\ref{lemma_gauge_transform_U} and \ref{lemma_vD_formal_self-adjoint}, we have that
\begin{equation}
\langle (\cD_{y}(a \lvert t , q) \Phi^{\vee}) U_{K} \varphi \rangle^{\vee}_{w_{n}} = \langle \Phi^{\vee} \cD_{y}(a \lvert t , q)U_{K} \varphi \rangle^{\vee}_{w_{n}} = \langle \Phi^{\vee} U_{K} \cD_{y}(a_{K(pt)} \lvert t , q) \varphi \rangle^{\vee}_{w_{n}}
\end{equation}
by straightforward calculations.
Using that $\varphi$ is an eigenfunction of the van Diejen operator $\cD_{y}(a_{K(pt)} \lvert t , q)$ and that $\psi_{K}^{\vee}(x;a\lvert q , t) = \langle \Phi^{\vee} U_{K} \varphi \rangle^{\vee}_{w_{n}}$, yields the desired result, \emph{i.e.}
\begin{multline}
\cD_{x}(a\lvert q , t) \psi_{K}^{\vee}(x ; a \lvert q , t) = - \frac{t^{m}\theta(q ;p)}{q^{n} \theta(t ; p)} \langle \Phi^{\vee} U_{K} \cD_{y}(a_{K(pt)} \lvert t , q) \varphi \rangle^{\vee}_{w_{n}}  \\
= - \frac{t^{m}\theta(q ;p)}{q^{n} \theta(t ; p)} \Lambda \psi_{K}^{\vee}(x;a \lvert q, t ).
\end{multline}

In the case where $K \neq \emptyset$, we again do not have any contradiction between the dual balancing condition $a_{0} \cdots a_{7} q^{2n} t^{2m} = p^{2} q^{2} t^{2}$ and that $\abs{a_s} <1$ for all $s\notin K$. When $K=\emptyset$, we have that $\abs{a_0 \cdots a_{7}} <1$ and require that
\begin{equation}
\abs{p^{2} q^{2} t^{2}} = \abs{ a_{0} \cdots a_{7} q^{2n} t^{2m}} < \abs{ q^{2n} t^{2m}},
\end{equation}
which yields the condition $\abs{p} < \abs{q^{n-1} t^{m-1}}$.
\end{proof}

\begin{proof}[Proof of Theorem~\ref{thm_dual-Cauchy_eigenfunction_transform_Type_I}]
The proof of Theorem~\ref{thm_dual-Cauchy_eigenfunction_transform_Type_I} is essentially the same as the proof of Theorem~\ref{thm_dual-Cauchy_eigenfunction_transform_Type_II}.
Writing the integral transform as $\langle \Phi^\vee U_{K} V \varphi \rangle^{\vee}_{w_{n}}$ yields
\begin{multline}
\langle (\cD_{y}(a \lvert t , q)\Phi^\vee) U_{K} V \varphi \rangle^{\vee}_{w_{n}} = \langle \Phi^\vee \cD_{y}(a \lvert t , q) U_{K} V \varphi \rangle^{\vee}_{w_{n}} \\
= \langle \Phi^\vee U_{K} \cD_{y}(a_{K(pt )} \lvert t , q ) V \varphi \rangle^{\vee}_{w_{n}} = \langle \Phi^\vee U_{K} V \cD_{y}(a_{K(pt)} \lvert t , p t q^{-1} )\varphi \rangle^{\vee}_{w_{n}}.
\end{multline}
If $\varphi$ is an eigenfunction of the van Diejen operator $\cD_{y}(a_{K(pt)} \lvert t , p t q^{-1} )$ with eigenvalue $\Lambda$, we find that
\begin{equation}
\cD_{x}(a\lvert q, t) \langle \Phi^\vee U_{K} V \varphi \rangle^{\vee}_{w_{n}} = -\frac{t^{m} \theta( q ;p) }{q^{n} \theta(t ;p)} \Lambda \langle \Phi^\vee U_{K} V \varphi \rangle^{\vee}_{w_{n}}.
\end{equation}
Setting $\tilde{\psi}_{K}^{\vee}(x;a\lvert q , t) = \langle \Phi^\vee U_{K} V \varphi \rangle^{\vee}_{w_{n}}$ then yields the desired result.
\end{proof}

\section{Domain of holomorphy}
\label{sec_domain_of_holomorphy}
The integral transformations allow for an analytic continuation with respect to the parameters and the $x$-variables. In this Section, we discuss the analytic continuation and the domains of the integral transformations.

Before proceeding, it is more convenient to extend the weight function to arbitrary number of parameters, \emph{i.e.} $b \to \tilde{b} = (\tilde{b}_{0}, \ldots, \tilde{b}_{\ell -1 }) \in (\C^{\ast})^{\ell}$, by
\begin{equation}
w_{n,\ell}(y ; \tilde{b} \lvert q , t) = \prod_{1 \leq k \leq n} \frac{\prod_{0 \leq s\leq \ell-1} \Gamma( \tilde{b}_{s} y_{k}^{\pm} ; p , q)}{\Gamma( y_{k}^{\pm 2} ; p , q)} \prod_{1 \leq k < l \leq n} \frac{\Gamma( t y_{k}^{\pm} y_{l}^{\pm} ; p , q)}{\Gamma( y_{k}^{\pm} y_{l}^{\pm} ; p , q)}.
\end{equation}
The analytic continuation of the elliptic hypergeometric integral is then obtained from Lemma~5.1 of \cite{NI19}:
\begin{lemma}[Lemma~5.1 of \cite{NI19}]
\label{lemma_NI19}
Let $\ell\in\Z_{\geq 0}$, $\tilde{b}=(\tilde{b}_{0} ,\ldots , \tilde{b}_{\ell-1})\in (\C^{\ast})^{\ell}$, and $f(y ; \tilde{b})$ be a holomorphic function on $(\C^{\ast})^{n}\times (\C^{\ast})^{\ell}$. Suppose the parameters satisfy $\abs{\tilde{b}_{s}} < 1$ for all $s\in\{0,1,\ldots, \ell-1 \}$ and $\abs{t} < R^{2}$ for some $R\in (0,1]$. Then the integral
\begin{equation}
\int_{\mathbb{T}^{n}} d\omega_{n}(y) \ w_{n,\ell}( y ; \tilde{b} \lvert q , t) f(y ; \tilde{b})
\label{eq_elliptic_hypergeometric_integral}
\end{equation}
can be continued to a holomorphic function on
\begin{equation}
\label{eq_holomorphic_domain}
\{ \tilde{b} = (\tilde{b}_{0} , \ldots , \tilde{b}_{\ell - 1} )\in (\C^{\ast})^{\ell} \lvert \ \abs{\tilde{b}_{s}} < R^{-1} ,  \tilde{b}_{r} \tilde{b}_{s} \notin p^{\Z_{\leq 0}} q^{\Z_{\leq 0}}  \  ( r,s \in \{0,1,\ldots,\ell -1\})\}.
\end{equation}
\end{lemma}

The product of the Cauchy kernel function in \eqref{eq_Cauchy_KF} and the weight function in \eqref{eq_weight_function} can be expressed as 
\begin{equation}
\Phi(x,  y \lvert q , t) w_{n}(y; b \lvert q , t) = w_{n,2m+8}( y ; \tilde{b} \lvert q, t)
\end{equation}
where $\tilde{b}\in(\C^{\ast})^{2m + 8}$ is given by
\begin{equation}
\label{eq_b_vector}
\tilde{b}= ( b_{0},\ldots, b_{7}, c x_{1}, c x_{1}^{-1} ,\ldots,  c x_{m} , c x_{m}^{-1}),
\end{equation}
with $c = p^{\half} q^{\half} t^{-\half}$.
The results of \cite[Lemma~5.1]{NI19} allows us to then find the domain of holomorphy in a straightforward way by starting in the domain $\abs{\tilde{b}_{s}}<1$ for all $s\in \{0,1,\ldots, \ell -1 \}$, with $\ell =2m+8$, as well as the cases where the gauge functions $U_{K}(y;b\lvert q)$ in \eqref{eq_gauge_function_U} are included. 

The requirements on the function $f(y; \tilde{b})$ in Lemma~\ref{lemma_NI19} can also be relaxed and we only require that the function is holomorphic in the domain $\mathbb{A}^{n}_{\rho} \times \cU$, where $\cU$ is the domain in \eqref{eq_holomorphic_domain}.

\begin{remark}
\label{remark_Type_I_integrals}
Note that the integrand for elliptic hypergeometric integrals of Type I do not have poles at 
\begin{equation}
y_{k}^{\ve} y_{l}^{\ve^\prime} \in  p^{\Z_{\geq 0} }q^{\Z_{\geq0}} t \quad (1 \leq k < l \leq n; \ \ve,\ve^\prime\in\{\pm\}).
\end{equation}
The results of Lemma~\ref{lemma_NI19} still hold, although without any requirement on the parameter $t$ and it follows from Cauchy's theorem that the integral does not depend on $R\in(0,1]$. For $f(y ; \tilde{b})$ holomorphic in $\mathbb{A}^{n}_{\rho} \times (\C^{\ast})^{\ell}$ for some $\rho\in (0,1]$, the elliptic hypergeometric integrals of Type I can be analytically continued to a holomorphic function on
\begin{equation}
\{ \tilde{b}\in(\C^{\ast})^{\ell } \lvert \ \abs{\tilde{b}_{r}} < \rho^{-1} , \ \tilde{b}_{r}\tilde{b}_{s} \notin p^{\Z_{\leq 0}} q^{\Z_{\leq 0}} \ (r,s\in\{0,1,\ldots,\ell - 1 \} )\}.
\end{equation}
\end{remark}

We are now able to find the analytic continuation of the eigenfunction transformations in a straightforward way:
\begin{proposition}
\label{prop_Cauchy_kernel_transform}
Let $R \in (0,1]$ and $x\in \mathbb{T}^{m}$. Suppose the model parameters satisfy
\begin{equation}
\abs{p} < \min( 1 , \abs{q^{-1} t } ) , \quad \abs{t} < R^{2} , \quad  \abs{a_s} > \abs{(p q t)^{\half}}  \ (s\in\{0,1,\ldots,7\}),
\end{equation}                            
and let $\varphi(y)$ be holomorphic in $(\C^{\ast})^{n}$. Then the function
\begin{equation}
\psi(x;a \lvert q, t) = \int_{\mathbb{T}^{n}} d \omega_{n}(y) \ w_{n}(y ; b \lvert q , t) \Phi(x , y \lvert q ,t ) \varphi(y),
\end{equation}
where $b=(b_0 , \ldots, b_7)$ with $b_s = p^{\half} q^{\half} t^{\half} a_{s}^{-1}$ for all $s\in\{0,1,\ldots,7\}$, can be analytically continued to a holomorphic function on
\begin{multline}
\label{eq_image_domain_Cauchy}
\bigl\{ x\in (\C^{\ast})^{m} \lvert \  \abs{ (p q t^{-1})^{\half} } R < \abs{x_{i}} <  \abs{ (p q t^{-1})^{-\half}} R^{-1} \ ( i \in\{1,2,\ldots,m\} ), \bigr. \\
x_{i}^{\ve} \notin a_{s} p^{\Z_{<0}} q^{\Z_{<0}} \ (s\in\{0,1,\ldots,7\} ; \ i\in\{1,2,\ldots,m\} ; \ \ve\in\{\pm\} ), \\
\bigl. x_{i}^{\ve} x_{j}^{\ve^\prime} \notin  p^{\Z_{<0}} q^{\Z_{<0}} t \  ( 1 \leq i < j \leq m ; \ \ve , \ve^\prime \in \{\pm\})\bigr\}
\end{multline}
and parameters 
\begin{equation}
 \quad \abs{p} < \min( 1 , \abs{q^{-1} t } R^{-2}  ) , \quad \abs{a_s} > \abs{(p q t)^{\half}}R  \quad (s\in\{0,1,\ldots,7\}),
\end{equation}      
if they satisfy $t \notin p^{\Z_{>0}} q^{\Z_{>0}}$ and $a_{r}a_{s} \notin p^{\Z_{>0}} q^{\Z_{>0}} t$ for all $r,s\in\{0,1,\ldots,7\}$.
\end{proposition}
\begin{proof}
As already mentioned, we can express the integral transform as
\begin{equation}
 \int_{\mathbb{T}^{n}} d \omega_{n}(y) \ w_{n,2m+8}(y ; \tilde{b} \lvert q , t) \varphi(y),
\end{equation}
with $\tilde{b}$ as in \eqref{eq_b_vector}. The integral defines a holomorphic function on $\abs{\tilde{b}_{s}}<1$ for all $s\in\{0,1,\ldots, 2m+7\}$ and it is straightforward to check that it yields a holomorphic function of $x\in\mathbb{T}^{n}$ as $\abs{p^{\half} q^{\half} t^{-\half}} < 1$. Since $\abs{t}< R^{2}$ ($R \in (0,1]$), the integral can be continued to a holomorphic function on 
\begin{equation}
\abs{\tilde{b}_{s}} < R^{-1} \quad (s\in\{0,1,\ldots,2m+7\})
\end{equation}
by Lemma~\ref{lemma_NI19}, if the parameters satisfy 
\begin{equation}
\tilde{b}_{s} \tilde{b}_{r} \notin p^{\Z_{\leq 0}} q^{\Z_{\leq 0}}
\end{equation}
for all $r,s \in\{0,1,\ldots,2m+7\}$. Let us start by considering the first condition, \emph{i.e.} $\abs{\tilde{b}_{s}} < R^{-1}$ for all $s\in\{0,1,\ldots,2m+7\}$. When $s\in\{0,1,\ldots, 7\}$, we have that
\begin{equation}
\abs{\tilde{b}_{s}} = \abs{b_{s}} = \abs{ (p q t)^{\half} a_{s}^{-1}} < R^{-1} \quad (s\in\{0,1,\ldots,7\})
\end{equation}
which yields the condition $\abs{a_{s}} > \abs{ (p q t)^{\half} } R$. If $s\in\{8,9,\ldots,2m+7\}$, then
\begin{equation}
\abs{\tilde{b}_{s}} < R^{-1} \Leftrightarrow \begin{cases} \abs{(p q t^{-1})^{\half} x_{i} }< R^{-1} , \ &\text{ if } s = 2 i + 6 \\
 \abs{(p q t^{-1})^{\half} x_{i}^{-1} } < R^{-1} , \ &\text{ if } s = 2 i + 7
\end{cases}
\quad ( i\in\{1,2,\ldots,m\} ).
\end{equation}
These yield that the domain can be continued to the $m$-dimensional annulus
\begin{equation}
\abs{ (p q t^{-1})^{\half}} R < \abs{x_{i}} < \abs{ (p q t^{-1})^{-\half}} R^{-1}\quad (i\in\{1,2,\ldots,m\})
\end{equation}
with respect to the $x$-variables. Assuming that $\abs{p q t^{-1}} \leq R^{-2}$ ensures that this $m$-dimensional annulus is non-empty and if $\abs{ q^{-1} t} R^{-2} > 1$ it is clear that this condition is already fulfilled for $\abs{p}<1$.

We consider the second condition, \emph{i.e.} $\tilde{b}_{r} \tilde{b}_{s}\notin p^{\Z_{\leq 0}} q^{\Z_{\leq 0}}$ ($r,s\in\{0,1,\ldots,2m+7\}$), in the various cases. 
\newline(1) When $r,s\in\{0,1,\ldots,7\}$, we have that
\begin{equation}
\tilde{b}_{r} \tilde{b}_{s} = p q t a_{r}^{-1} a_{s}^{-1} \notin p^{\Z_{\leq 0}} q^{\Z_{\leq 0}} \Leftrightarrow a_{r} a_{s} \notin p^{\Z_{>0}} q^{\Z_{>0}} t.
\end{equation}
(2) When $r\in\{0 ,1, \ldots, 7\}$ and $s = 2 i + 6$, or $s = 2 i + 7$, with $i\in\{1,2,\ldots,m\}$, we obtain that
\begin{equation}
\tilde{b}_{r} \tilde{b}_{s} = p q a_{r}^{-1} x_{i} \notin p^{\Z_{\leq 0}} q^{\Z_{\leq 0}} \Leftrightarrow x_{i}\notin a_{r} p^{\Z_{< 0}} q^{\Z_{< 0}}
\end{equation}
and
\begin{equation}
\tilde{b}_{r} \tilde{b}_{s} = p q a_{r}^{-1} x_{i}^{-1} \notin p^{\Z_{\leq 0}} q^{\Z_{\leq 0}} \Leftrightarrow  x_{i}^{-1} \notin a_{r} p^{\Z_{< 0}} q^{\Z_{< 0}},
\end{equation}
respectively. Combining these yields that $x_{i}^{\ve} \notin a_{s} p^{\Z_{< 0}} q^{\Z_{< 0}}$ for all $s\in\{0,1,\ldots,7\}$, $i \in\{1,2,\ldots,m\}$ and $\ve\in\{\pm\}$.
\newline(3) When $r = 2 i + 6$ and $s= 2i + 7$ with $i\in\{1,2,\ldots,m\}$, we obtain that $\tilde{b}_{r} \tilde{b}_{s} = p q t^{-1} \notin p^{\Z_{\leq 0}} q^{\Z_{\leq 0}}$ which gives the condition $t \notin p^{\Z_{>0}} q^{\Z_{>0}}$.
\newline(4) When $r = 2 i + 6$ and $s= 2j + 6$, or $s= 2j+7$, with $1 \leq i < j \leq m$, we obtain that
\begin{equation}
\tilde{b}_{r} \tilde{b}_{s} = p q t^{-1} x_{i} x_{j} \notin p^{\Z_{\leq 0}} q^{\Z_{\leq 0}} \Leftrightarrow x_{i} x_{j} \notin  p^{\Z_{< 0}} q^{\Z_{<0}} t
\end{equation}
and 
\begin{equation}
\tilde{b}_{r} \tilde{b}_{s} = p q t^{-1} x_{i} x_{j}^{-1} \notin p^{\Z_{\leq 0}} q^{\Z_{\leq 0}} \Leftrightarrow x_{i} x_{j}^{-1} \notin p^{\Z_{<0}} q^{\Z_{<0}} t,
\end{equation}
respectively. These yield the conditions $x_i x_j^{\ve} \notin p^{\Z_{<0}} q^{\Z_{<0}} t$ for $\ve\in\{\pm\}$.
\newline(5) When $r=2i + 7$ and $s=2j +6$, or $2 j + 7$, with $1 \leq i < j \leq m$, we obtain that
$
\tilde{b}_{r} \tilde{b}_{s} = p q t^{-1} x_{i}^{-1} x_{j} \notin p^{\Z_{\leq 0}} q^{\Z_{\leq 0}},
$ and $
\tilde{b}_{r} \tilde{b}_{s} = p q t^{-1} x_{i}^{-1} x_{j}^{-1} \notin p^{\Z_{\leq 0}} q^{\Z_{\leq 0}},
$
respectively. These yield the conditions $x_{i}^{-1} x_{j}^{\ve} \notin p^{\Z_{<0}} q^{\Z_{<0}} t$ for $\ve\in\{\pm\}$.
Combining (4) and (5) yields that $x_{i}^{\ve} x_{j}^{\ve^{\prime}} \notin p^{\Z_{<0}} q^{\Z_{<0}} t$ for all $1 \leq i < j \leq m$ and $\ve,\ve^\prime \in \{\pm\}$.
\end{proof}

\begin{remark}
It is important to keep in mind that when we are interested in eigenfunctions of analytic difference operators obtained by these integral transformations, we need to restrict the domains beforehand to ensure that the shift operators do not push these eigenfunctions outside their domain of holomorphy.
\end{remark}
It is clear that the result of Lemma~\ref{lemma_NI19} can also be used for the integral transformations that include the gauge functions $V(y \lvert q ,t )$ in \eqref{eq_gauge_function_V} and $U_{K}(y ; b \lvert q )$ in \eqref{eq_gauge_function_U}: It is straightforward to check that including the gauge function $V(y \lvert q , t)$ will not require the restriction $\abs{t}<R^2$ on the parameter $t$. The integral transform can then be analytically continued to a (globally) meromorphic function on $(\C^{\ast})^{m}$ as explained in Remark~\ref{remark_Type_I_integrals}. Including the gauge function $U_{K}(y;b \vert q)$  ($K\subseteq\{0,1,\ldots,7\}$, $b = p^{\half}q^{\half} t^{\half} a^{-1}$) will only affect the domain of holomorphy with respect to the parameters $a$. The restrictions on the model parameters are then given by
\begin{equation}
\abs{a_s} > \abs{ (p q t )^{\half} } R , \quad a_{r} a_{s} \notin p^{\Z_{>0}} q^{\Z_{>0}} t \quad (r,s \notin K)
\end{equation}
for those model parameters $\{a_s\}_{s\notin K}$ whose index $s$ is not part of $K\subseteq\{0,1,\ldots,7\}$. 

\begin{remark}
For particular functions $\varphi(y)$, it known that the integral transforms in Theorems~\ref{thm_Cauchy_eigenfunction_transform_Type_II} and \ref{thm_Cauchy_eigenfunction_transform_Type_I} can be analytically continued to (globally) meromorphic functions of $(x,a)$ on particular hypersurfaces in $(\C^{\ast})^{m+8}$; see Remark~4.2 of \cite{NI19}. This can also be directly observed in special cases where the integrals have known evaluations; see Section~\ref{sec_known_evaluations}. A thorough investigation of the requirements for the functions $\varphi(y)$ is left for future work.
\end{remark}

The restriction on $\varphi(y)$ in Proposition~\ref{prop_Cauchy_kernel_transform}, \emph{i.e.} to be holomorphic in $(\C^{\ast})^{n}$, is stricter than necessary. Indeed, it is clear that the integral is well-defined for any function $\varphi(y)$ holomorphic in any domain where we can continuously deform the integration contour. The following result then follows:
\begin{corollary}
Let $R \in (0,1]$ and fix the parameters such that 
\begin{equation}
\abs{p} < \min( 1 , \abs{q^{-1} t } R^{-2}), \quad  \abs{q} <1, \quad \abs{t} < R^{2}, \quad \abs{b_{s}}< R^{-1} \ (s\in\{0,1,\ldots,7\}),
\end{equation}
satisfying $t \notin p^{\Z_{>0}} q^{\Z_{>0}}$ and $b_{r} b_{s} \notin p^{\Z_{\leq 0}} q^{\Z_{\leq 0}}$ for all $r,s\in\{0,1,\ldots,7\}$. The domain of the integral transform
\begin{equation}
\varphi(y) \mapsto \int_{\mathcal{C}_{n}(\tilde{b}\lvert q , t)} d \omega_{n}(y) \ w_{n}(y ; b \lvert q , t) \Phi(x , y \lvert q ,t ) \varphi(y)
\end{equation} 
consists of all functions $\varphi(y)$ that are holomorphic in $\mathbb{A}_{\rho}^{n}$ for $ \rho \in (0,R]$, and its image is given by functions that are holomorphic in \eqref{eq_image_domain_Cauchy}.
\end{corollary}

The domain of holomorphy for the integral transform using the kernel function of dual-Cauchy type \eqref{eq_dual_Cauchy_KF} is more straightforward to compute since the kernel function $\tPhi(x,y)$ is holomorphic in $(\C^{\ast})^{m}\times (\C^{\ast})^{n}$. In fact, it was shown in \cite[Eq.~(2.5)]{NI19} that $\tPhi(x,y)$ has a finite sum expansion in terms of the elliptic interpolation polynomials of type $BC$ in the $x$-variables and products of theta functions in the $y$-variables.
The only restrictions are then obtained from the weight function $w_{n}(y; a \lvert t , q)$ and the following result is obtained directly from Lemma~\ref{lemma_NI19}:
\begin{proposition}
\label{prop_dual-Cauchy_kernel_transform}
Let $R\in (0,1]$ and $x \in \mathbb{T}^{m}$. Suppose the model parameters satisfy 
\begin{equation}
\abs{t}< 1 ,\quad \abs{q}< R^{2}, \quad \abs{a_s} <1 \ (s\in\{0,1,\ldots,7\}),
\end{equation}
and let $\varphi(y)$ be holomorphic in $(\C^{\ast})^{n}$. Then the function
\begin{equation}
{\psi^\vee}(x;a\lvert q , t) = \int_{\mathbb{T}^{n}} d\omega_{n}(y) w_{n}( y ; a \lvert t , q ) {\tPhi}(x,y) \varphi(y)
\end{equation}
can be continued to a holomorphic function on $(\C^{\ast})^{m}\times \cU^{\vee}$, where
\begin{equation}
\cU^{\vee} = \{a \in (\C^{\ast})^{8} \lvert \ \abs{a_{s}} < R^{-1}, \ a_{r} a_{s} \notin p^{\Z_{\leq 0}} t^{\Z_{\leq 0}} \ (s\in\{0,1,\ldots,7\})\}.
\end{equation}
\end{proposition}
\noindent{}Note that we have interchanged the roll of the parameters $q$ and $t$ in this integral transformation since the dual-Cauchy kernel function ${\tPhi}(x,y)$ relates a pair of van Diejen operators with $q$ and $t$ interchanged; see Lemma~\ref{lemma_dual-Cauchy_KFI}.

\begin{corollary}
Let $R \in (0,1]$ and fix the parameters such that 
\begin{equation}
\abs{p} < 1, \quad  \abs{q} < R^{2}, \quad \abs{t} < 1, \quad \abs{a_{s}}< R^{-1} \ (s\in\{0,1,\ldots,7\}),
\end{equation}
satisfying $a_{r} a_{s} \notin p^{\Z_{\leq 0}} t^{\Z_{\leq 0}}$ for all $r,s\in\{0,1,\ldots,7\}$. The domain of the integral transform
\begin{equation}
\varphi(y) \mapsto \int_{\mathcal{C}_{n}(a\lvert t , q)} d \omega_{n}(y) \ w_{n}(y ; a \lvert t , q) \tPhi(x , y) \varphi(y),
\end{equation} 
consists of all functions $\varphi(y)$ that are holomorphic in $\mathbb{A}^{n}_{\rho}$ for $\rho \in (0,R]$ and its image is given by functions that are holomorphic with respect to $x$ in $(\C^{\ast})^{m}$ and meromorphic with respect to $a$.
\end{corollary}


\section{The $E_8$ symmetry}
\label{sec_E8_symmetry}

The van Diejen model, as defined by the principal operator, admits a symmetry under the action of the Weyl group associated to the Lie algebra $D_8$. A symmetry under the action of the $E_{8}$ Weyl group is known in the univariate case from the work of Ruijsenaars \cite{Rui15} and alluded to for the multivariate case in \cite{Rui09a}. In this Section, we show how the gauge transformations $U_{K}$ in Lemma~\ref{lemma_gauge_transform_U} and the Cauchy-type integral transform in Theorem~\ref{thm_Cauchy_eigenfunction_transform_Type_II} are related to the $E_8$ Weyl group reflections on the parameter space.

\subsection{The $E_{8}$ root lattice and Weyl group.}
Before proceeding, let us recall some basic facts about the root lattice of type $E_8$. We mainly follow the conventions in \cite{Nou18}. Denote by $\cV=\C^8=\C \epsilon_0 \oplus \C \epsilon_{1} \oplus \cdots \oplus \C \epsilon_{7}$ a complex vector space with canonical basis $\{\epsilon_{0},\ldots,\epsilon_{7}\}$ and by $( \cdot \lvert \cdot ): \cV \times \cV \to \C$ the symmetric bilinear form satisfying $( \epsilon_i \lvert \epsilon_j ) = \delta_{i,j}$ ($i,j\in \{0,1,\ldots,7\}$). Setting $\phi = (\half , \half ,\ldots, \half)\in\cV$, we realize the root lattice $P=Q(E_{8})$ and root system $\Delta(E_{8})$ of type $E_{8}$ as
\begin{equation}
P = \{ v \in \Z^{8} \cup ( \phi + \Z^{8}) \ \lvert \ ( \phi \lvert v ) \in \Z\} \subset \cV, \quad \Delta(E_{8}) = \{ \alpha \in P \lvert ( \alpha \lvert \alpha)=2\}.
\end{equation}
The root system consists of two different classes of vectors
\begin{align}
(1)&\quad \pm \epsilon_{r} \pm \epsilon_{s} \quad (0 \leq r < s \leq 7) \\
(2)&\quad  \half( \pm \epsilon_{0} \pm \epsilon_{1} \pm \cdots \pm \epsilon_{7}) \quad (\text{with even number of minus signs}).
\end{align}
We then have the \emph{simple roots}
\begin{equation}
\alpha_{0} = \phi - \epsilon_{0} - \epsilon_{1} - \epsilon_{2} - \epsilon_{3} , \quad \alpha_{j} = \epsilon_{j} - \epsilon_{j+1} \ (j\in\{1,2,\ldots,6\}), \quad \alpha_{7} = \epsilon_{7} + \epsilon_{0},
\label{eq_simple_roots}
\end{equation}
corresponding to the Dynkin diagram in Fig.~\ref{fig_Dynkin_diagram_E8},
\begin{figure}[h!]
\label{fig_Dynkin_diagram_E8}
\begin{picture}(220,45)(-20,-5)
\multiput(0,0)(30,0){7}{\circle{4}}
\put(60,25){\circle{4}}
\multiput(2,0)(30,0){6}{\line(1,0){26}}
\put(60,2){\line(0,1){21}}
\put(57,32){\small $\alpha_0$}
\put(-3,-12){\small $\alpha_1$}
\put(27,-12){\small $\alpha_2$}
\put(57,-12){\small $\alpha_3$}
\put(87,-12){\small $\alpha_4$}
\put(117,-12){\small $\alpha_5$}
\put(147,-12){\small $\alpha_6$}
\put(177,-12){\small $\alpha_7$}
\end{picture}
\caption{The Dynkin diagram corresponding to the $E_{8}$ root system.}%
\end{figure}
such that 
\begin{equation}
P = Q(E_{8}) = \Z \alpha_{0} \oplus \cdots \oplus \Z \alpha_{7}.
\end{equation}
We also have the highest root
\begin{equation}
\phi = 3 \alpha_{0} + 2\alpha_{1} + 4 \alpha_{2} + 6\alpha_3 + 5 \alpha_4 + 4 \alpha_5 +3 \alpha_6 + 2 \alpha_7
\end{equation}
with respect to the simple roots. 

For each $\alpha \in\cV \setminus \{\alpha \in\cV \lvert ( \alpha \lvert \alpha ) = 0\}$ we define the reflection with respect to $\alpha$ by
\begin{equation}
\r_{\alpha} : \cV \to \cV , \quad \r_{\alpha}(v) = v - 2 \frac{(\alpha \lvert v)}{(\alpha \lvert \alpha)} \alpha \quad ( v \in \cV)
\end{equation}
We denote by $\r_{i}$  ($i \in\{0,1,\ldots,7\}$) the simple reflections corresponding to the simple roots, \emph{i.e.} $\r_{i}= \r_{\alpha_i}$ for all $i \in\{0,1,\ldots,7\}$. Then the Weyl group $W(E_{8})$ is generated by the simple reflections $\r_{i}$ ($i \in\{0,1,\ldots,7\}$), \emph{i.e.} $W(E_{8}) = \langle \r_{0},\r_{1},\ldots,\r_{7}\rangle$, that satisfy the fundamental relations $\r_{i}^{2} = 1$, $\r_{i}\r_{i} = \r_{j} \r_{i}$ for distinct $i,j\in\{0,1,\ldots,7\}$ such that $(\alpha_{i}\lvert\alpha_{j})=0$, and $\r_{i} \r_{j} \r_{i} = \r_{j} \r_{i} \r_{j}$ for distinct $i,j\in\{0,1,\ldots,7\}$ such that $(\alpha_{i}\lvert\alpha_{j})=-1$. 

We remark that $W(E_8)=\langle \r_0, \r_1,\ldots,\r_7\rangle$ contains the following Weyl groups of type $D_8$ and of type $A_7$: 
\begin{equation}
W(E_8)
\supset
W(D_8)=\langle \r_1,\ldots,\r_7,\r_8\rangle
\supset
W(A_7)=\langle \r_1,\ldots,\r_6,\r_8\rangle
\end{equation}
where $\r_8=\r_{\alpha_8}$ denotes the reflection by the root $\alpha_8=\epsilon_7-\epsilon_0$.\footnote{We have that, \emph{e.g.}, $\r_8 = \r_0 \r_3 \r_4 \r_5 \r_2 \r_3 \r_4 \r_1 \r_2 \r_3 \r_0 \r_3 \r_4 \r_5 \r_2 \r_3 \r_4 \r_1 \r_2 \r_3 \r_0$.}
The subgroup $W(D_8)$ corresponds to the Dynkin diagram in Fig.~\ref{fig_Dynkin_diagram_D8}.  
\begin{figure}[h!]
\begin{picture}(220,45)(-20,-10)
\multiput(0,0)(30,0){7}{\circle{4}}
\put(150,25){\circle{4}}
\multiput(2,0)(30,0){6}{\line(1,0){26}}
\put(150,2){\line(0,1){21}}
\put(147,32){\small $\alpha_8$}
\put(-3,-12){\small $\alpha_1$}
\put(27,-12){\small $\alpha_2$}
\put(57,-12){\small $\alpha_3$}
\put(87,-12){\small $\alpha_4$}
\put(117,-12){\small $\alpha_5$}
\put(147,-12){\small $\alpha_6$}
\put(177,-12){\small $\alpha_7$}
\end{picture}
\label{fig_Dynkin_diagram_D8}
\caption{The Dynkin diagram corresponding to the $D_{8}$ root system.}%
\end{figure}\newline
Note that $W(A_7)=\langle \r_1,\ldots,\r_6,\r_8\rangle$ represents 
the permutation group $\mathfrak{S}_8$ of the vectors $\epsilon_i$ 
($i\in\{0,1,\ldots,7\}$) and that 
$W(D_8)=\langle \r_1,\ldots,\r_7,\r_8\rangle$ provides 
the extension of $\mathfrak{S}_8$ by the even sign changes of $\epsilon_i$. 
Note also that 
\begin{equation}
\Delta(D_8)=\{\pm \epsilon_i\pm \epsilon_j \ |\ i,j\in\{0,1,\ldots,7\},\ i < j \}.
\end{equation}

We denote by $u=(u_0,u_1,\ldots,u_7)$ the canonical coordinates of $\cV=\C^8$.  Then the action of the reflections $\r_\alpha$ ($\alpha\in\Delta(E_8))$ on $\mathcal{V}$ are explicitly described as follows:  
The first class of $E_8$ roots are in fact the $D_8$ roots, and their actions are given by 
\begin{equation}
\r_{\epsilon_i\mp\epsilon_j}.(u_0,\ldots,u_7)=(v_0,\ldots,v_7),
\quad 
v_k
=\begin{cases}
\ \pm u_j\quad &(k=i)\\
\ \pm u_i\quad &(k=j)\\
\ u_k\quad &(k\ne i,j)
\end{cases},
\end{equation}
with double signs in the same order.  
The second class of $E_8$ roots are expressed as 
\begin{equation}
\alpha=\frac{1}{2}(-\sum_{i\in K}\epsilon_i+\sum_{i\not\in K}\epsilon_i) 
=\phi-\epsilon_K,\quad \epsilon_K=\sum_{i\in K}\epsilon_i
\end{equation}
with a subset $K\subseteq\{0,1,\ldots,7\}$ of indices with even cardinality.  
The reflections of this class are given by 
\begin{eqnarray}
&&\r_{\phi-\epsilon_K}.(u_0,\ldots,u_7)=(v_0,\ldots,v_7), \quad
\\
&&v_k=
\begin{cases}
\ u_k-\frac{1}{4}(\sum_{i\in K}u_i-\sum_{j\notin K}u_j)\quad&(k\in K),
\\[4pt]
\ u_k+\frac{1}{4}(\sum_{i\in K}u_i-\sum_{j\notin K}u_j)\quad&(k\notin K).  
\end{cases}
\nonumber
\end{eqnarray}

Passing to the multiplicative variables, we denote by $\mu=(\mu_0,\mu_1,\ldots,\mu_7)$ the canonical coordinates of the 8-dimensional algebraic torus $(\C^\ast)^8$.  
Then, we obtain the following actions of $W(E_8)$ on $(\C^\ast)^8$:  
For each distinct pair $i,j\in\{0,1,\ldots,7\}$, 
\begin{equation}
\r_{\epsilon_i\mp\epsilon_j}.(\mu_0,\ldots,\mu_7)=(\nu_0,\ldots,\nu_7),
\quad 
\nu_k
=\begin{cases}
\ \mu_j^{\pm1 }\quad &(k=i)\\
\ \mu_i^{\pm 1}\quad &(k=j)\\
\ \mu_k\quad &(k\ne i,j)
\end{cases},
\end{equation}
and for each subset $K\subseteq\{0,1,\ldots,7\}$ with even cardinality, 
\begin{eqnarray}
&&\r_{\phi-\epsilon_K}.(\mu_0,\ldots,\mu_7)=(\nu_0,\ldots,\nu_7),\quad
\\
&&\nu_k=
\begin{cases}
\ \mu_k\,\prod_{i\in K}\mu_i^{-\frac{1}{4}}{\prod_{j\notin K}\mu_j^{\frac{1}{4}}}
\quad&(k\in K),
\\[4pt]
\ \mu_k\,\prod_{i\in K}\mu_i^{\frac{1}{4}}{\prod_{j\notin K}\mu_j^{-\frac{1}{4}}}
\quad&(k\notin K).  
\end{cases}
\nonumber
\end{eqnarray}
{\smallskip}

It follows from the definition that the van Diejen operator $\cD_x(a|q,t)$ is invariant 
under permutations of the parameters $a_s$, \emph{i.e.} invariant under the action of $W(A_7)=\mathfrak{S}_8$. Let $a = p^{\half} q^{\half} \mu$, then $\cD_x(a_{\{i,j\}(pq)}|q,t)=\cD_x(\r_{\epsilon_i+\epsilon_j}.a|q,t)$ since $\cD_x(a|q,t)$ is invariant under the action of $\mathfrak{S}_8$ and it follows from Lemma~\ref{lemma_gauge_transform_U} that the $D_8$ reflections $\r_{\epsilon_i+\epsilon_j}$ ($i\neq j$) correspond to the 
gauge transformation $U_{K}(x;a|q)^{-1} \circ \cD_x(a|q,t) \circ U_{K}(x;a|q)$ 
for $K=\{i,j\}$,  \emph{i.e.} 
\begin{equation}
\label{eq_D8_reflection}
U_{\{i,j\}}(x;a|q)^{-1}\circ \cD_x(a|q,t) \circ U_{\{i,j\}}(x;a|q)=
\cD_{x}(r_{\epsilon_i+\epsilon_j}.a|q,t).
\end{equation}
For a general subset $K\subseteq\{0,1,\ldots,7\}$ with even cardinality $|K|=2r$ we can express $K$ as $K=\{i_1,\ldots,i_r,j_1,\ldots,j_r\}$ and the gauge transformation by $U_K(x;a|q)$ then
corresponds to the transformation of parameters by
\begin{equation}
\label{eq_W8_element}
{w}_K=\r_{\epsilon_{i_1}+\epsilon_{j_1}}\cdots \r_{\epsilon_{i_r}+\epsilon_{j_r}}
\in W(D_8)
\end{equation}
and it is clear that
\begin{equation}
U_{K}(x;a|q)^{-1} \circ \cD_x(a|q,t) \circ U_{K}(x;a|q)=\cD_x(w_K.a|q,t).  
\end{equation}
This implies that, if $\varphi(x;b|q,t)$ is an eigenfunction of $\cD_x(x;b|q,t)$, 
$b=w_{K}.a$, 
with eigenvalue $\Lambda$, then 
\begin{equation}
\label{eq_D8_eigfunc_transf}
\psi(x;a|q,t)=U_K(x;a\lvert q)\varphi(x;b|q,t)
\end{equation}
is an eigenfunction of $\cD_x(a|q,t)$ with the same eigenvalue, as in Proposition~\ref{prop_gauge_transf_U}.  

The second class of $E_8$ reflections on the parameter space can be associated with 
integral transforms of Cauchy type.  
We first look at the reflection $\r_\phi$ by the highest root $\phi$. 
It is related to the $m=n$ gauge-integral transform of Cauchy type given by 
\begin{equation}
\psi(x;a\lvert q , t) = \int_{\mathbb{T}^{m}}d\omega_{m}(y) \,w_{m}(y; b \lvert q , t) \Phi(x,y\lvert q , t) \,U_{\{0,1,\ldots,7\}}(y;b|q)\varphi(y;pq b^{-1}\lvert q,t)
\end{equation}
under the balancing condition $a_0\cdots a_7=p^2 q^2 t^2$,  where $b=p^{\half} q^{\half} t^{\half}a^{-1}$. 
In fact, we have
\begin{equation}
(pq b^{-1})_s=p^{\half}q^{\half}t^{-\half}a_s=a_s pq(a_0\cdots a_7)^{-\frac{1}{4}}
=(\r_{\phi}.a)_s\quad(s \in \{0,1,\ldots,7\}).  
\end{equation}
Similarly, for a subset $K\subseteq\{0,1,\ldots,7\}$ with even cardinality, 
the $E_8$ reflection $\r_{\phi-\epsilon_K}$ can be attained by the gauge-integral transform 
of the form 
\begin{equation}
\label{eq_E8_transform}
\psi(x;a\lvert q , t) = U_K(x;a|q)\int_{\mathbb{T}^{m}}d\omega_{m}(y) \,w_{m}(y; b \lvert q , t) \Phi(x,y\lvert q , t) \,U_{K^{\mathrm{c}}}(y;b|q)\varphi(y;d \lvert q,t),
\end{equation}
under the balancing condition 
\begin{equation}
\label{eq_E8_restriction}
\prod_{s\in K}a_s^{-1} \prod_{s\notin K}a_s=(pq)^{2-|K|}t^2, 
\end{equation}
where $K^\mathrm{c}=\{0,1,\ldots,7\}\backslash K$, 
$b=p^{\half} q^{\half} t^{\half} a_{K(pq)}^{-1}$, and $d =b_{K^{\mathrm{c}}(pq)}$. This implies that, if $\varphi(y;d \lvert q,t)$ is an eigenfunction of $\cD_x(d \lvert q , t)$ with eigenvalue $\Lambda$, then $\psi(x;a \lvert q , t)$ \eqref{eq_E8_transform} is an eigenfunction of $\cD_{x}(a \lvert q, t)$ with the same eigenvalue, as proven in Theorem~\ref{thm_Cauchy_eigenfunction_transform_Type_II}.

To summarize, we have the following Theorem:

\begin{theorem} Let $K\subseteq\{0,1,\ldots,7\}$ with even cardinality. \newline{} 
$(1)$  By the transformation in \eqref{eq_D8_eigfunc_transf}, the model parameters of eigenfunctions of the van Diejen operator transform according as the action of the element $w_{K}\in W(D_{8})$ in \eqref{eq_W8_element}.
\newline{}$(2)$ Let $t$ satisfy the restriction \eqref{eq_E8_restriction}. By the transformation in \eqref{eq_E8_transform}, the model parameters of eigenfunctions of the van Diejen operator transform according as the action of the reflection $\r_{\phi- \epsilon_{K}}$ by the $E_{8}$ root $\phi - \epsilon_{K}$.
\end{theorem}
{\smallskip}
It is also worth noting that we can construct an operator which is invariant under the transformations $a\to a_{K(pq)}$, for all $K\subseteq\{0,1,\ldots,7\}$ with even cardinality, and $t\to p q t^{-1}$: Define
\begin{equation}
G(x;a\lvert q, t )= \prod_{1 \leq i \leq m} \prod_{0\leq s \leq 7} \Gamma(a_{s} x_{i};p,q) \prod_{1 \leq i < j \leq m} \Gamma( t x_{i} x_{j}^{\pm};p,q),
\end{equation}
then it is straightforward to check that 
\begin{equation}
U_{K}(x;a\lvert q ) = \frac{G(x ; w_{K}.a\lvert q, t)}{G(x ; a \lvert q ,t)}, \quad V(x\lvert q, t ) = \frac{G(x ;a \lvert q , p q t^{-1})}{G(x;a \lvert q, t)}
\label{eq_gauge_relations}
\end{equation}
for any $K\subseteq\{0,1,\ldots,7\}$ with even cardinality. It follows from Lemmas~\ref{lemma_gauge_transform_V} and \ref{lemma_gauge_transform_U} and Eq.~\eqref{eq_gauge_relations} that
\begin{multline}
G(x;a\lvert q, t) \circ \cD_{x}(a \lvert q , t) \circ G(x;a\lvert q, t)^{-1} \\
= G(x; w_K.a\lvert q , t)\circ  {\cD}_{x}(w_{K}.a \lvert q , t) \circ G(x; w_{K}.a\lvert q ,t)^{-1}  \\
= G(x;a \lvert q, pqt^{-1}) \circ  \cD_{x}(a \lvert q , p q t^{-1}) \circ G(x;a\lvert q, p q t^{-1})^{-1} 
\end{multline}
for all $K\subseteq\{0,1,\ldots,7\}$ with even cardinality by using the reflection property. Then it is clear that the operator $\widetilde{\cD}_{x}(a\lvert q , t) = G(x;a\lvert q ,t)\circ \cD_{x}(a\lvert q ,t) \circ G(x;a\lvert q , t)^{-1}$, and consequently also the eigenvalues of the van Diejen operator, is invariant under the action of $W(D_{8})\times W(A_{1})$. (Note that $\widetilde{\cD}_{x}(a\lvert q , t)$ is not $W_{m}$-invariant.)

\section{Conclusions and outlook}
\label{sec_conclusions_outlook}

In this paper, we have introduced several transformations for the eigenfunctions of the principal van Diejen Hamiltonian. In particular, we obtained various exact eigenfunctions of the van Diejen operator which can be expressed as both Type I and Type II $BC_{n}$ elliptic hypergeometric integrals. Using these transformations from the known eigenfunctions in Section~\ref{sec_previous_results}, or the simple explicit functions in Section~\ref{sec_special_cases_eigenfunction}, it is possible to construct several different eigenfunctions of the van Diejen operator under certain parameter restrictions. We presented two of the many possible transforms where the eigenfunctions are given by the elliptic hypergeometric integrals of Selberg type in Theorems~\ref{thm_Selberg_type_II} and \ref{thm_Selberg_type_I} starting from the constant eigenfunction $1$. Another interpretation of these results is that the $BC_{n}$ elliptic hypergeometric integrals of Selberg type are governed by eigenvalue equations for a van Diejen operator with particular eigenvalues. Our results provide even further motivation for the study of elliptic hypergeometric integrals and for finding new evaluation/transformation formulas.

Also, it is an open important question whether the construction of general eigenfunctions can be obtained by suitable combinations/iterations of integral transformations. The main difficulty in this approach is to avoid contradictory balancing conditions while applying the sequence of transformations. It would also be an intriguing problem to describe the totality of eigenfunctions obtained by the gauge and integral transformations, starting from the constant eigenfunctions or the free eigenfunctions.

The full family of commuting analytic difference operators for the van Diejen model was constructed by Hikami and Komori \cite{KH97}. We believe that all these analytic difference operators satisfy the same symmetry relations that we have found in this paper, although it is not a priori clear to us from their construction, and that our transformations yield joint eigenfunctions of the van Diejen model.

{\medskip}
The van Diejen model also has a physics interpretation of relativistic particles in one dimensions when the model parameters, including the shift parameter and elliptic nom\'e, are restricted such that the weight function $w_{m}(x;a\lvert q, t)$ is positive definite and the operator $\cD_{x}(a\lvert q, t)$ is formally self-adjoint with respect to the weighted $L^{2}( \mathbb{T}^{m}; w_{m}(x;a\lvert q, t) d\omega_{m}(x))$ inner product. Although we have not emphasized it in the paper, our results can also be applied for these model parameters. Using our results, it could be possible to find a suitable Hilbert space where the van Diejen operators are self-adjoint and obtain interpretations of the eigenfunctions as quantum mechanical wave functions that diagonalize the van Diejen Hamiltonian. Self-adjoint extensions of the van Diejen Hamiltonian was found by Ruijsenaars in the univariate case \cite{Rui15} by applying the integral transform of Cauchy type, which required further conditions on the model parameters. It would be interesting to explore whether our transformations can be used in either the construction of the suitable Hilbert space or its analytic continuation, with respect to the parameter space, for general $m\in\Z_{\geq 0}$.

{\medskip}
An interesting generalization of the models of Calogero-Moser-Sutherland (CMS) type was found by Chalykh, Feigin, Sergeev, and Veselov \cite{CFV98,Ser02} with application in quantum field theory and super-symmetric gauge theories; see \cite{AHL21} and references therein. Such generalizations are also known for the Macdonald-Ruijsenaars models \cite{SV09b,FS14,AHL14} and was recently found by one of the authors for the van Diejen model as well \cite{Ata20}. The weight function and kernel functions for this so-called \emph{deformed van Diejen model} was constructed in \cite{Ata20}, which also suggest that similar type of gauge and integral transformations exists for the deformed van Diejen model and its eigenfunctions can be expressed in terms of (possibly generalized) $BC$-type elliptic hypergeometric integrals. In particular, a kernel function linking the standard van Diejen operator to the deformed operator was explicitly constructed in \cite{Ata20} which would yield exact eigenfunctions of the deformed model in terms of Type II elliptic hypergeometric integrals of Selberg type. It would be interesting to explore this further, however, recent results in \cite{AHL21} suggests that the investigation of appropriate integration cycles are crucial for the deformed models.

\subsection{Higher symmetries}
In this paper we have presented several symmetries of the van Diejen model in the parameter space $(a \lvert q, t)\in(\C^{\ast})^{8} \times (\C^{\ast})^{2}$. Together with the known symmetries, these include (0) $\mathfrak{S}_{8}$-symmetry in $a$, (1) $p$-shifts in the parameters $a$ by $p^{P}$, (2) the transformation $(a \lvert q, t)\leftrightarrow(a \lvert q, p q t^{-1})$ by the gauge function $V(x\lvert q , t)$ \eqref{eq_gauge_function_V}, (3) the transformation $(a \lvert q, t) \leftrightarrow (a_{K(pq)} \lvert q, t)$ for $K\subseteq \{0,1,\ldots,7\}$ with even cardinality by the gauge function $U_{K}(x;a\lvert q)$ \eqref{eq_gauge_function_U}, (4) the integral transform of Cauchy type where $(a \lvert q , t) \leftrightarrow ( p^\half  q^\half t^\half a^{-1} \lvert q,t)$, and (5) the integral transform of dual-Cauchy type where  $(a \lvert q , t) \leftrightarrow ( a \lvert t , q )$. 
We showed in Section~\ref{sec_E8_symmetry} that the action of the Weyl group $W(D_{8})$ can be obtained through permutation symmetry and the $U_{K}$ gauge transformations. Combining these with the $m=n$ gauge-integral transformations of Cauchy-type allowed us to find a relation between the transformations and the $E_8$ Weyl group $W(E_{8})$. It is unclear if the symmetries of the van Diejen operator can be extended to a larger symmetry group by including the other transformations mentioned above.

\subsection{Different forms of the van Diejen operator.}
Finally, we wish to point out that the (informal) definition of the zeroth order coefficient in Remark~\ref{remark_zeroth_order_coefficient} allows us to express the van Diejen operator in a similar form as \eqref{eq_vD_operator_balanced} by finding two (meromorphic) functions $\phi^\pm(x;a \lvert q ,t )$ such that 
\begin{equation}
B^{0}(x;a\lvert q , t) = -\sum_{1\leq i \leq m} A^{+}_{i}(x;a\lvert q ,t) \phi^{+}(x_{i} ;a \lvert q , t) + A^{-}_{i}(x;a\lvert q ,t) \phi^{-}(x_{i} ;a \lvert q , t)
\end{equation}
satisfies the conditions in Remark~\ref{remark_zeroth_order_coefficient}. This entails that 
the functions $\phi^{\pm}$ satisfy 
\newline
(1) $\phi^{\pm}(p x ; a \lvert q , t) = ( p^{2} q^{2} t^{2}/a_{0}\cdots a_{7} t^{2m})^{\mp1} \phi^{\pm}( x ; a \lvert q , t)$, 
\newline
(2) $\phi^{+}(x^{-1};a\lvert q, t) = \phi^{-}(x ;a\lvert q, t)$, 
\newline
(3) $\phi^{-\ve}( \varpm p^{\half \ell} q^{\half \ve} ; a \lvert q , t) = ( p^{2} q^{2} t^{2} / a_{0} \cdots a_{7} t^{2m})^{\half \ve \ell}$ for all $\ell\in\Z$ and $\ve\in\{\pm\}$, 
\newline
(4) $B^{0}(x;a\lvert q, t)$ should only have poles at $x_{i} \in \varpm p^{\half \Z} q^{\half \ve}$ ($i\in\{1,\ldots,m\}$; $\ve\in\pm$). 
The van Diejen operator can then be expressed as
\begin{equation}
\cD_{x}(a\lvert q, t) = \sum_{1 \leq i \leq m} \sum_{\ve=\pm} A_{i}^{\ve}(x;a\lvert q , t) \bigl( T_{q,x_{i}}^{\ve} - \phi^{\ve}(x_{i} ; a \lvert q ,t ) \bigr) + C
\end{equation}
for some constant $C$. We find that a class of solutions to this problem can be readily obtained by the gauge functions: It follows from straightforward calculations that 
\begin{equation}
\phi^{+}(x_{i} ; a\lvert q, t )= (T_{q,x_{i}} U_{K}(x;a\lvert q ))/U_{K}(x;a\lvert q) \quad ( K\subseteq \{0,1,\ldots,7\})
\end{equation}
fulfil the conditions above if $K$ has even cardinality and the model parameters satisfy \eqref{eq_balancing_condition_U}. (Setting $K=\emptyset$ yields that $\phi^{+}=1$ and we obtain \eqref{eq_vD_operator_balanced} under the ellipticity condition \eqref{eq_ellipticity_condition_1}.) We also find that 
\begin{equation}
\phi^{+} = (T_{q,x_{i}} U_{K}(x;a\lvert q ) V(x\lvert q, t) )/ (U_{K}(x;a\lvert q) V(x\lvert q , t))
\end{equation}
($K\subseteq \{0,1,\ldots,7\}$ with even cardinality) is another solution when the model parameters satisfy 
\begin{equation}
( \prod_{s\in K} a_{s}^{-1})( \prod_{s\notin K} a_{s}) p^{2m+ \abs{K}} q^{2m+ \abs{K}} t^{-2m} = p^{4} q^{4} t^{-2}
\end{equation}
even though the solution is a function of all $x$-variables.
It would be interesting to check whether there are other (non-transcendental) solutions to this problem, apart from the ones above.

\section*{Acknowledgements}

The second author (M.N.) would like to thank the mathematics department of KTH for their warm hospitality during his stay in 2020-2021. M.N. is grateful to the Knut and Alice Wallenberg Foundation for funding his guest professorship at KTH where part of this work was conducted. M.N. gratefully acknowledges partial financial support by JSPS Kakenhi Grants (S) 17H06127 and (B) 18H01130.

\bigskip\noindent
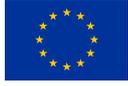
\parbox{.135\textwidth}{\begin{tikzpicture}[scale=.03]
\fill[fill={rgb,255:red,0;green,51;blue,153}] (-27,-18) rectangle (27,18);  
\pgfmathsetmacro\inr{tan(36)/cos(18)}
\foreach \i in {0,1,...,11} {
\begin{scope}[shift={(30*\i:12)}]
\fill[fill={rgb,255:red,255;green,204;blue,0}] (90:2)
\foreach \x in {0,1,...,4} { -- (90+72*\x:2) -- (126+72*\x:\inr) };
\end{scope}}
\end{tikzpicture}} \parbox{.85\textwidth}{This project has received funding from the European Union's Horizon 2020 research and innovation programme under the Marie Sk{\l}odowska-Curie grant agreement No 895029.}

\appendix
\section{The theta and Gamma functions}
\label{sec_theta_and_gamma_functions}

We recall that the multiplicative theta function with base $p\in\C$ is given by 
\begin{equation}
\theta(x ; p) = \prod_{\ell \in \Z_{\geq 0}}( 1 - p^{\ell} x ) ( 1 - p^{\ell +1 } x^{-1}) = ( x , p x^{-1} ; p)_{\infty} \quad (\abs{p} < 1).
\label{eq_theta_function}
\end{equation}
For $\abs{p}<1$, we have that $\theta(x ; p)$ is analytic for $x \in \C^{\ast}$ and has (simple) zeroes at $x \in p^{\Z}$.
Using \eqref{eq_theta_function}, it is straightforward to check that the theta function satisfies
\begin{equation}
\label{eq_theta_properties}
\theta( p x ; p ) = \theta(x^{-1} ; p)  = - (1/x) \theta(x ; p), \quad \theta( x ; p ) = \theta( p x^{-1} ; p).
\end{equation}
A useful identity for the theta function is the duplication formula
\begin{equation}
\theta(x^2 ; p ) = \theta( x , - x , p^{\half} x , -p^{\half} x ; p).
\label{eq_double_angle_theta}
\end{equation}

The elliptic Gamma function $\Gamma(x ; p , q )$ is defined as the common solution to the analytic difference equations
\begin{equation}
\Gamma( p x ; p , q) = \theta(x ; q) \Gamma (x ; p , q) , \quad \Gamma(q x ; p , q ) = \theta(x;p) \Gamma(x ; p , q).
\label{eq_Gamma_difference_eq}
\end{equation}
Ruijsenaars introduced the so-called minimal solutions to these analytic difference equations \cite{Rui97} given by
\begin{equation}
\Gamma(x ; p , q ) = \prod_{i,j \in \Z_{\geq 0}} \frac{1 - p^{i+1} q^{j+1} x^{-1}}{ 1 - p^{i} q^{j} x } = \frac{(p q x^{-1} ; p ,q )_{\infty} }{( x ; p , q )_{\infty}}.
\label{eq_elliptic_Gamma_function}
\end{equation}
The elliptic Gamma function is a meromorphic function of $x \in \C^{\ast}$ for $\abs{p}<1$ and $\abs{q}<1$, that is invariant under interchanging $p \leftrightarrow q$, with zeroes at $x \in p^{\Z_{>0}} q^{\Z_{>0}}$ and poles at $x \in p^{\Z_{\leq 0}} q^{\Z_{\leq0}}$. Moreover, it is straightforward to check that it satisfies the reflection property
\begin{equation}
\Gamma(p q x^{-1} ;p , q) = 1 / \Gamma(x ; p , q).
\label{eq_reflection_property}
\end{equation}
The elliptic Gamma function also satisfies the duplication formula
\begin{equation}
\Gamma( x^{2} ; p , q) = \Gamma( x , -x , p^{\half} x , - p^{\half} x, q^{\half} x , -q^{\half} x , p^{\half} q^{\half} x , - p^{\half} q^{\half} x ; p , q).
\end{equation}
These duplication formulas are useful when considering the free case in Section~\ref{sec_special_cases_eigenfunction}.

\section{Elliptic hypergeometric series and integrals}
\label{app_elliptic_SF}

Some known relations between the elliptic hypergeometric series \cite{DJKMO88,FT97} and the elliptic hypergeometric integrals are presented here, as well as some of the known evaluations of the elliptic Selberg integrals of Type I and Type II.
\subsection{The elliptic hypergeometric series}

The elliptic hypergeometric series ${}_{12}V_{11}(b_{0} ; b_{1}, \ldots,b_{7} ; q , p)$ \cite{DJKMO88,FT97} is formally defined by
\begin{equation}
{}_{12}V_{11}(b_{0} ; b_{1}, \ldots, b_{7} ; q ,  p) = \sum_{\ell\in\Z_{\geq 0}} \frac{\theta( b_{0} q^{2 \ell}  ; p)}{\theta( b_{0} ; p )} \frac{\theta( b_{0}, \ldots, b_{7} ; p )_{q,\ell}}{\theta( q , b_{0} q / b_{1}, b_{0} q / b_{2} , \ldots, b_{0} q / b_{7} ; p )_{q,\ell}} q^{\ell},
\end{equation}
where 
\begin{equation} 
\theta( x ; p )_{q,\ell} = \frac{ \Gamma( q^{\ell} x ; p , q) }{ \Gamma( x ; p , q)} = \theta(x ; p) \theta(q x ; p) \cdots \theta(q^{\ell-1} x ; p) \quad (\ell\in\Z_{\geq 0}),
\end{equation}
under the parameter condition $b_{1}^{2} \cdots b_{7}^{2} = q^{4} b_{0}^{6}$. 
It follows from the properties of the multiplicative theta function that the series is terminating if $b_{s} \in p^{\Z} q^{\Z_{\leq 0}}$ for some $s\in\{1,2,\ldots,7\}$, which then yields a (globally) meromorphic function of the remaining parameters $b_{0},\ldots,b_{7} \in \C^{\ast}$.

\subsection{The elliptic hypergeometric integral}
The elliptic hypergeometric integral $I( b_{0},\ldots, b_{7} ; p , q)$ is defined by
\begin{equation}
\label{eq_hypergeometric_integral}
I( b_{0},\ldots, b_{7} ; p , q) = \frac{(p;p)_{\infty} (q;q)_{\infty}}{2} \int_{\mathcal{C}_{1}} \frac{d y}{2 \pi \imag y} \frac{\prod_{0 \leq s \leq 7} \Gamma( b_{s} y^{\pm} ; p , q) }{\Gamma(y^{\pm2}; p q)} 
\end{equation}
with cycle $\mathcal{C}_{1}$ as in Section~\ref{sec_scalar_product} and the parameters satisfy $b_{r} b_{s} \notin p^{\Z_{\leq0}} q^{\Z_{\leq 0}}$ for all $r,s \in\{0,1,\ldots,7\}$. Note that \eqref{eq_hypergeometric_integral} differs from the elliptic hypergeometric integral presented in Section~\ref{sec_intro} by a multiplicative constant. When all the parameters satisfy $\abs{b_{s}}<1$ ($s\in\{0,1,\ldots,7\}$), then the cycle can be continuously deformed to be on the unit circle $\abs{y}=1$ (with positive orientation), as we have discussed in Section~\ref{sec_scalar_product}. If one of the parameters satisfy $b_{s} = p^{-M} q^{-N}$ for some $M,N\in\Z_{\geq 0}$, then the integral \eqref{eq_hypergeometric_integral} yields the elliptic hypergeometric series ${}_{12}V_{11}$. In particular, there are two relations that are of interest for our purposes: \newline (1) If the parameters satisfy the balancing condition $b_{0} \cdots b_{7} = p^{2} q^{2}$ and that either $p q b_{0}^{-1} b_{r}^{-1} \in q^{\Z_{\leq0}}$ for some $r\in\{1,2,\ldots, 6 \}$ or $p^{2} q b_{0}^{-1} b_{7}^{-1} \in p q^{\Z_{\leq0}}$, then
\begin{multline}
I( b_{0}, \ldots, b_{7} ; p , q) = \frac{\Gamma( p^{2} q^{2} b_{0}^{-2}, b_{0} b_{7}^{-1} ;p ,q)}{\prod_{r=1}^{6} \Gamma( p q b_{0}^{-1} b_{r},  p^{2} q b_{r}^{-1} b_{7}^{-1} ; p ,q) } \prod_{1 \leq r < s \leq 6} \Gamma( b_{r} b_{s} ; p , q) \\
\times  {}_{12}V_{11}( p^{2}q b_{0}^{-2} ; p q b_{0}^{-1} b_{1}^{-1},\ldots , p q b_{0}^{-1} b_{6}^{-1}, p^{2} q b_{0}^{-1} b_{7}^{-1} ; q , p ).
\end{multline}
(2) If $p q b_{0}^{-1} b_{r}^{-1} = p^{-M} q^{-N}$ for some $r\in\{1,2,\ldots, 7\}$ and $M,N\in\Z_{\geq 0}$, then \cite{Kom04}
\begin{align}
I( b_{0}, \ldots, b_{7} ; p ,q ) &= \frac{\Gamma(p^{2} q^{2} b_{0}^{-2} ; p , q) \prod_{0 \leq r < s \leq 7} \Gamma( b_{r} b_{s} ; p , q)}{\prod_{ 1 \leq r \leq 7} \Gamma( p q b_{0}^{-1} b_{r} ; p , q)} \\
&\times {}_{12}V_{11}( p^{2} q b_{0}^{-2}; p q b_{0}^{-1} b_{1}^{-1},\ldots, p q^{-N} , \ldots, p q b_{0}^{-1} b_{7}^{-1}; q , p) \nonumber \\
&\times  {}_{12}V_{11}( p q^{2} b_{0}^{-2}; p q b_{0}^{-1} b_{1}^{-1},\ldots, p^{-M} q ,\ldots , p q b_{0}^{-1} b_{7}^{-1} ;  p , q). \nonumber
\end{align}
(Note that these specializations of the parameters break the $p\leftrightarrow q$ symmetry of the integral unless $M=N$ and the other parameters are given generic values.)

These relations can be computed in a straightforward manner by using residue calculus and the transformations formulas due to Rains \cite{Rai10} and Spiridonov \cite{Spi04}, as explained in Proposition~5.4 of \cite{Nou18}.

Using these results, and the result of Theorem~\ref{thm_Cauchy_eigenfunction_transform_Type_II}, it is straightforward to check that the elliptic hypergeometric series is an eigenfunction of the van Diejen operator.

\subsection{Known evaluations of elliptic hypergeometric integrals}
\label{sec_known_evaluations}
The following $BC_n$ elliptic Selberg hypergeometric integral of type II has the known evaluation \cite{vDS01,Rai10}
\begin{multline}
\label{eq_Selberg_integral_Type_II}
\int_{\mathbb{T}^{n}} d\omega_{n}(y) \prod_{1 \leq k \leq n}\frac{\prod_{ 0 \leq s \leq 5} \Gamma( b_{s} y_{k}^{\pm}; p ,q )}{\Gamma( y_{k}^{\pm2} ; p , q)} \prod_{1 \leq k< l \leq n}\frac{\Gamma( t y_{k}^{\pm} y_{l}^{\pm};p,q) }{\Gamma( y_{k}^{\pm} y_{l}^{\pm} ; p , q) } \\
	= \frac{2^{n} n!}{(p;p)^{n}_{\infty} (q;q)^{n}_{\infty}} \prod_{ 0 \leq k \leq n-1} \frac{\Gamma(t^{k+1};p,q)}{\Gamma(t ; p , q)}\prod_{0 \leq r < s \leq 5} \Gamma( t^{k} b_{r} b_{s} ; p ,q)
\end{multline}
if parameters satisfy the condition $b_{0}\cdots b_{5} t^{2n} = p q t^{2}$.

The following $BC_n$ elliptic Selberg hypergeometric integral of type I has the known evaluation \cite{vDS01,Rai10}
\begin{multline}
\label{eq_Selberg_integral_Type_I}
\int_{\mathbb{T}^{n}} d\omega_{n}(y) \prod_{ 1 \leq k \leq n} \frac{\prod_{ 0 \leq s \leq 2n+3} \Gamma( b_{s} y_{k}^{\pm} ; p ,q) }{\Gamma(y_{k}^{\pm 2} ; p , q) } \prod_{1 \leq k < l \leq n} \frac{1}{\Gamma( y_{k}^{\pm} y_{l}^{\pm} ; p ,q )}\\
= \frac{2^{n} n! }{(p;p)_{\infty}^{n} (q;q)_{\infty}^{n} } \prod_{0 \leq r < s \leq 2n +3 } \Gamma( b_{r} b_{s} ; p , q)
\end{multline}
if the parameters satisfy $ b_{0} \cdots b_{2n+3} = p q$.

\section{List of eigenfunction transformations}
\label{app_transforms}

In this Section we collect the different transformations that can be obtained from Lemmas~\ref{lemma_gauge_transform_V} and \ref{lemma_gauge_transform_U} and Theorems~\ref{thm_Cauchy_eigenfunction_transform_Type_II}-\ref{thm_dual-Cauchy_eigenfunction_transform_Type_I}. In the following, we always assume that the parameters satisfy the necessary restrictions so that the integration cycle can be chosen as the $n$-dimensional torus $\mathbb{T}^n$.

\subsection{Eigenfunction transforms of Cauchy type} 
Fix two index sets $I,J\subseteq \{0,1,\ldots,7\}$ with even cardinality, \emph{i.e.} $\abs{I}, \abs{J} \in 2 \Z_{\geq0}$. (Note that we allow for these index sets to also be empty.) Suppose that $\varphi(y;b\lvert q, t)$ is an eigenfunction of the van Diejen operator $\cD_{y}(b\lvert q, t)$, then the functions
\begin{multline}
\label{eq_Cauchy_transf_eigenfunctions_1} 
\psi_{1}(x) = U_{I}(x;a \lvert q ) \int_{\mathbb{T}^{n}} d\omega_{n}(y) \ w_{n}(y ; b_{I(t)} \lvert q , t ) \ \Phi(x , y \lvert q , t ) \\
\cdot U_{J}(y ; b_{I(t)} \lvert q) \ \varphi(y ; (b_{I(t)})_{J(pq) }\lvert q ,t) ,
\end{multline}
\begin{multline}
\label{eq_Cauchy_transf_eigenfunctions_2}
\psi_{2}(x) = U_{I}(x ; a \lvert q) \int_{\mathbb{T}^{n}} d\omega_{n}(y) \ w_{n}(y ; b_{I(t)} \lvert q ,t ) \ \Phi( x, y \lvert q, t ) \\
\cdot U_{J}(y ; b_{I(t)} \lvert q ) \ V(y \lvert q , t ) \  \varphi(y ; (b_{I(t)})_{J(pq)}\lvert q , p q t^{-1} ) ,
\end{multline}
\begin{multline}
\label{eq_Cauchy_transf_eigenfunctions_3}
\psi_{3}(x) = U_{I}(x;a \lvert q) V(x \lvert q , t ) \int_{\mathbb{T}^{n}} d\omega_{n}(y) \ w_{n}(y ; d_{I(p q t^{-1})} \lvert q  ,p q t^{-1}) \ \Phi(x , y  \lvert q  , p q t^{-1} ) \\
\cdot  U_{J}(y ; d_{I(p q t^{-1})} \lvert q) \  \varphi(y ; (d_{I(p q t^{-1})})_{J(pq)} \lvert q ,p q t^{-1} ),
\end{multline}
and
\begin{multline}
\psi_{4}(x) = U_{I}(x;a\lvert q ) V(x \lvert q,t) \int_{\mathbb{T}^{n}} d\omega_{n}(y) \ w_{n}(y ; d_{I(p q t^{-1})} \lvert q ,p q t^{-1}) \ \Phi(x , y \lvert q , p q t^{-1}) \\
\cdot U_{J}(y ; d_{I(p q t^{-1})} \lvert q ) \ V(y \lvert q , p q t^{-1}) \ \varphi( y ; (d_{I(p q t^{-1})})_{J(p q)} \lvert q ,t ) ,
\label{eq_Cauchy_transf_eigenfunctions_4}
\end{multline}
where $b = p^{\half} q^{\half} t^{\half} a^{-1}$ and $d =p q t^{-\half} a^{-1}$, are eigenfunctions of the van Diejen operator $\cD_{x}{(a \lvert q , t )}$ if the parameters satisfy
\begin{equation}
(\prod_{s\in I} p q a_{s}^{-1} )(\prod_{s\notin I} a_{s}) t^{2(m-n)} = p^{2} q^{2} t^{2},
\end{equation}
resp.
\begin{equation}
(\prod_{s\in I} p q a_{s}^{-1} )(\prod_{s\notin I} a_{s}) p^{2(m-n)} q^{2(m-n)} t^{-2(m-n)} = p^{4} q^{4} t^{-2},
\end{equation}
in \eqref{eq_Cauchy_transf_eigenfunctions_1} and \eqref{eq_Cauchy_transf_eigenfunctions_2}, resp. \eqref{eq_Cauchy_transf_eigenfunctions_3} and \eqref{eq_Cauchy_transf_eigenfunctions_4}.
 (Throughout the Section, we have that $(a_{I(c_1)})_{J(c_2)} =(((a_{I(c_1)})_{J(c_2)})_{0} , \ldots, ((a_{I(c_1)})_{J(c_2)})_7)$ is given by 
\begin{equation}
((a_{I(c_1)})_{J(c_2)})_{s} = \begin{cases}
c_{2} c_{1}^{-1} a_{s} &\text{if } s\in I \bigcap J \\
c_{1} a_{s}^{-1}  &\text{if } s\in I \setminus J \\
c_{2} a_{s}^{-1} &\text{if } s\in J \setminus I \\
a_{s} &\text{if } s\notin I \bigcup J \\
\end{cases} \quad (s\in\{0,1,\ldots,7\}).
\end{equation}
for any vector $a\in(\C^{\ast})^{8}$, index sets $I,J\subseteq\{0,1,\ldots,7\}$, and constants $c_1,c_2 \in\C^{\ast}$.)

\subsection{Eigenfunction transforms of dual-Cauchy type}
Fix two index sets $I,J\subseteq \{0,1,\ldots,7\}$ with even cardinality and suppose that $\varphi(y;a \lvert t ,q  )$ is an eigenfunction of the van Diejen operator $\cD_{y}(a \lvert t , q)$, then the functions
\begin{multline}
\label{eq_dual_Cauchy_transf_1}
	\psi^{\vee}_{1}(x) = U_{I}( x ; a \lvert q) \int_{\mathbb{T}^{n}} d\omega_{n}(y) \  w_{n}( y ; a_{I(pq)} \lvert t, q )  \ \tPhi(x , y) \\ 
	\cdot U_{J}( y ; a_{I(pq)} \lvert t ) \ \varphi( y ; (a_{I(pq)})_{J(pt) } \lvert t , q ),
\end{multline}
\begin{multline}
\label{eq_dual_Cauchy_transf_2}
\psi^{\vee}_{2}(x) = U_{I}(x ; a \lvert q) \int_{\mathbb{T}^{n}} d\omega_{n}(y) \ w_{n}(y ; a_{I(pq)} \lvert t , q ) \ \tPhi(x , y ) \\ 
\cdot U_{J}(y ; a_{I(pq)} \lvert t ) \ V(y \lvert t , q ) \ \varphi(y ; (a_{I(pq)})_{J(pt)} \lvert t , p q^{-1} t ) ,
\end{multline}
\begin{multline}
\label{eq_dual_Cauchy_transf_3}
 \psi^{\vee}_{3}(x) = U_{I}( x ; a \lvert q ) V( x \lvert q  , t ) \int_{\mathbb{T}^{n}} d\omega_{n}(y) \ w_{n}( y ; a_{I(pq)} \lvert p q t^{-1} , q ) \ {\tPhi}(x , y ) \\
\cdot U_{J}( y ; a_{I(pq)} \lvert p q t^{-1} ) \ \varphi( y ; (a_{I(pq)})_{J(p^{2} q t^{-1})} \lvert p q t^{-1}  , q ),
\end{multline}
and
\begin{multline}
\label{eq_dual_Cauchy_transf_4}
\psi^{\vee}_{4}(x) = U_{I}(x ; a \lvert q) V( x \lvert q  , t ) \int_{\mathbb{T}^{n}} d\omega_{n}(y) \ w_{n}( y ; a_{I(pq)} \lvert p q t^{-1} , q ) \ {\tPhi}( x , y ) \\ 
\cdot U_{J}(y ; a_{I(pq)} \lvert p q t^{-1}) \  V(y \lvert p q t^{-1} , q ) \ \varphi(y ; (a_{I(pq)})_{J(p^{2} q t^{-1})} \lvert p q t^{-1} ,p^{2} t^{-1})
\end{multline}
are eigenfunctions of the van Diejen operator $\cD_{x}{(a \lvert q , t )}$ if the parameters satisfy
\begin{equation}
(\prod_{s\in I} p q a_{s}^{-1})(\prod_{s\notin I} a_{s}) q^{2n} t^{2m} = p^{2} q^{2} t^{2},
\end{equation}
resp.
\begin{equation}
(\prod_{s\in I} p q a_{s}^{-1}) ( \prod_{s\notin I} a_{s}) p^{2m} q^{2(n + m)} t^{-2m} = p^{4} q^{4} t^{-2},
\end{equation}
in \eqref{eq_dual_Cauchy_transf_1} and \eqref{eq_dual_Cauchy_transf_2}, resp. \eqref{eq_dual_Cauchy_transf_3} and \eqref{eq_dual_Cauchy_transf_4}.

\section{Relation to previous works}

\label{app_KFI}
\subsection{Relation to \cite{KNS09}}

In this Section, we will present the relation between the van Diejen operator constructed in \cite{KNS09} and our operator $\cD_x(a\lvert q, t)$. One possible way of relating such operators was already given in Appendix~B of \cite{KNS09}. Our approach in this paper differs as we make uses the result in Lemma~\ref{lemma_parameter_p_shift}. We also show how the kernel functions $\Phi(x,y \lvert q, t)$ and $\tPhi(x,y)$, and their corresponding kernel function identities, are obtained from the results in \cite{KNS09}.
\subsubsection{Preliminaries for the additive notation}

Explaining the relation requires the introduction, and specification, of the odd function $[u]$ used in \cite{KNS09}. 

The odd function $[u]$ is an entire function in $u \in \C$ satisfying the three-term relation
\begin{equation}
\label{eq_Weierstrass_relation}
[x \pm u][y\pm v] - [x \pm v] [y \pm u] = [x \pm y][u\pm v] \quad ( x,y,u,v\in\C)
\end{equation}
where $[u \pm v] = [u+v][u-v]$. Such functions are fall into three categories: rational, trigonometric/hyperbolic, or elliptic. In the elliptic case, $[u]$ coincides with the Weierstrass $\sigma$ function \cite{WhitWat} associated with a period lattice $\Omega = \Z \omega_1 \oplus \Z \omega_2$, where $\omega_1$, $\omega_2$ are linearly independent over $\R$, multiplied by a Gaussian term $\exp(c u^2 )$ and an overall constant. We also need to introduce the \emph{Legendre constants} $\eta_{\omega}$ ($\omega\in \Omega$) which are obtained from the quasi-periodicity relations 
\begin{equation}
[u + \omega ] = \epsilon_{\omega} \e^{2 \pi \imag \eta_{\omega} ( u + \half \omega)} [u] \quad ( \omega \in \Omega),
\end{equation}
for some $\epsilon : \Omega \to \{\pm\}$. The normalization of the Legendre constants are chosen such that $ \eta_{\omega_1} \omega_2 - \eta_{\omega_2} \omega_1 = 1$. We also find it useful to introduce 
\begin{equation}
\omega_0 = 0 \quad \omega_3 = -\omega_1 - \omega_2.
\end{equation}
Our specialization of the constants and periods are such that 
\begin{equation}
\label{eq_odd_function_specialization}
[u] = e( - \frac{1}{2} u ) \theta( e(u ) ; p ) , \quad e(u) = \exp( 2\pi \imag u), \quad p = e(\omega_2).
\end{equation}
In this specialization, we have that $\omega_1 = 1 $, with $\omega_2\in\C$ satisfying $\Im(\omega_{2}) > 0$, which also gives that $\epsilon_{\omega_{0}}=1$, $\epsilon_{\omega_{1}}=\epsilon_{\omega_{2}}=\epsilon_{\omega_{3}}=-1$ and $\eta_{\omega_{0}} = \eta_{\omega_{1}}=0$, $\eta_{\omega_{2}} = - \eta_{\omega_{3}} = -1$.

The elliptic Gamma functions, denoted by $G_{\ve}( u \lvert \delta)$ ($\ve\in\{\pm\}$) in \cite{KNS09}, are solutions to the analytic difference equations
\begin{equation}
G_{\ve}(u + \delta \lvert \delta) = \ve [u] G_{\ve}(u \lvert \delta) \quad (\ve \in \{\pm\})
\end{equation}
and, for our specialization, are given by
\begin{align*}
G_{+}(u \lvert \delta) = e( - \half \delta \binom{u / \delta}{2}) \Gamma( e(u) ; p , e(\delta)), \
G_{-}(u \lvert \delta) = e(\half \delta \binom{u/\delta}{2}) \Gamma( p e(u) ; p , e(\delta)).
\end{align*}

\subsubsection{The van Diejen operator in additive notation}
Having given the preliminaries, we are now in the position to proceed to the van Diejen operator. The principal van Diejen operator, in the additive notation, is given by the analytic difference operator \cite[Eq. (2.13)]{KNS09}
\begin{equation}
\label{eq_van_diejen_operator_additive_notation}
\mathcal{E}_{u}^{(\mu \lvert \delta , \kappa )} = \cV^{0}(u ; \mu \lvert \delta , \kappa) + \sum_{1 \leq i \leq m} \sum_{\ve = \pm} \cV_{i}^{\ve}( u ; \mu \lvert \delta , \kappa) \exp({\ve \delta \frac{\partial}{\partial u_i}})\end{equation}
where 
\begin{equation}
\cV_{i}^{\ve}(u ; \mu \lvert \delta , \kappa) =\frac{\prod_{0 \leq s \leq 7} [ \ve u_i + \mu_{s} ]}{[ \ve 2 u_{i}] [ \ve 2 u_{i} + \delta ]} \prod_{j \neq i} \frac{[ \ve u_{i} \pm u_{j} + \kappa] }{[ \ve u_{i} \pm u_{j} ]}
\end{equation}
for all $i\in\{1,2,\ldots,m\}$ and $\ve \in \{\pm\}$, and $\cV^{0}(u ; \mu \lvert \delta , \kappa)$ given by
\begin{equation}
\cV^{0}(u ; \mu \lvert \delta , \kappa) = \frac{1}{2} \sum_{0 \leq r \leq 3} \cV^{0}_{r}( u ; \mu \lvert \delta , \kappa),
\end{equation}
where
\begin{multline}
 \cV^{0}_{r}( u ; \mu \lvert \delta , \kappa) =  \frac{ e(-\half \eta_{\omega_r}( c_{m+1}(\mu \lvert 2\delta , \kappa) +2 \omega_{r})) \prod_{0 \leq s \leq 7}[ \half ( \omega_r - \delta ) + \mu_s]}{ [ \kappa ] [ \kappa - \delta ]} \\ 
\cdot \prod_{1 \leq j \leq m} \frac{[\half( \omega_r - \delta ) \pm u_j + \kappa ]}{[ \half ( \omega_r - \delta ) \pm u_j]}.
\end{multline}
with 
\begin{equation}
c_{m}(\mu \lvert \delta , \kappa ) = \sum_{0 \leq s \leq 7} \mu_s + 2(m-1) \kappa - 2 \delta \quad (m\in\Z_{\geq 0}).
\end{equation}
(We note that the form above holds for any odd function $[u]$ that satisfies the three-term relation \cite{KNS09}.)
\subsubsection{van Diejen's operator from additive to multiplicative notation}

Using our specialization of the odd function $[u]$ \eqref{eq_odd_function_specialization} yields that the van Diejen operator can be factorized as 
\begin{multline}
\label{eq_vd_op_first_step}
\cE_{u}^{(\mu \lvert \delta , \kappa)} = \frac{e(\half(2 \kappa + \delta))}{e(\half(2 m \kappa + \sum_{0 \leq s \leq 7} \mu_{s})) } \Bigl( {\mathcal{A}}^{0}(u ; \mu \lvert \delta , \kappa) 
 + \sum_{1 \leq i \leq m} \sum_{\ve = \pm} {\mathcal{A}_{i}}^{\ve}( u ; \mu \lvert \delta , \kappa) \e^{\ve \delta \frac{\partial}{\partial u_i}} \Bigr)
\end{multline}
with coefficients 
\begin{multline}
\mathcal{A}^{\ve}_{i}(u ; \mu \lvert \delta , \kappa) = e(-2 \ve u_{i})\frac{\prod_{0 \leq s \leq 7} \theta( e(\mu_{s}) e(\ve u_{i});p) }{\theta( e(\ve 2 u_{i}) , e(\delta) e(\ve 2 u_{i});p)}  \prod_{j\neq i}  \frac{\theta(e( \kappa) e(\ve u_{i}) e( \pm u_{j}) ;p)}{ \theta( e(\ve u_{i}) e(\pm u_{j});p)},
\end{multline}
for all $i\in\{1,2,\ldots,m\}$ and $\ve\in\{\pm\}$, and $\mathcal{A}^{0}(u ;\mu \lvert \delta, \kappa) = \frac{1}{2} \sum_{ 0 \leq r \leq 3}\mathcal{A}^{0}_{r}(u ;\mu \lvert \delta, \kappa)$ where
\begin{multline}
\mathcal{A}_{r}^{0}(u ; \mu \lvert \delta , \kappa ) = e( - \half \eta_{\omega_{r}}( c_{m+1}(\mu\lvert 2 \delta,\kappa)  + 2\omega_{r}) + (\delta -2  \omega_{r} )) \\ 
\cdot \frac{\prod_{0 \leq s \leq 7} \theta( e(\half \omega_{r}) e(-\half \delta) \e(\mu_{s});p)}{\theta( e( \kappa) , e(-\delta) e( \kappa) ;p)} \prod_{ 1 \leq j \leq m} \frac{\theta( e(\half \omega_{r}) e(-\half \delta) e(\kappa) e(\pm u_{j}) ; p ) }{\theta( e(\half \omega_{r}) e(-\half \delta) e(\pm u_{j}) ; p )}
\end{multline}
for all $r\in\{0,1,2,3\}$. The parameters and variables are then related by
\begin{equation}
\label{eq_relations}
\begin{split}
e(\mu_{0}) = p^{-1} a_{0} , \quad e(\mu_{1})= p^{-1} a_{1},\quad e(\mu_{s})= a_{s} \ (s\in\{2,3,\ldots,7\}) \\
e(\delta) = q , \quad e(\kappa) = t , \quad e(u_{i} ) = x_{i} \ (i\in\{1,2,\ldots,m\})
\end{split}
\end{equation}
and we have $c_{r} = e(\half \omega_r)$ for all $r\in\{0,1,2,3\}$. Here, we are using the result of Lemma~\ref{lemma_parameter_p_shift}, with $\abs{K}=2$, to remove the factors $x_{i}^{-\ve 2}$ in front of the coefficients by choosing $K=\{0,1\}$. Different choices will only change the operator by an overall multiplicative constant. Using this relation, we find that 
$
\mathcal{A}^{\ve}_{i}(u ; \mu \lvert \delta , \kappa) = a_{0} a_{1}p^{-2} A^{\ve}_{i}(x; a\lvert q ,t)
$
and 
$
\mathcal{A}_{r}^{0}(u ; \mu \lvert \delta , \kappa ) =  a_{0} a_{1} p^{-2} A^{0}_{r}(x ; a \lvert q , t)
$
by straightforward calculations, which yields that 
\begin{equation}
\label{eq_vd_op_second_step}
\mathcal{E}_{u}^{(\mu\lvert \delta , \kappa)} =  a_{0}^\half a_{1}^{\half}a_{2}^{-\half} \cdots a_{7}^{-\half}p^{-1} q^{\half} t^{-m+1}  \cD_{x}(a \lvert q, t).
\end{equation}

\subsubsection{Kernel function identities.} 

From the result in Theorem~2.3 of \cite{KNS09}, we have the kernel function identities
\begin{equation}
\label{eq_KNS_KFI_Cauchy}
\cE_{u}( \mu \lvert \delta , \kappa) \Phi_{BC}(u , v \lvert \delta , \kappa) = \cE_{v}( \nu \lvert \delta , \kappa) \Phi_{BC}(u , v \lvert \delta , \kappa),
\end{equation}
where $\nu = (\nu_{0},\ldots,\nu_{7})$ for $\nu_{s} = \half ( \delta +\kappa) - \mu_{s}$, under the balancing condition $2(m-n-2)\kappa -2 \delta + \sum_{0 \leq s \leq 7} \mu_{s} = 0$
and 
\begin{equation}
\label{eq_KNS_KFI_dual}
[\kappa] \cE_{u}^{(\mu\lvert \delta , \kappa)} \Psi_{BC}( u ,v ) + [\delta] \cE_{v}^{(\mu\lvert \kappa, \delta)} \Psi_{BC}(u,v) = 0
\end{equation}
under the balancing condition $2(m-1) \kappa + 2(n-1)\delta + \sum_{ 0 \leq s \leq 7} \mu_{s} = 0$, for the functions
\begin{multline}
\Phi_{BC}(u,v\lvert q, t) = \prod_{\substack{ 1 \leq i \leq m \\ 1 \leq k \leq n}} \prod_{\ve=\pm}  G_{+}( \ve u_{i} + v_{k} + \half(\delta - \kappa)\lvert \delta)G_{-}( \ve u_{i} - v_{k} + \half(\delta - \kappa)\lvert \delta)
\end{multline}
and 
\begin{equation}
\Psi_{BC}(u,v) = \prod_{\substack{1 \leq i \leq m \\ 1 \leq k \leq n}} [u_{i} + v_{k}] [u_{i} - v_{k}].
\end{equation}
Note that the balancing conditions above are the same as our balancing condition in \eqref{eq_balancing_condition_KFI} and \eqref{eq_dual_balancing_condition} when using the parametrization in \eqref{eq_relations}. Let us start with \eqref{eq_KNS_KFI_Cauchy}: Using \eqref{eq_vd_op_first_step}, we have that 
\begin{equation}
\cE_{v}^{(\nu\lvert q , t)} = \frac{e(\half \sum_{0 \leq s \leq 7} \mu_{s}) }{e(\half(2(m+1)\kappa + 3 \delta))}  \Bigl( {\mathcal{A}}^{0}(v ; \nu \lvert \delta , \kappa) 
 + \sum_{1 \leq k \leq n} \sum_{\ve = \pm} {\mathcal{A}_{k}}^{\pm}( v ; \nu \lvert \delta , \kappa) \e^{\ve \delta \frac{\partial}{\partial v_k}} \Bigr)
\end{equation}
and use the parametrization in \eqref{eq_relations}, $e(v_{k}) = p^{1/2} y_{k}$ for all $k\in\{1,2,\ldots,n\}$, and
\begin{equation}
e(\nu_{0}) = p^{\half} b_{0}, \quad e(\nu_{1})= p^{\half} b_{1} , \quad e(\nu_{s}) = p^{-\half} b_{s} \ (s\in\{2,3,\ldots,7\}) 
\end{equation}
 to obtain that 
\begin{multline}
\cE_{v}^{(\nu\lvert q , t)} = \frac{ p q^\half t}{ t^{n} (b_{0}\cdots b_{7})^{\half}} \frac{pq^{2}}{b_{0} b_{1}} \\
\cdot \Bigl( A^{0}(p^\half y ; p^{-\half} b \lvert q , t) + \sum_{1 \leq k \leq n} \sum_{\ve = \pm} q^{-2} (p^{\half} y_{k})^{-4 \ve} A_{k}^{\ve}(p^\half y, p^{-\half} b \lvert q , t) T_{q,y_{k}}^{\ve} \Bigr)
\end{multline}
by straightforward calculations. Let $g(y)$ be any function satisfying 
\begin{equation}
T_{q,y_{k}} g(y) = q^{-2} p^{-2} y_{k}^{-4} g(y) \ \Leftrightarrow \ T_{q,y_{k}}^{-1} g(y) = q^{-2} p^{2} y_{k}^{2} g(y)
\end{equation}
for all $k\in\{1,2,\ldots,n\}$, then 
\begin{equation}
\cE_{v}^{(\nu\lvert q , t)} = \frac{ p q^\half t}{ t^{n} (b_{0}\cdots b_{7})^{\half}} \frac{pq^{2}}{b_{0} b_{1}} g(y)^{-1} \circ \cD_{p^{\half} y}(p^{-\half} b \lvert q, t) \circ g(y).
\end{equation}
Using the result of Lemma~\ref{lemma_half-period_shift} allows us to express the kernel function identity as
\begin{equation}
\frac{q^{\half} t a_{0} a_{1} }{(a_{0}\cdots a_{7})^{\half} p t^m} \cD_{x}(a \lvert q,t) F(y) \Phi_{BC} = \frac{q^{\half} t }{ b_{0} b_{1} p t^{n}} \cD_{y}(b\lvert q , t) F(y) \Phi_{BC}
\end{equation}
where $F(y) = g(y) f(y;p^{-\half} b \lvert q , t)$, with $f(x;a\lvert q ,t)$ as in Lemma~\ref{lemma_half-period_shift}, satisfies the $q$-difference equation $T_{q,y_{k}}F(y) = (b_{0} \cdots b_{7})^{-\half} p q t^{-n+1} F(y)= t^{-m}F(y)$ under the balancing condition. Using the balancing condition in the identity above, it is straightforward to check that the factors in front of the operators cancel, and we obtain that
 \begin{equation}
\label{eq_KFI_1}
 \cD_{x}(a \lvert q,t) F(y) \Phi_{BC} = \cD_{y}(b \lvert q, t) F(y) \Phi_{BC}
\end{equation}
Finally, we obtain that 
\begin{equation}
\Phi_{BC}(u, v \lvert \delta , \kappa) 
= \Phi(x,y\lvert q , t) \prod_{1 \leq k \leq n}e(- m \frac{\kappa}{\delta }v_{k}).
\end{equation}
Using that $\exp( \delta \frac{\partial}{\partial v_{k}}) \prod_{1 \leq k \leq n}e(- m \frac{\kappa}{\delta }v_{k}) = t^{m}  \prod_{1 \leq k \leq n}e(- m \frac{\kappa}{\delta }v_{k})$, we find that the product of this factor and $F(y)$ in \eqref{eq_KFI_1} yields a $q$-periodic function, \emph{i.e.} 
\begin{equation}
F(y) \Phi_{BC} = [\mbox{{quasi-const.}}] \cdot \Phi( x , y \lvert q, t).
\end{equation}
The quasi-constant can be ignored in the kernel function identity and we obtain the kernel function identity in Lemma~\ref{lemma_Cauchy_KFI}.

The kernel function of dual-Cauchy type is obtained in a straightforward way from $\Psi_{\text{BC}}$ using our specialization of $[u]$ in \eqref{eq_odd_function_specialization}, Eqs. \eqref{eq_relations} and  \eqref{eq_vd_op_second_step}, and setting $e(v_{k})= y_{k}$ for all $k\in\{1,2,\ldots,n\}$. The kernel function identity also follows from straightforward calculations using this specialization.

\subsection{Relation to \cite{Rui09a}}

As we have previously said, the kernel function, and kernel function identity, in Lemma~\ref{lemma_Cauchy_KFI} for $m=n$ was constructed by Ruijsenaars. In this Section, the exact relations between our notation and those in \cite{Rui09a} is presented.

The periods and shift parameters are related by
\begin{equation}
\omega_{1} =  \text{``} \frac{\pi}{r} \text{''}, \quad q = e( \imag \frac{ a_{-\ve}}{\omega_{1}}), \quad p = e(\imag \frac{a_{\ve}}{\omega_{1}} ) \quad (\ve\in\{\pm\})
\end{equation}
where ``$r$''$\in \R_{>0}$ corresponds to the period along the real line and $a_{\ve} \in\R_{>0}$ corresponds to the (quasi-)periods in the imaginary direction. Here, the choice of $\ve\in\{\pm\}$ is arbitrary and reflects the $p\leftrightarrow q$ duality of the model. We proceed by fixing $\ve \in\{\pm\}$ in this part. 

Ruijsenaars' `right-hand-side' function $R_{\ve}(u) = R(r , a_{\ve} ; u)$ is defined by the product
\begin{equation}
R_{\ve}(u) = \prod_{\ell\in\Z_{\geq 0}} (1 - e( \imag (2\ell + 1) \frac{a_{\ve}}{2\omega_{1}}) e(\frac{u}{\omega_{1}}))(1 - e( \imag (2\ell + 1) \frac{a_{\ve}}{2\omega_{1}}) e(-\frac{u}{\omega_{1}}))
\end{equation}
and is related to the multiplicative theta function by $R_{\ve}(u) = \theta( p^{\half} e(\frac{u}{\omega_{1}});p)$.
The elliptic Gamma function $G(r , a_{+} ,a_{-} ; u)$ in \cite{Rui09a} is defined by the infinite product
\begin{equation}
G(r,a_{+}, a_{-} ;u) = \prod_{i,j \in \Z_{\geq 0}} \frac{1- e(\imag (2 i + 1) \frac{a_{+}}{2\omega_{1}})e(\imag (2 j + 1) \frac{a_{-}}{\omega_{1}}) e(- \frac{u}{\omega_{1}})  }{1- e(\imag (2 i + 1) \frac{a_{+}}{2\omega_{1}})e(\imag (2 j + 1) \frac{a_{-}}{\omega_{1}}) e( \frac{u}{\omega_{1}}) }
\end{equation}
which is related to the multiplicative Gamma function $\Gamma(x;p,q)$ by
\begin{equation}
G(r,a_{+}, a_{-} ;u)  = \Gamma( p^\half q^\half e(\frac{u}{\omega_{1}}); p , q).
\end{equation}
Setting $x= e(\frac{u}{\omega_{1}})$, we find the relations 
\begin{equation}
R_{\ve}(u) = \theta(p^\half x ; p ) , \quad G(r,a_{+}, a_{-} ;u) = \Gamma( p^\half q^\half x ; p , q)
\end{equation}
between our notation and that of \cite{Rui09a} for the multiplicative theta and Gamma functions. We will suppress the dependence on the scaling parameters from now on and just write $G(x)$ for Ruijsenaars' elliptic Gamma function.

The van Diejen operator is denoted by ``$\cA_{\ve}(h , \mu ; u)$'' $(\ve\in\{\pm\})$ in \cite{Rui09a} and is defined as
\begin{equation}
\cA_{\ve}(h , \mu ; u) =\cV(h,\mu ; u) + \sum_{1 \leq i \leq m}( \cV_{i}(h,\mu; u) \e^{-\imag a_{-\ve} \frac{\partial}{\partial u_{i}}} + \cV_{i}(h , \mu ; -u) \e^{\imag a_{-\ve} \frac{\partial}{\partial u_{i}}})
\end{equation}
with coefficients given by 
\begin{equation}
\cV_{i}(h,\mu;u) = \frac{\prod_{0 \leq s \leq 7} R_{\ve}(u_{i} - h_{s} - \imag a_{-\ve} / 2)}{R_{\ve}( 2 u_{i} + \imag a_{\ve}/2) R_{\ve}( 2 u_{i} - \imag a_{-\ve} + \imag a_{\ve}/2)}\prod_{j\neq i} \frac{R_{\ve}(u_{i} \pm u_{j} - \mu + \imag a_{\ve}/2)}{R_{\ve}( u_{i} \pm u_{j} + \imag a_{\ve}/2)},
\end{equation}
where we use the notation $R_{\ve}(u\pm v+\alpha) = R_{\ve}(u+ v+\alpha) R_{\ve}(u - v+\alpha) $, and
\begin{equation}
\cV(h,\mu;u) = \frac{\sum_{0 \leq r \leq 3} p_{r}(h) \bigl[ (\prod_{1 \leq j \leq m} \mathcal{E}_{r}(\mu;u_{j})) - \mathcal{E}_{r}(\mu ; \omega_{r} /2)^{m}\bigr]}{2 R_{\ve}(\mu - \imag a_{\ve} /2) R_{\ve}( \mu - \imag a_{-\ve} - \imag a_{\ve}/2)}
\end{equation}
where
\begin{align*}
p_{0}(h) = \prod_{ 0 \leq s \leq 7} R_{\ve}(h_{s}), \quad  & p_{2}(h) = p \prod_{ 0 \leq s \leq 7} e(-\frac{h_{s}}{2\omega_{1}}) R_{\ve}( h_{s} - \half \omega_{2}), \\
 p_{1}(h) = \prod_{ 0 \leq s \leq 7} R_{\ve}(h_{s} - \half \omega_1), \quad & p_{3}(h) = p \prod_{ 0 \leq s \leq 7} e(\frac{h_{s}}{2\omega_{1}}) R_{\ve}(h_{s} - \half \omega_{3}),
\end{align*}
and 
\begin{equation}
\mathcal{E}_{r}(\mu;u_{j}) = \frac{R_{\ve}( \pm u_{j} + \mu - \half \imag ( a_{\ve}+a_{-\ve}) - \half \omega_{r})}{R_{\ve}( \pm u_{j}  - \half \imag ( a_{\ve} + a_{-\ve}) - \half \omega_{r})}
\end{equation}
with $\omega_{0}=0$, $\omega_{2} = \imag a_{\ve}$ and $\omega_{3}=-\omega_{1} - \omega_{2}$.

We note that the operator $\exp( \imag a_{-\ve} \frac{\partial}{\partial u })$ acts as the $T_{q,x}$-shifts operator on the multiplicative variables $x = e( u / \omega_1)$, \emph{i.e.}
\begin{equation}
\e^{\ve' \imag a_{-\ve} \frac{\partial}{\partial u }} e(u/\omega_{1}) = e( (\ve' \imag a_{-\ve})/\omega_1) e(u / \omega_1) = q^{\ve' } e(u/\omega_1) \quad (\ve' \in\{\pm\}),
\end{equation}
and expect that $\cV_{i}(h , \mu ; \pm u)$ should correspond to $A_{i}^{\mp}(x ; a \lvert q , t)$ ($i\in\{1,2,\ldots,m\}$) in our notation. It follows from straightforward calculations, and the reflection property, that 
\begin{multline}
\cV_{i}(h , \mu ;  u) = \frac{\prod_{ 0 \leq s \leq 7} \theta( p^{\half} q^{-\half} e(-h_{s}/\omega_{1}) x_{i}; p )}{\theta( p x_{i}^{2} ;p ) \theta( p q^{-1} x^{2}_{i};p)} \prod_{j\neq i} \frac{\theta( p e(-\mu/ \omega_1) x_{i} x_{j}^{\pm};p) }{\theta(p x_{i} x_{j}^{\pm};p)} \\
= \frac{\prod_{ 0 \leq s \leq 7} \theta( p^{\half} q^{\half} e(h_{s}/\omega_{1}) x_{i}^{-1}; p )}{\theta( x_{i}^{-2} ;p ) \theta( q x^{-2}_{i};p)} \prod_{j\neq i} \frac{\theta(  e(\mu/ \omega_1) x_{i}^{-1} x_{j}^{\pm};p) }{\theta( x_{i}^{-1} x_{j}^{\pm};p)}
\end{multline}
which allows us to identify the correspondence between parameters
\begin{equation}
a_{s} = p^{\half} q^{\half} e( h_{s} / \omega_1) \ (s\in\{0,1,\ldots,7\}), \quad t = e(\mu/ \omega_1).
\end{equation}
We then obtain that $\cV_{i}(h, \mu;u) = A_{i}^{-}(x; a \lvert q , t)$ for all $i\in\{1,2,\ldots,m\}$. (Note that $e(h_s/\omega_1)$ are the same  as the parameters $\mu_{s}$ used in Section~\ref{sec_E8_symmetry} when considering the $W(D_{8})$ reflections.) Using the correspondence, we obtain that 
\begin{equation}
\mathcal{E}_{r}(\mu;u_{j}) = t^{\delta_{r,2}-\delta_{r,3}} \frac{\theta( c_{r} q^{-\half} t x_{j}^{\pm};p)}{\theta( c_{r} q^{-\half} x_{j}^{\pm};p)} \quad (r\in\{0,1,2,3\})
\end{equation}
by using the quasi-periodicity of the multiplicative theta function, and $p_{r}(h) = L_{r}^{(0)}(a\lvert q, - ) \prod_{ 0 \leq s \leq 7} \theta(c_r q^{-\half} a_{s};p)$, for all $r\in\{0,1,2,3\}$. It is then clear that 
\begin{equation}
\cV(h , \mu ; u) = A^{0}(x;a \lvert q, t ) - C
\end{equation}
where the constant is obtained by straightforward calculations to equal
\begin{equation}
\label{eq_Ruijs_constant}
C = \frac{1}{2}\sum_{ 0 \leq r \leq 3} L_{r}^{(2 m)}(a\lvert q, t) \frac{\prod_{0 \leq s \leq 7}\theta(c_r q^{-\half} a_{s};p) }{\theta( t , q^{-1} t ; p )} \Bigl( \frac{\theta(q^{-\half} t ;p)}{\theta( q^{-\half} ;p)}\Bigr)^{2m}.
\end{equation}

Ruijsenaars' kernel function 
\begin{equation}
\mathcal{S}(h; u , v ) = \prod_{1 \leq i,k \leq m} G(\pm u_{i} \pm v_{k} - \half \imag (a_{+} + a_{-}) - \frac{1}{4} \sum_{s=0}^{7} h_{s})
\end{equation}
satisfies the kernel function identity
\begin{equation}
\bigl(\cA_{\ve}(h;\mu; u) - \cA_{\ve}( - J_{R} h , \mu ; v) - \sigma_{\ve}(h)\bigr)\mathcal{S}(h;u,v) = 0
\end{equation}
for $(-J_{R} h)_{s} = -h_{s} + \frac{1}{4} \sum_{0 \leq s \leq 7} h_s$ ($s\in\{0,1,\ldots,7\}$), some constant $\sigma_{\ve}(h)\in\C$, if the parameters satisfy $2 \mu = 2 \imag (a_{+} + a_{-}) + \sum_{0 \leq s \leq 7} h_{s}$; see Proposition~4.1 of \cite{Rui09a}. In our convention, we have that the balancing conditions coincide since
\begin{equation}
2 \mu = 2 \imag (a_{+} + a_{-}) + \sum_{0\leq s \leq 7} h_{s} \Leftrightarrow t^{2} = p^{-2} q^{-2} a_{0} \cdots a_{7}
\end{equation}
and that $\exp( (-J_{r} h)_{s}) = p^{\half} q^{\half} t^{\half} a_{s}^{-1} = b_{s}$ for all $s\in\{0,1,\ldots,7\}$.
It is clear that $G( u - \half \imag ( a_{+} + a_{-})) = \Gamma( e(u/\omega_1) ; p ,q )$, and using 
\begin{equation}
e( -\frac{1}{4} \sum_{0 \leq s \leq 7} h_{s}/\omega_1) = e( ( \imag a_{+} + \imag a_{-} - \mu)/2\omega_1) = p^{\half} q^{\half} t^{-\half}
 \end{equation}
under the balancing condition, gives that 
\begin{equation}
\mathcal{S}(h ; u , v ) = \prod_{1 \leq i,k\leq m} \Gamma( p^{\half} q^{\half} t^{-\half} x_{i}^{\pm} e( v_{k}/\omega_1); p ,q ).
\end{equation}
Setting $y_{k} = e(v_{k}/\omega_1)$ for all $k\in\{1,2,\ldots,m\}$ yields the kernel function identity in Lemma~\ref{lemma_Cauchy_KFI} for the $n=m$ case.

To summarize: The van Diejen operator in \cite{Rui09a} (and \cite{Rui15}) equals the van Diejen operator $\cD_{x}(a\lvert q , t)$ up to an additive constant $C$ \eqref{eq_Ruijs_constant} for
\begin{equation}
a_{s}= p^\half q^\half \e^{2 \imag r h_{s}} \ (s\in\{0,1,\ldots,7\}), \ q= \e^{-2r a_{-}} , \ p=\e^{-2r a_{+}}, \ t= \e^{2\imag r \mu},
\end{equation}
and setting $x_{i} = \exp( 2 \imag r u_{i})$ for all $i\in\{1,2,\ldots,m\}$. The kernel function identity are then obtain by inserting these parameters.

\newcommand{\etalchar}[1]{$^{#1}$}

\end{document}